\documentclass[reqno]{article}

\usepackage{amsopn}

\usepackage[utf8]{inputenc}
\usepackage{mathtools}
\usepackage{amssymb} 
\usepackage{makeidx}
\usepackage[english]{babel}
\usepackage{graphicx}
\usepackage{amsfonts,amsmath,amssymb}
\usepackage{oldgerm}
\usepackage{bbold}

\usepackage{mathrsfs}
\usepackage[active]{srcltx}
\usepackage{verbatim}
\usepackage{enumitem} 
\usepackage[toc,page]{appendix}
\usepackage{aliascnt,bbm}
\usepackage{array}
\usepackage[textwidth=4cm,textsize=footnotesize]{todonotes}
\usepackage{xargs}
\usepackage{cellspace}
\usepackage{algorithm}
\usepackage{algorithmic}

\usepackage[Symbolsmallscale]{upgreek}
 \usepackage{ifpdf}
 \usepackage{amsthm}
 \usepackage{hyperref}
 \usepackage{authblk}
\usepackage{accents}
\usepackage{dsfont}
\usepackage{aliascnt}
\usepackage{cleveref}
\usepackage{lmodern}
\usepackage{rotating}
\makeatletter

\crefname{theorem}{theorem}{Theorems}
\Crefname{theorem}{Theorem}{Theorems}

\newaliascnt{lemma}{theorem}

\aliascntresetthe{lemma}
\crefname{lemma}{lemma}{lemmas}
\Crefname{lemma}{Lemma}{Lemmas}

\newaliascnt{corollary}{theorem}

\aliascntresetthe{corollary}
\crefname{corollary}{corollary}{corollaries}
\Crefname{corollary}{Corollary}{Corollaries}

\newaliascnt{proposition}{theorem}
\newtheorem{proposition}[proposition]{Proposition}
\aliascntresetthe{proposition}
\crefname{proposition}{proposition}{propositions}
\Crefname{proposition}{Proposition}{Propositions}

 \crefname{example}{example}{examples}
 \Crefname{Example}{Example}{Examples}

\crefname{figure}{figure}{figures}
\Crefname{Figure}{Figure}{Figures}

\Crefname{assumption}{\textbf{H}\hspace{-3pt}}{\textbf{H}\hspace{-3pt}}
\crefname{assumption}{\textbf{H}}{\textbf{H}}

\Crefname{assumptionG}{\textbf{G}\hspace{-3pt}}{\textbf{G}\hspace{-3pt}}
\crefname{assumptionG}{\textbf{G}}{\textbf{G}}

\Crefname{assumptionA}{\textbf{A}\hspace{-3pt}}{\textbf{A}\hspace{-3pt}}
\crefname{assumptionA}{\textbf{A}}{\textbf{A}}

\makeatother

\usepackage{autonum}

\def\Ly{\Lt_y}
\def\xGT{x^0}
\def\thetaMSE{\theta_{\dagger}}
\def\EQ{\mathrm{EQ}}
\def\SUG{\mathrm{SUG}}

\def\czero{c_0}
\def\powerdelta{p}
\def\Frob{\mathrm{Frob}}

\def\dimf{d_{\mathtt{f}}}
\def\dimp{d_{\mathtt{p}}}
\def\dimm{d_{\mathtt{m}}}

\def\MAPname{$\mathrm{MAP}$}
\def\SUGARname{$\mathrm{SUGAR}$}
\def\MSEname{$\mathrm{MSE}$}
\def\SNRname{$\mathrm{SNR}$}
\def\PSNRname{$\mathrm{PSNR}$}
\def\SUGARname{$\mathrm{SUGAR}$}
\def\tvname{$\mathrm{TV}$}
\def\tgvname{$\mathrm{TGV}$}
\def\ttheta{\tilde{\theta}}

\def\tg{\tilde{g}}
\def\thetaavg{\bar{\theta}}

\def\thetaEB{\bar{\theta}_N}

\def\thetaStar{\theta_{\star}}

\def\MAP{\mathrm{MAP}}
\def\MMSE{\mathrm{MMSE}}
\def\tx{\tilde{x}}
\def\dim{d}
\def\dimy{d_y}
\def\dimtheta{d_{\Theta}}

\newcommand{\cball}[2]{\overline{\operatorname{B}}(#1,#2)}

\newcommand{\prox}{\operatorname{prox}}

\newcommand{\bRker}{\bar{\mathrm{R}}}

\newcommand{\Rker}{\mathrm{R}}

\newcommand{\Lt}{\mathtt{L}}

\newcommand{\btheta}{\theta}

\def\x{{x}}

\def\y{{y}}

\def\g{{g}}
\def\f{{f_y}}

\newcommandx{\norm}[2][1=]{\ifthenelse{\equal{#1}{}}{\left\Vert #2 \right\Vert}{\left\Vert #2 \right\Vert^{#1}}}
\newcommandx{\normLigne}[2][1=]{\ifthenelse{\equal{#1}{}}{\Vert #2 \Vert}{\Vert #2\Vert^{#1}}}

\def\msa{\mathsf{A}}

\def\rset{\mathbb{R}}

\def\cset{\mathbb{C}}

\def\nset{\mathbb{N}}
\def\nsets{\mathbb{N}^*}

\def\rmd{\mathrm{d}}
\def\rmZ{\mathrm{Z}}
\def\rmB{\mathrm{B}}

\def\rme{\mathrm{e}}

\def\rmw{\mathrm{w}}
\def\rmA{\mathrm{A}}

\newcommandx{\functionspace}[2][1=+]{\mathbb{F}_{#1}(#2)}
\newcommand{\argmax}{\operatorname*{arg\,max}}
\newcommand{\argmin}{\operatorname*{arg\,min}}

\newcommandx{\VarDeux}[3][3=]{\operatorname{Var}^{#3}_{#1}\left\{#2 \right\}}

\newcommand{\LeftEqNo}{\let\veqno\@@leqno}

\newcommand{\N}{\ensuremath{\mathbb{N}}}

\newcommand{\abs}[1]{\left\vert #1 \right\vert}
\newcommand{\absLigne}[1]{\vert #1 \vert}

\newcommandx{\Vnorm}[2][1=V]{\| #2 \|_{#1}}
\newcommandx{\VnormEq}[2][1=V]{\left\| #2 \right\|_{#1}}

\newcommand{\parentheseDeux}[1]{\left[ #1 \right]}

\newcommand{\defEns}[1]{\left\lbrace #1 \right\rbrace }
\newcommand{\defEnsLigne}[1]{\lbrace #1 \rbrace }

\newcommandx\probaMarkovTilde[2][2=]
{\ifthenelse{\equal{#2}{}}{{\widetilde{\mathbb{P}}_{#1}}}{\widetilde{\mathbb{P}}_{#1}\left[ #2\right]}}

\def\ie{\textit{i.e.}}

\def\eqsp{\;}
\newcommand{\coint}[1]{\left[#1\right)}
\newcommand{\ocint}[1]{\left(#1\right]}
\newcommand{\ooint}[1]{\left(#1\right)}
\newcommand{\ccint}[1]{\left[#1\right]}

\newcommandx{\weight}[2][2=n]{\omega_{#1,#2}^N}

\RequirePackage{xcolor}[2.11]
\definecolor{header1}{cmyk}{.9,.5,0,.35}
\definecolor{blue1}{cmyk}{.9,.7,0,0}
\definecolor{blue2}{cmyk}{.93,.95,.2,.07}
\definecolor{maroon}{cmyk}{.4,1,.3,.2}
\definecolor{gold1}{cmyk}{.2,.2,1,.1}
\definecolor{gray}{cmyk}{0,0,0,.5}
\definecolor{green1}{cmyk}{1,0,1,0}
\definecolor{proofcolor}{cmyk}{1,0,1,0}
\definecolor{red1}{cmyk}{0,1,.8,0}
\definecolor{orange1}{cmyk}{0,.55,1,0}
\definecolor{strip}{cmyk}{.6,.1,.1,.1}

\newcommandx\sequence[3][2=,3=]
{\ifthenelse{\equal{#3}{}}{\ensuremath{\{ #1_{#2}\}}}{\ensuremath{\{ #1_{#2}, \eqsp #2 \in #3 \}}}}
\newcommandx\sequenceD[3][2=,3=]
{\ifthenelse{\equal{#3}{}}{\ensuremath{\{ #1_{#2}\}}}{\ensuremath{( #1)_{ #2 \in #3} }}}

\newcommandx{\sequencen}[2][2=n\in\N]{\ensuremath{\{ #1_n, \eqsp #2 \}}}
\newcommandx\sequenceDouble[4][3=,4=]
{\ifthenelse{\equal{#3}{}}{\ensuremath{\{ (#1_{#3},#2_{#3}) \}}}{\ensuremath{\{  (#1_{#3},#2_{#3}), \eqsp #3 \in #4 \}}}}
\newcommandx{\sequencenDouble}[3][3=n\in\N]{\ensuremath{\{ (#1_{n},#2_{n}), \eqsp #3 \}}}

\def\eg{e.g.}

\newcommand{\opnorm}[1]{{\left\vert\kern-0.25ex\left\vert\kern-0.25ex\left\vert #1 
    \right\vert\kern-0.25ex\right\vert\kern-0.25ex\right\vert}}

\def\Id{\operatorname{Id}}

\newcommandx{\CPE}[3][1=]{{\mathbb E}_{#1}\left[#2 \left \vert #3 \right. \right]} 
\newcommandx{\CPVar}[3][1=]{\mathrm{Var}^{#3}_{#1}\left\{ #2 \right\}}
\newcommand{\CPP}[3][]
{\ifthenelse{\equal{#1}{}}{{\mathbb P}\left(\left. #2 \, \right| #3 \right)}{{\mathbb P}_{#1}\left(\left. #2 \, \right | #3 \right)}}

\newcommandx{\osc}[2][1=]{\mathrm{osc}_{#1}(#2)}

\def\Id{\operatorname{Id}}

\def\transpose{\operatorname{T}}

\def\y{y}

\def\bgamma{\bar{\gamma}}

\def\bX{\bar{X}}

\def\tx{\tilde{x}}

\def\bX{\bar{X}}

\newcommand{\ensembleLigne}[2]{\{#1\,:\eqsp #2\}}

\newcommand\coupling[2]{\Gamma(\mu,\nu)}

\renewcommand{\geq}{\geqslant}
\renewcommand{\leq}{\leqslant}

\def\interior{\mathrm{int}}

\def\vareps{\varepsilon}

\newcommandx{\KL}[2]{\text{KL}\left( #1 | #2 \right)}
\newcommandx{\KLbig}[2]{\text{KL}\left( #1 \middle| #2 \right)}

\ifpdf
\hypersetup{
	pdftitle={Maximum likelihood estimation of regularisation parameters in high-dimensional inverse problems: an Empirical Bayesian approach},
	pdfauthor={Ana Fernandez Vidal,Valentin De Bortoli, Marcelo Pereyra, Alain Durmus}
}
\fi

\title{Maximum likelihood estimation of regularisation parameters in high-dimensional inverse problems: an empirical Bayesian approach\\
Part I: Methodology and Experiments}
\author[1]{Ana F. Vidal \footnote{Email: af69@hw.ac.uk}}
\author[2]{Valentin De Bortoli \footnote{Email: debortoli@cmla.ens-cachan.fr}}
\author[1]{Marcelo Pereyra \footnote{Email: m.pereyra@hw.ac.uk} }
\author[2]{Alain Durmus \footnote{Email: durmus@cmla.ens-cachan.fr \newline \indent \  Part of this work has been presented at the 25th IEEE International Conference on Image Processing (ICIP)  \cite{vidal2018maximum}} }

\affil[1]{\small Maxwell Institute for Mathematical Sciences \& School of Mathematical and Computer Sciences, Heriot-Watt University, Edinburgh, EH14 4AS, United Kingdom.}
\affil[2]{\small CMLA - \'Ecole normale supérieure Paris-Saclay, CNRS, Université Paris-Saclay, 94235 Cachan, France.}

\usepackage[paperwidth=210mm,
paperheight=297mm,
left=3cm,right=3cm,top=2cm,bottom=2.5cm]{geometry}
\begin{document}
\maketitle
\begin{abstract} Many imaging problems require solving an inverse
  problem that is ill-conditioned or ill-posed. Imaging methods
  typically address this difficulty by regularising the estimation
  problem to make it well- posed. This often requires setting the
  value of the so-called regularisation parameters that control the
  amount of regularisation enforced. These parameters are notoriously
  difficult to set a priori, and can have a dramatic impact on the
  recovered estimates. In this work, we propose a general empirical
  Bayesian method for setting regularisation parameters in imaging
  problems that are convex w.r.t. the unknown image. Our method
  calibrates regularisation parameters directly from the observed data
  by maximum marginal likelihood estimation, and can simultaneously
  estimate multiple regularisation parameters. Furthermore, the
  proposed algorithm uses the same basic operators as proximal
  optimisation algorithms, namely gradient and proximal operators, and
  it is therefore straightforward to apply to problems that are
  currently solved by using proximal optimisation techniques.  Our
  methodology is demonstrated with a range of experiments and
  comparisons with alternative approaches from the literature. The
  considered experiments include image denoising, non-blind image
  deconvolution, and hyperspectral unmixing, using synthesis and
  analysis priors involving the $\ell_1$, total-variation,
  total-variation and $\ell_1$, and total-generalised-variation
  pseudo-norms. A detailed theoretical analysis of the proposed method is presented in the companion
  paper \cite{vidal:et:al:2019b}.
	\end{abstract}
	%
	
\section{Introduction}
\label{sec:intro}
Image estimation problems are ubiquitous in science and industry, and a central topic of research in imaging sciences. Canonical examples include, for instance, image denoising, image deblurring, compressive sensing, super-resolution, image inpainting, source separation, fusion, and phase retrieval 
\cite{monga2017handbookconvexopt}. Solving these problems has stimulated significant advances in imaging methods, models, theory, and algorithms \cite{kaipio2006statistical,pereyra2016survey,monga2017handbookconvexopt,chambolle2016introduction}.

Most image estimation problems are ill-conditioned or ill-posed \cite{kaipio2006statistical}, a difficulty that imaging methods typically address by regularising the estimation problem to make it well posed. This can be achieved in different ways. For example, in the variational framework, regularisation is introduced by using penalty functions that favour solutions with desired structural or regularity properties (e.g., smoothness, piecewise-regularity, sparsity, or constraints), see \cite{chambolle2016introduction}. In the Bayesian statistical framework, regularisation arises from the use of informative prior distributions  that also allow promoting solutions with expected structural or regularity properties \cite{kaipio2006statistical}. Moreover, regularisation can be explicitly specified, or learnt from data using modern machine learning techniques. We refer the reader to \cite{arridge2019solving} for an excellent introduction to variational, statistical, and machine learning regularisation approaches.

A main difficulty that arises when using any regularisation technique is deciding how much regularisation is appropriate,                                                                             as different imaging modalities, instrumental setups, scenes, and noise conditions often require using very different amounts of regularisation. The amount of regularisation is usually explicitly controlled by some of the parameters of the model. The difficulty resides in that setting  the value of these regularisation parameters a priori is notoriously difficult, particularly in problems that are ill-posed or ill-conditioned where the regularisation has a dramatic impact on the estimated solutions  (see \cite{molina1999bayesian,EUSIPCO,sugar2014stein} and the illustrative example in \Cref{fig:intro-impact-reg-param}). As a result, there is significant interest in methods for setting regularisation parameters in an automatic, robust, and adaptive way.  

Indeed, the developments of methods to automatically set regularisation parameters is a long-standing research topic in imaging sciences. Some methods such as  generalised cross-validation \cite{golub1979generalized}, the L-curve  \cite{lawson1995solving,hansen1993use,calvetti2000tikhonov}, the discrepancy principle  \cite{morozov2012methods, BENVENUTO20191} and residual whiteness measures  \cite{almeida2013parameter,lanza2013variational} operate by analysing the residual between the observed data and a prediction derived from the observation model. Such methods can perform well in certain imaging problems, but they are mainly limited to cases involving a single scalar regularisation parameter. Alternatively, methods based on Stein's unbiased risk estimator (SURE) have also received a lot of attention in the late \cite{giryes2011projected,eldar2009generalized, sugar2014stein}. These methods seek to select the value of the regularisation parameters by minimising SURE-based surrogates of the estimation mean squared error \cite{eldar2009generalized, pesquet2009sure, giryes2011projected}. SURE methods can perform remarkably well in mildly ill-posed or ill-conditioned problems, but they generally struggle with problems that are more severely ill-conditioned or ill-posed \cite{lucka2018risk}. Some recent works also consider learning regularisation parameters from a training dataset of clean images \cite{van2017learning}, or adopting a bilevel optimisation strategy \cite{calatroni2017bilevel,kunisch2013bilevel}.

Lastly, the Bayesian statistical framework provides two main strategies for addressing unknown regularisation parameters: the hierarchical and the empirical \cite{robert2007bayesian,molina1999bayesian}. So far, imaging methods have mainly adopted the hierarchical  strategy, where the unknown regularisation parameters are incorporated into the model to define an augmented posterior, and subsequently removed from the model by marginalisation or estimated jointly with the unknown image \cite{calvetti2008hypermodels,pereyra2016survey,EUSIPCO}. This is the strategy that is adopted by most Markov chain Monte Carlo and variational Bayesian approaches reported in the literature (see e.g., \cite{pereyra2013estimating,babacan2011variational}).

In this work we propose to adopt an empirical Bayesian approach to estimate the regularisation parameters directly from the observed data in a fully automatic and unsupervised way. We focus on imaging problems that are convex w.r.t. the unknown image and that can be efficiently solved by using modern convex optimisation techniques once the regularisation parameters have been set  \cite{monga2017handbookconvexopt,chambolle2016introduction}. In a manner akin to \cite{molina1999bayesian, atchade2006adaptive}, we set regularisation parameters directly from the observed data by maximum marginal likelihood estimation. A main novelty of our work is that this maximum marginal likelihood estimation problem is efficiently solved by using a stochastic proximal gradient algorithm that is powered by two proximal Markov chain Monte Carlo samplers, thus intimately combining the strengths of modern optimisation and sampling techniques. In addition to being highly efficient and delivering remarkably accurate solutions, the proposed method can be readily implemented with the same tools that are used to construct optimisation algorithms to estimate the unknown image by maximum-a-posteriori estimation, namely proximal and gradient operators.

The remainder of the paper is organised as follows. \Cref{sec:probStatement} defines the class of imaging problems we consider and introduces basic necessary concepts of Bayesian inference. In  \Cref{sec:proposedMethod} we present the proposed empirical Bayesian method to calibrate regularisation parameters, provide detailed implementation guidelines and discuss connections with hierarchical Bayesian approaches. \Cref{sec:experiments} first demonstrates the proposed methodology with a variety of non-blind image deblurring and imaging denoising problems involving scalar-valued regularisation parameters, including an experiment where we also estimate the noise variance. This is then followed by two challenging experiments that require setting vector-valued regularisation parameters, namely sparse hyperspectral image unmixing with the SUnSAL model \cite{iordache2012total}, and image denoising using a Total Generalised Variation regulariser \cite{bredies2010total}, where the parameters have strong dependencies making the estimation problem particularly difficult. We report comparisons with several alternative approaches from the literature, including the discrepancy principle \cite{morozov2012methods}, the SURE-based SUGAR method \cite{sugar2014stein}, and the hierarchical Bayesian method described in \cite{EUSIPCO}. Conclusions and perspectives for future work are finally reported in \Cref{sec:conc}. Additional practical guidelines are postponed to the appendix.  We refer the reader to the companion paper \cite{vidal:et:al:2019b} for a detailed theoretical analysis of our methodology. 
		
\section{Problem Statement}\label{sec:probStatement}
Let $\dim, \dimy, \dimtheta \in \nset$ and let $\Theta \subset \ooint{0,+\infty}^{\dimtheta}$ be a convex compact set. We consider the estimation of an unknown image $ \x \in \rset^{\dim} $ from an observation $ \y \in \cset^{\dimy} $ related to $ \x $ by a statistical model with likelihood function
\begin{equation}
p(\y|\x) \propto \rme^{-\f(\x)} \eqsp , 
\end{equation}
where $\f$ is convex and continuously differentiable with $\Ly$-Lipschitz gradient, \ie \ for any $u,v \in \rset^{\dim}$, 
$\normLigne{ \nabla \f (u) - \nabla \f (v) } \leq \Ly \normLigne{ u - v }$ where $\Ly >0$.
This class includes important observation models, in particular Gaussian linear models of the form $\y=  \rmA \x + \rmw$ where $\rmA \in \mathbb{C}^{\dimy \times \dim}$ and $\rmw$ is a $\dimy$-dimensional Gaussian random variable with zero mean and covariance matrix $\sigma^2 \Id$ with $\sigma >0$. 
We adopt a Bayesian approach and seek to use prior knowledge about $\x$ to regularise the estimation problem and improve results. We consider prior distributions given for any $x \in \rset^{\dim}$ and $\theta \in \Theta$ by 
\begin{equation}
p(\x|\theta) = \rme^{-\theta^{{\transpose}}g(x)} / \rmZ(\theta) \eqsp ,
\label{EQ: prior}
\end{equation}
for some convex and Lipschitz continuous vector of statistics $\g : \rset^{\dim} \rightarrow \rset^{\dimtheta}$ and where we recall that the normalising constant of the prior distribution $p(\x|\theta)$ is given by
\begin{equation}
\rmZ(\theta) = \int_{\rset^{\dim}} \rme^{ -\theta^{\transpose}  \g(\tx)} \rmd \tx \eqsp .
\label{EQ: normalizingConst_prior}
\end{equation}
Note   that $\theta$ controls the amount of regularity enforced. The function $g$ is allowed to be non-differentiable in order to include popular models such as $\g (\x) = \|\rmB\x\|_\dagger$ for some dictionary $\rmB \in \rset^{\dim_1 \times \dim}$ with $d_1 \in \nset$ and norm $\normLigne{\cdot}_\dagger$, as well as  constraints on the solution space such as pixel-positivity.

Although rarely mentioned in the literature, these widely used prior distributions regularise the estimation problem by promoting solutions for which $\g(\x)$ is close to the expected value $\bar{g}_\theta = \int_{\rset^{\dim}}\g(\tx)p(\tx|\theta)~ \rmd\tx$, which depends on $\theta$. Formally, by differentiating \eqref{EQ: normalizingConst_prior} and using Leibniz integral rule \cite{pereyra2014computing} we obtain that for any $\theta \in \Theta$ 
\begin{equation}
\bar{g}_\theta = \int_{\rset^{\dim}}\g(\tx)p(\tx|\theta)~ \rmd \tx = -\nabla_{\theta} \log \rmZ(\theta).
\label{EQ: prior_expect_g}
\end{equation}
Additionally, because the prior distribution $x \mapsto p(x|\theta)$ is log-concave, using \cite[Theorem 1.2]{Bobkov2011} we have that for any $\varepsilon \in \ccint{0,2}$ 
\begin{equation}
\int_{C_{\theta, \vareps}}p(\tx | \theta) \rmd \tx \leq 3 \exp\parentheseDeux{-\varepsilon^2 \dim / 16} \eqsp ,
\label{EQ: prior_prob_deviate_mean}
\end{equation}
with $C_{\theta, \vareps} = \ensembleLigne{\tx \in \rset^{\dim}}{\dim^{-1} \absLigne{\theta^{\transpose}(g(\tx) - \bar{g}_{\theta})} \geq \vareps}$.
This result establishes that the prior distribution $x \mapsto p(x|\theta)$ strongly promotes solutions for which $g(x) \approx -\nabla_{\theta} \log \rmZ(\theta)$ with high probability when $\dim$ is large. 

Once the likelihood and prior $p(\y|\x)$ and $p(\x|\theta)$ are specified, we use Bayes' theorem \cite{robert2007bayesian} to derive the posterior for any $\theta \in \Theta$ and $x \in \rset^{\dim}$
\begin{equation}
p(\x|\y,\theta) = p(\y|\x)p(\x|\theta) / p(\y|\theta) = \left .\exp[-\f(\x)-\theta^{\transpose} g(\x)] \middle / \int_{\rset^{\dim}} \exp[-\f(\tx)-\theta^{\transpose} g(\tx)] \rmd \tx \right . \eqsp .
\label{EQ: posterior}
\end{equation}
This posterior distribution underpins all inferences about the image $x$ given {the} observed data $y$.  In particular, imaging methods typically use the maximum-a-posteriori (MAP) estimator, given for any $\theta \in \Theta$ by
\begin{equation}
\hat{\x}_{\theta, \MAP} \in \underset{\tx \in \rset^{\dim}}{\mathrm{argmin}} \defEnsLigne{\f(\tx) + \theta^{\transpose} g(\tx) } \eqsp .
\label{EQ: mapEstim}
\end{equation}
This Bayesian estimator has several favourable theoretical and computational properties (see \cite{pereyra2019revisiting} for a recent theoretical analysis of this estimator). From a computation viewpoint, since the posterior $x \mapsto p(\x|\y,\theta)$  is log-concave, the computation of $\hat{\x}_{\theta, \MAP}$ is a convex optimisation problem that can usually be efficiently solved using modern optimisation algorithms, see \cite{chambolle2016introduction}. Imaging MAP algorithms typically adopt a proximal splitting approach \cite{combettes2011proximalsplitting} involving the gradient $\nabla \f$ and the proximal operator of $\g$, $\prox_g^{\lambda} : \ \rset^{\dim} \to \rset^{\dimtheta}$, see \cite[Definition 12.23]{bauschke2017convex}. This operator is defined for any $\lambda > 0$ and $x \in \rset^{\dim}$ by 
\begin{equation}
\prox_{g}^{\lambda}(\x)=\underset{\tx \in \rset^{\dim}}{\mathrm{argmin}}~ \defEnsLigne{\g(\tx)+ \norm{\tx-\x}_{2}^{2}/(2\lambda) } \eqsp ,
\label{EQ: proxOperator}
\end{equation}
The smoothness parameter $\lambda >0$ controls the regularity properties of the proximal operator. 
As mentioned previously, the regularisation parameter $\theta \in \Theta$ controls the balance between observed and prior information and can significantly impact inferences about the unknown image $\x \in \rset^{\dim}$, especially in problems that are ill-posed or ill-conditioned. In \Cref{fig:intro-impact-reg-param}, 
we illustrate the dramatic effect that the value of $\theta \in \Theta$ may have on the recovered image for a deconvolution problem with a total-variation prior. As expected, when $\theta$ is too small the estimated image is very noisy due to lack of regularisation, and when $\theta$ is too large the resulting image is over-regularised.
\begin{figure}[!h]		
	\centering
	\centerline{\includegraphics[width=\linewidth,trim={0cm 0.4cm 0cm 0cm},clip]{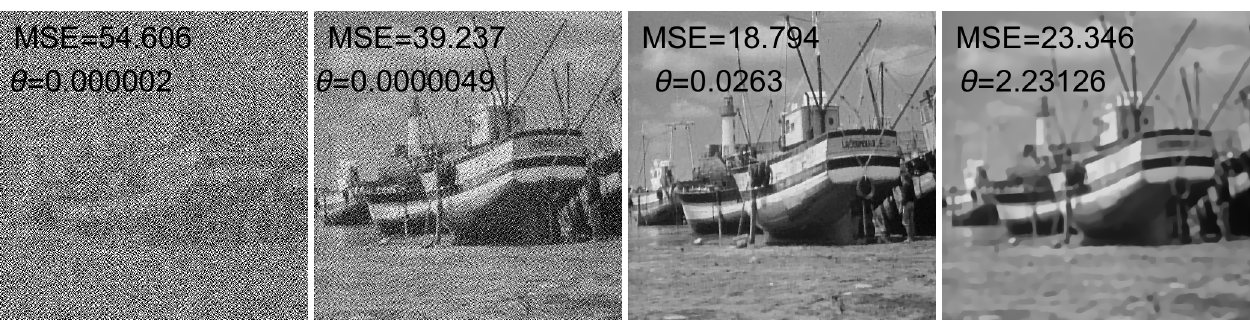}}	
	\caption{Deblurring of the \texttt{boat} image with total-variation prior  (\SNRname=40~$\mathrm{dB}$).  Maximum-a-posteriori estimators for different values of $\theta > 0$ illustrating the effect of regularisation (increasing from left to right).}
	\label{fig:intro-impact-reg-param}
	\normalsize
\end{figure}

Lastly, we want to point out that this work focuses on the estimation of $\theta$ and hence assumes that $\f(\x)$ is known, which in Gaussian observation models reduces to the knowledge of  the noise variance $\sigma^2$. This is a standard assumption in the literature (see, e.g., \cite{sugar2014stein,EUSIPCO}) that is sometimes difficult to verify in practice. To mitigate this issue, \Cref{sssec:experiments.TV-deblur-unknown-sigma} explains how to incorporate the estimation of $\sigma^2$ into the proposed scheme and illustrates this modification with an  image deconvolution experiment in which $\sigma^2$ and $\theta$ are estimated jointly. The theory that we present in our companion paper \cite{vidal:et:al:2019b} also assumes that $\f(\x)$ is known. Although a full generalisation is not possible, we believe that our theory can be extended to provide weaker convergence guarantees for some blind and semi-blind problems by using arguments similar to \cite{de2019efficient}; this is an important perspective for future work.

\section{Proposed Empirical Bayes methodology}
\label{sec:proposedMethod}
\subsection{Empirical Bayes estimation}
\label{ssec:empirical-bayes}
Under an empirical Bayesian paradigm, the regularisation parameter $\theta \in \Theta$ is estimated directly from the observed data $\y$, for example by maximum marginal likelihood estimation. That is, we compute
\begin{equation}
\thetaStar \in \underset{\theta \in \Theta}{\mathrm{argmax}}~ p(\y|\theta)\,,
\label{EQ: EB_theta_hat}
\end{equation} 
where we recall that the marginal likelihood $p(y | \theta)$ is given for any $\theta \in \Theta$ by 
\begin{equation}
p(\y|\theta)= \int_{\rset^{\dim}}p(\y|\tx)p(\tx|\theta) \rmd \tx \eqsp .
\label{EQ: EB_marginal_likelihood}
\end{equation}

Given $\thetaStar$, empirical Bayesian approaches base inferences on the pseudo-posterior distribution $x \mapsto p(x | y, \thetaStar)$, \cite{carlin2000empirical}, given for any $x \in \rset^d$ by 
\begin{equation}
p(\x|\y,\thetaStar) = \left . \exp[-\f(\x)-\thetaStar^{\top} g(\x)] \middle/ \int_{\rset^{\dim}} \exp[-\f(\tx)-\thetaStar^{\top} g(\tx)]  \rmd \tx \right . \eqsp .
\label{EQ: EBposterior}
\end{equation}
Observe that this strategy is equivalent to Bayesian model selection on a continuous class of models parametrised by $\theta$, where $\thetaStar$ produces the model with the best fit-to-data (under some additional assumptions, $p(\y|\thetaStar)$ provides the best approximation of the true distribution of $y$ in a Kullback--Leibler divergence sense, see \cite[Section 2]{White1982} and references therein).

Empirical Bayesian approaches were first considered in the statistical methodology community (see \eg \ \cite{robbins1985empirical,carlin2000empirical}), which stimulated developments in computational statistics  \cite{Robert,Atchade2011,atchade2017perturbed} to enable empirical Bayesian inference for general statistical models.  This was recently followed by important theoretical works on the validity of the empirical approach and connections to the hierarchical Bayesian paradigm (see \eg \ \cite{petrone2014bayes, knapik2016bayes, rousseau2017asymptotic}). Unfortunately, this powerful inference strategy is difficult to apply in imaging problems \cite{sanders2020effective} because the marginal likelihood $\theta \mapsto p(\y|\theta)$ is computationally intractable as it involves two intractable $\dim$-dimensional integrals, namely \eqref{EQ: normalizingConst_prior} and \eqref{EQ: EB_marginal_likelihood}, thus making the optimisation problem \eqref{EQ: EB_theta_hat} very challenging. The aim of this paper is to enable empirical Bayesian inference in imaging inverse problems, with a focus on automatic selection of regularisation parameters for convex problems that would be typically solved by using proximal optimisation techniques. More precisely, inspired by \cite{Atchade2011,atchade2017perturbed}, we propose a stochastic gradient Markov chain Monte Carlo (MCMC) algorithm to efficiently solve \eqref{EQ: EB_theta_hat} for imaging models of the general form \eqref{EQ: posterior}, where two main novelties are that we use state-of-the-art proximal MCMC methods \cite{durmus2016efficient} to construct a stochastic optimisation scheme that scales efficiently to high dimensions, and that we provide easily verifiable theoretical conditions ensuring convergence {(the latter are studied in depth in  the companion paper \cite{vidal:et:al:2019b})}. 

The maximum likelihood estimation problem \eqref{EQ: EB_theta_hat} raises natural questions about the uniqueness of $\thetaStar$, and about the log-concavity of the marginal likelihood $\theta \mapsto p(y|\theta)$, which are important for the convergence of iterative algorithms to compute $\thetaStar$. In particular, $p(\y|\theta)$ could potentially admit more than one maximiser. However, we have not observed this in practice in any imaging problem. Indeed, because in our experiments $\dimy \gg \dimtheta$, we suspect that the marginal likelihood $\theta \mapsto p(\y|\theta)$ concentrates sharply around a single maximiser $\thetaStar$, and is strongly log-concave w.r.t. $\theta$ in the neighbourhood of $\thetaStar$. These favourable properties can be formally derived under simplifying assumptions (\eg \ that $p(y|\theta)$ is fully separable on $y$ \cite{van2000asymptotic}). Extending conditions for uniqueness of \eqref{EQ: EB_theta_hat} to more general imaging problems is an important perspective for future work.

Lastly, we note that empirical Bayesian methods have found many applications in machine learning, for example in the context of feature selection (see, e.g., \cite{neal2012bayesian,tipping2001sparse,tipping2003fast}). In this field, the challenges related to high-dimensionality have been mainly addressed by using conditional Gaussian models for which the high-dimensional integrals \eqref{EQ: normalizingConst_prior} and \eqref{EQ: EB_marginal_likelihood} become tractable, thus  enabling the use of specialised strategies to solve the optimisation problem \eqref{EQ: EB_theta_hat}.

\subsection{Stochastic gradient MCMC algorithm}
\label{sec:stoch-grad-mcmc}

We now present the proposed empirical Bayesian method to solve the marginal maximum likelihood estimation problem \eqref{EQ: EB_theta_hat} and set regularisation parameters. As mentioned previously, the main difficulty in solving \eqref{EQ: EB_theta_hat} is that the marginal likelihood function $\theta \mapsto p(\y|\theta)$ is computationally intractable. 

Suppose for now that $\theta \mapsto p(\y|\theta)$ was tractable and that we had access to the gradient mapping $\theta \mapsto \nabla_{\theta} \log ~p(\y|\theta)$. Recalling that $\Theta$ is a convex compact set, we could seek to iteratively solve \eqref{EQ: EB_theta_hat} by using the projected gradient algorithm \cite{combettes2011proximalsplitting} which is given by $(\theta_n)_{n \in \nset}$ with $\theta_0 \in \Theta$ and associated with the following recursion for any $n \in \nset$
\begin{equation}
\theta_{n+1} = \Pi_{\Theta} \left[\theta_n+\delta_n \nabla_{\theta} \log ~p(\y|\theta_n)\right] \eqsp ,
\label{EQ: projected_gradient}
\end{equation}
where $\Pi_{\Theta}$ is the projection onto $\Theta$ and $(\delta_n)_{n \in \nset}$ is a sequence of non-increasing step-sizes. As mentioned previously, because in imaging problems $\dimy \gg \dimtheta$, the marginal likelihood $\theta \mapsto p(y|\theta)$ typically exhibits a single maximiser $\thetaStar$ and is strongly log-concave w.r.t. $\theta$ in the neighbourhood of $\thetaStar$. Therefore, we expect that \eqref{EQ: projected_gradient} would quickly converge.

Since $\theta \mapsto \nabla_{\theta} \log ~p(\y|\theta)$ is not tractable, we cannot directly use \eqref{EQ: projected_gradient} to compute $\thetaStar$.
However, we can replace $\theta \mapsto \nabla_{\theta} \log ~p(\y|\theta)$ with a noisy estimate and consider a stochastic variant of the projected gradient algorithm. In particular, under mild regularity assumptions, using Fisher's identity  (see \Cref{prop:fisher}) and the fact that for any $x \in \rset^{\dim}$, $y \in \rset^{\dimy}$ and $\theta \in \Theta$, $p(x, y | \theta) = p(y | x) p(x | \theta)$, we obtain that for any $\theta \in \Theta$
\begin{equation}
\nabla_{\theta} \log p(\y|\theta)= \int_{\rset^{\dim}}p(\tx|y, \theta) \nabla_{\theta} \log p(\tx,\y|\theta)  \rmd \tx 
= - \int_{\rset^{\dim}} g(\tx) p(\tx|y, \theta)  \rmd \tx - \nabla_{\theta} \log(\rmZ(\theta)) \eqsp .
\label{EQ: fisherIdentity_vectorial}
\end{equation}
Hence, we can use Monte Carlo Markov chain methods  to approximate $\theta \mapsto \nabla_{\theta} \log p(\y | \theta)$ for any $\theta \in \Theta$.

We consider a stochastic approximation proximal gradient algorithm (SAPG), see \cite{fort2017stochastic}, where the expectation $\int_{\rset^{\dim}} g(\tx) p(\tx|\y,\theta) \rmd \tx $  is replaced by a Monte Carlo estimator leading to the following gradient estimate for any $\theta \in \Theta$
\begin{equation}
\label{EQ: MC-gradient-estimate}
\Delta_{m, \theta} = \frac{1}{m} \sum_{k=1}^{m} \nabla_{\theta} \log p(X_k ,y |\theta ) = -\nabla_{\theta}\log \rmZ(\theta) -\frac{1}{m} \sum_{k=1}^{m} \g(X_k)\eqsp ,
\end{equation}
where  $(X_k)_{k \in \lbrace 0, \dots, m \rbrace}$ is a sample of size $m \in \nsets$ generated by using a Markov Chain targeting $p(x|y,\theta) = p(x,y|\theta)/p(y|\theta)$, or a regularised approximation of this density. Therefore, to compute $\thetaStar$, we can build a new sequence $(\theta_n)_{n \in \nset}$ associated with the following recursion for any $n \in \nset$
\begin{equation}
\theta_{n+1} = \Pi_{\Theta}[\theta_n + \delta_{n+1} \Delta_{m_n, \theta_n}] \eqsp ,  \qquad \Delta_{m_n, \theta_n} = -\nabla_{\theta}\log \rmZ(\theta_n) -\frac{1}{m_n} \sum_{k=1}^{m_n} \g(X_k^n)\eqsp ,
\label{EQ: EB_gral_def_proposed_algo}
\end{equation}
starting from some $\theta_0 \in \Theta$, and where $(m_n)_{n \in \nset}$ is a sequence of non-decreasing sample sizes.

Under some assumptions on $(m_n)_{n \in \nset},\, (\delta_n)_{n \in \nset}$ and on the Markov kernels (see \Cref{ssec:selectionParamsMYULA}), the errors in the gradient estimates asymptotically average out and the algorithm converges to a maximiser of $\theta \mapsto p(y | \theta)$. More precisely, given $N \in \nset$, a sequence of non-increasing weights $(\omega_n)_{n \in \nset}$, and a sequence $(\theta_n)_{n = 0}^{N-1}$ generated using \eqref{EQ: EB_gral_def_proposed_algo}, an approximate solution of \eqref{EQ: EB_theta_hat} can be obtained by calculating, for example, the weighted average\footnote{\label{foot:rev2-comment2} Averaging iterates is standard in stochastic approximation algorithms. Most known convergence results concern the almost sure convergence of $(p(y|\thetaavg_N))_{N \in \nset}$ towards $\min_{\theta \in \Theta} p(y|\theta)$, or alternatively a weaker convergence in expectation (see, e.g., \cite{bachmoulines2011,rakhlin2011making,atchade2017perturbed}).}
\begin{equation}
\label{eq:thetaavg}
\thetaEB =  \left. \sum_{n=0}^{N-1} \omega_{n} \theta_n \middle/ \sum_{n=0}^{N-1} \omega_n \right.  \eqsp.
\end{equation}
which converges asymptotically to a solution of \eqref{EQ: EB_theta_hat} as $N \rightarrow \infty$ (see \cite{atchade2017perturbed} for details).

Applying this strategy to imaging problems is highly non-trivial for two reasons: i) it requires generating very high-dimensional Markov chains $\ensembleLigne{(X_k^n)_{k \in \lbrace 0, \dots, m_n \rbrace}}{n \in \nset}$ in a way that is computationally efficient;  ii) the Markov chains must satisfy a number of complex technical conditions to ensure the convergence of the optimisation scheme (these conditions are related to the stochastic properties of the sequence of errors of the gradient estimates computed from the Markov chains; they are discussed in detail in our companion paper \cite{vidal:et:al:2019b}).

In this work, we address these two major difficulties by
constructing an SAPG scheme with state-of-the-art unadjusted
proximal Markov kernels that are highly computationally efficient
and that automatically satisfy the required theoretical
conditions.  More importantly, we
show both theoretically and empirically that a single sample
($m_n=1$) per iteration is enough to guarantee the convergence of
the proposed SAPG scheme.
This allows  delivering accurate estimates of regularisation parameters in a
computationally scalable way and with theoretical
guarantees.

More precisely, to construct the SAPG optimisation scheme we use the Moreau-Yoshida Unadjusted Langevin Algorithm (MYULA) \cite{durmus2016efficient}, which is a state-of-the-art proximal Markov kernel specifically designed for high-dimensional distributions with terms that are log-concave but not smooth. Accordingly, to draw samples from the posterior $p(x|y,\theta) = p(x,y|\theta)/p(y|\theta)$, we define the Markov chain $(X_k)_{k \in  \nset}$, starting from $X_0 \in \rset^{\dim}$, given by the recursion
\begin{equation}
\label{EQ: MYULA_explicit_posterior} 
\Rker_{\gamma,\lambda,\theta}:\,\quad X_{k+1} = X_k - \gamma \nabla_{x}\f(X_k) - \gamma \left.\defEns{X_k - \prox_{\theta^{\top}g}^{\lambda}(X_k)}\middle/\lambda \right. + \sqrt{2\gamma} Z_{k+1} \eqsp ,
\end{equation}
where $\prox_{\theta^{\top}g}^{\lambda}$ is defined by
\eqref{EQ: proxOperator}, $\lambda \in \mathbb{R}^+$ is a
smoothing parameter, $\gamma \in \mathbb{R}^+$ is a
discretisation step-size, and $(Z_k)_{k \in \N^*}$ is a
sequence of i.i.d. $\dim$-dimensional zero-mean Gaussian
random variables with an identity covariance matrix. For any
$\gamma \in \mathbb{R}^+$ and $\theta \in \Theta$, we denote
by $\Rker_{\gamma,\lambda,\theta}$ the Markov kernel
associated with the recursion \eqref{EQ:
	MYULA_explicit_posterior}. We refer the reader to
\cite{durmus2016efficient} and to our companion paper
\cite{vidal:et:al:2019b} for more details about this Markov
kernel, particularly concerning its relationship to the
Langevin diffusion process and to proximal optimisation
methodology.

Lastly, observe that in order to use \eqref{EQ: EB_gral_def_proposed_algo} it is necessary to evaluate $\theta \mapsto \nabla_{\theta}\log \rmZ(\theta)$. For most models of interest, $\theta \mapsto \nabla_{\theta}\log \rmZ(\theta)$  cannot be computed exactly and needs to be approximated.  Hence,  we propose three different strategies to address this calculation depending on whether $g$ is a homogeneous function or not.

\subsubsection{Scalar-valued $\theta$ with $\alpha$ positively homogeneous regulariser}
\label{ssec:proposedMethod-scalarTheta-homogeneous}       
For scalar-valued $\theta$, \ie \ $\dimtheta = 1$, \eqref{EQ: fisherIdentity_vectorial} is given for any $\theta \in \Theta$ by
\begin{equation}
\frac{\textrm{d}}{\textrm{d}{\theta}}  \log p(\y|\theta) = - \int_{\rset^{\dim}} g(\tx) p(\tx|\y,\theta)~ \textrm{d}\tx \,-\frac{\textrm{d}}{\textrm{d}{\theta}} \log \rmZ(\theta)\,.
\label{EQ: fisherIdentityApplied-scalar}
\end{equation}
Assume that there exists $\alpha \in \rset \backslash \{0\}$ such that $\g$ is a $\alpha$ positively homogeneous function, \ie \ for any $x \in \rset^d$ and $t > 0$, $g(tx) = t^{\alpha} g(x)$, and recalling that $\Theta \subset \ooint{0,+\infty}$ we have for any $\theta \in \Theta$
\begin{equation}
\rmZ(\theta) = \int_{\rset^{\dim}} \rme^{- \theta g(\tx)} \rmd \tx = \int_{\rset^{\dim}} \rme^{-  g(\theta^{1/\alpha} \tx)} \rmd \tx = \theta^{-\dim/\alpha} \int_{\rset^d} \rme^{-g(\tx)} \rmd \tx \eqsp ,
\label{EQ: partFuncHomogenCase}
\end{equation} and therefore
\begin{equation}
\frac{\textrm{d}}{\textrm{d}{\theta}} \log \rmZ(\theta)=-\dim /(\alpha  \theta) \eqsp .
\label{EQ: partFuncDerivHomogenCase}
\end{equation}
Hence, \eqref{EQ: fisherIdentityApplied-scalar} becomes for any $\theta \in \Theta$
\begin{equation}
\frac{\textrm{d}}{\textrm{d}{\theta}}  \log p(\y|\theta) = \dim /(\alpha  \theta) - \int_{\rset^{\dim}} g(\tx) p(\tx|\y,\theta)\rmd \tx \eqsp , 
\label{EQ: gradientHomogCase}
\end{equation}  
which leads to  \Cref{algo:MCMC_single_chain} below. 	
We want to point out that many commonly used regularisers are positively homogeneous. For example, all norms such as $\ell_1$, $\ell_2$, total variation (\tvname ), nuclear or compositions of norms with linear operators (e.g., analysis terms of the form $\|\Psi \x\|_1$, where $\Psi \in \rset^{\dim_1} \times \rset^{\dim}$ with $\dim_1 \in \nset$) are $1$ positively homogeneous. Moreover, powers of norms with exponent $q > 0 $ are $q$ positively homogeneous, and all linear combinations of positively homogeneous functions with the same homogeneity constant $\alpha$, are also $\alpha$ positively homogeneous. 
\begin{algorithm}
	\caption{SAPG algorithm - Scalar $\theta$ and $\alpha$ positively homogeneous regulariser $\g$}
	\label{algo:MCMC_single_chain}
	\begin{algorithmic}[1]		
		\STATE Input: initial $\{\theta_0, X_0^0\}$, $(\delta_n,\omega_n,m_n)_{n \in \nset}$, $ \Theta $,  kernel parameters $\gamma,\lambda$, iterations $N$.
		\FOR{$n = 0$ to $N-1$}
		\IF{$n>0$}
		\STATE Set
		$X_0^n = X_{m_{n-1}}^{n-1}$, 
		\ENDIF
		\FOR{$k = 0$ to $m_n-1$}
		\STATE Sample $X^n_{k+1} \sim \Rker_{\gamma,\lambda,\theta_n}(X^n_{k}, \cdot)$,
		\ENDFOR
		\STATE Set $\theta_{n+1} = \Pi_{\Theta}\left[\theta_n + \frac{\delta_{n+1}}{m_n} \sum_{k=1}^{m_n} \defEns{ \frac{\dim}{\alpha \theta_n} - \g(X_k^{n})}\right]$.	
		\ENDFOR				
		\STATE Output: $\thetaEB$ computed with \eqref{eq:thetaavg}.
	\end{algorithmic}
\end{algorithm}

\subsubsection{Separably homogeneous regulariser}
\label{ssec:separable_vectorial_theta}        
For the special case of separably homogeneous regularisers, \Cref{algo:MCMC_single_chain} can be adapted for multivariate $\theta$. This is because in this class of regulariser, each component of $\theta$ affects independent subsets of the components of $\x$. More precisely, assume that $g$ is separably homogeneous in the following sense: there exist  $(\tg_i)_{i \in \{1, \dots, \dimtheta\}}$, $(\msa_i)_{i \in \{1, \dots, \dimtheta\}}$ pairwise disjoint subsets of $\{1, \dots, \dim\}$ and $(\alpha_i)_{i \in \{1, \dots, \dimtheta\}} $ such that for any $i \in \{1, \dots, \dimtheta \}$, $\tg_i : \ \rset^{d_i} \to \rset$ is $\alpha_i$-positively homogeneous with $\alpha_i > 0$ and for any $x \in \rset^{\dim}$, $g(x) = (\tg_i(x_{[\msa_i]}))_{i \in \{1, \dots \dimtheta\}}$ where for any $\msa = \{i_1, \dots, i_{\ell}\} \subset \{1, \dots, \dim\}$, $x_{[\msa]} = (x_{i_1}, \dots, x_{i_{\ell}})$.        
In this case we have for any $\theta \in \Theta$
\begin{align}
	\rmZ(\theta) &= \int_{\rset^{\dim}} \exp[-\theta^{\top} \g(\tx) ] \rmd \tx 
	=  \int_{\rset^d} \exp\parentheseDeux{- \sum_{i=1}^{\dimtheta} \theta^i \tg_i(\tx_{[\msa_i]})} \rmd \tx \\
	&= \prod_{i=1}^{\dimtheta} \int_{\rset^{\abs{\msa_i}}} \exp[- \theta^i \tg_i(\tx_{[\msa_i]})] \rmd \tx \eqsp .
\end{align}
Therefore, for any $i \in \{1, \dots, \dimtheta\}$ and $\theta \in \Theta$ we get that
$$
[\partial \log \rmZ / \partial \theta^i](\theta) = -\abs{\msa_i} / (\alpha_i \theta^i).
$$
Using this property we obtain \Cref{algo:MCMC_single_chain_separable},  
where for any $i \in \{1, \dots, \dimtheta \}, \  {\theta^i \in \Theta^i  \subset \ooint{0,+\infty}}$ and $ \Pi_{\Theta^i}$ is the projection onto $\Theta^i$.
\begin{algorithm}
	\caption{SAPG algorithm - Multivariate $\theta$ and separably homogeneous regulariser}
	\label{algo:MCMC_single_chain_separable}
	\begin{algorithmic}[1]		
		\STATE Input: initial $\{\theta_0, X_0^0\}$, $(\delta_n,\omega_n,m_n)_{n \in \nset}$, $ \Theta $,  kernel parameters $\gamma,\lambda$, iterations $N$.
		\FOR{$n = 0$ to $N-1$}
		\IF{$n>0$}
		\STATE Set
		$X_0^n = X_{m_{n-1}}^{n-1}$, 
		\ENDIF
		\FOR{$k = 0$ to $m_n-1$}
		\STATE Sample $X^n_{k+1} \sim \Rker_{\gamma,\lambda,\theta_n}(X^n_{k}, \cdot)$,
		\ENDFOR
		\FOR{$i = 1$ to $\dimtheta$}
		\STATE Set $\theta^i_{n+1} = \Pi_{\Theta^i}\left[\theta^i_n + \frac{\delta_{n+1}}{m_n} \sum_{k=1}^{m_n} \defEns{ \frac{\abs{\msa_i}}{\alpha_i \theta_n^i}- \tg_i\left({X^n_{k}}_{[\msa_i]}\right)}\right]$.	
		\ENDFOR		
		\ENDFOR					
		\STATE Output: $\thetaEB$ computed with \eqref{eq:thetaavg}.
	\end{algorithmic}
\end{algorithm}

For example, many works in the imaging literature adopt a so-called synthesis formulation where $x$ represents the unknown image on some orthonormal wavelet basis $\Psi \in \mathbb{R}^{d\times d}$ with $J \in \nset $ levels\footnote{In synthesis formulations $x \in \mathbb{R}^d$ represents the unknown image on some basis $\Psi \in \mathbb{R}^{d\times d}$; the solution in the pixel domain is given by $\Psi^\top x$.}, and consider level-adapted $\ell_1$ regularisations of the form
$$
\theta^\top g(x) = \sum_{j = 1}^J \theta_j \|x_{[\msa_j]}\|_1
$$
where $x_{[\msa_j]}$ are the elements of $x$ associated with the $J$th level and $\theta \in \mathbb{R}^J$. Here, $g$ is a separably homogeneous functional as it can be expressed as $g = (\tg_1, \ldots,\tg_J)$ where, for any $j \in \{1,\ldots,J\}$, $ \tg_j $ is 1-positively homogeneous and $d_j=\abs{\msa_j}$. Notice that the domain in which $x$ is represented is not relevant here; \Cref{algo:MCMC_single_chain_separable} can be directly applied to any model where $g$ is homogenous separable via a change of basis because the same expression for $\rmZ(\theta)$ holds.

\subsubsection{General case: inhomogeneous regulariser}
\label{ssec:proposedMethod-vectorialTheta-GeneralCase}
When g is neither homogeneous nor separably homogeneous, we address the evaluation of $\theta \mapsto \nabla_\theta \log \rmZ(\theta)$ numerically by stochastic simulation. Using that $\y$ is conditionally independent of $\theta$ given $\x$, and using identity \eqref{EQ: prior_expect_g}, we express $\theta \mapsto \nabla_{\theta} \log ~ p(\y|\theta)$ as the difference between two expectations, \ie \ for any $\theta \in \Theta$
\begin{equation}
\nabla_{\theta} \log ~ p(\y|\theta) = \int_{\rset^{\dim}}\g(\tx)p(\tx|\theta)~ \textrm{d}\tx -\int_{\rset^{\dim}}\g(\tx)p(\tx|\y,\theta)~ \textrm{d}\tx \eqsp .
\label{EQ: simplifiedGradientEstimate}
\end{equation}
We then use two families of Markov kernels $\ensembleLigne{\Rker_{\gamma,\lambda,\theta}, \bRker_{\gamma',\lambda',\theta}}{\gamma,\gamma' \in \ocint{0, \bgamma}, \  {\lambda,\lambda' \in \mathbb{R}^+},\ {\theta \in \Theta}}$ that  target the posterior $p(x|y, \theta)$ and the prior $p(x|\theta)$, respectively, within the SAPG \Cref{algo:MCMC_double_chain} below.

\begin{algorithm}
	\caption{SAPG algorithm - General form}
	\label{algo:MCMC_double_chain}
	\begin{algorithmic}[1]
		\STATE Input: initial $\{\theta_0, X_0^0, \bX_0^0\}$, $(\delta_n,\omega_n,m_n)_{n \in \nset}$, $ \Theta $,  kernel parameters $\gamma,\gamma',\lambda,\lambda'$, iterations $N$.		
		\FOR{$n = 0$ to $N-1$}
		\IF{$n>0$}
		\STATE Set
		$X_0^n = X_{m_{n-1}}^{n-1}$, 
		\STATE Set $\bar{X}_0^n = \bar{X}^{n-1}_{m_{n-1}}$, 
		\ENDIF
		\FOR{$k = 0$ to $m_n-1$}
		\STATE Sample $X^n_{k+1} \sim \Rker_{\gamma,\lambda,\theta_n}(X^n_k, \cdot)$,
		\STATE Sample $\bX^n_{k+1} \sim \bRker_{\gamma',\lambda',\theta_n}(\bX^n_k, \cdot)$,
		\ENDFOR
		\STATE Set $\theta_{n+1} = \Pi_{\Theta}\left[\theta_n + \frac{\delta_{n+1}}{m_n} \sum_{k=1}^{m_n} \defEns{g(\bX^n_{k}) - \g(X^n_{k})}\right]$.	
		\ENDFOR			
		\STATE Output: $\thetaEB$ computed with \eqref{eq:thetaavg}.
	\end{algorithmic}
\end{algorithm}

For any $\theta \in \Theta$ and $\gamma, \gamma' \in \mathbb{R}^+$, the Markov kernel $\Rker_{\gamma,\lambda,\theta}$ is as defined previously in \eqref{EQ: MYULA_explicit_posterior}, and the additional Markov kernel $\bRker_{\gamma',\lambda',\theta}$ is  associated
with the MYULA algorithm targeting $p(x|\theta)$, which
defines $(\bX_k)_{k \in \nset}$, starting from $\bX_0 \in \rset^{\dim}$, given by the  recursion 
\begin{equation}
\label{EQ: MYULA_explicit_prior}
\bRker_{\gamma',\lambda',\theta}:\,\quad \bX_{k+1} = \bX_k - \gamma'\left. \defEns{\bX_k - \prox_{\theta^{\top}g}^{\lambda'}(\bX_k)}\middle /\lambda'
+ \sqrt{2\gamma'} Z_{k+1} \right. \eqsp ,
\end{equation}		
where $\lambda^\prime, \gamma^\prime>0$ are the smoothing parameter and step-size parameter, respectively.

\subsection{Implementation guidelines}
\label{ssec:selectionParamsMYULA}
We now discuss suitable ranges and recommended values for the parameters of \Cref{algo:MCMC_single_chain}, \Cref{algo:MCMC_single_chain_separable} and \Cref{algo:MCMC_double_chain}. Rather than optimal values for specific models, our recommendations seek to provide general rules that are simple and robust. We also discuss some other considerations related to the implementation of the methods. Please see \Cref{append:sec:practicalImplementGuidelines} for implementation and troubleshooting guidelines.
\subsubsection{Setting the algorithm parameters}
\label{append:ssec:setting-algo-parameters}
{   
	\paragraph{Selecting $\gamma$} Our theoretical convergence analysis \cite{vidal:et:al:2019b} requires setting $0 < \gamma < (\Ly +1/\lambda)^{-1}$; this is related to the numerical stability of the Markov chains and stems from the fact that $\Ly +1/\lambda$ bounds the Lipschitz constant of $\nabla \f + \left(x -
	\prox_{\theta^{\top} \g}^{\lambda}(x)\right)/\lambda $. Within this stability range, $\gamma$ controls a trade-off between computational efficiency and accuracy, with larger values of $\gamma$ leading to higher efficiency but also to a larger asymptotic bias. Given the dimensionality involved, and that in our experiments we did not observe any significant bias issues, we recommend using a large $\gamma$, e.g., $\gamma = 0.98 (\Ly +1/\lambda)^{-1}$.
	
	\paragraph{Selecting $\lambda$}This parameter controls the regularity of the smooth approximation of $g$ within MYULA and hence another trade-off between bias and convergence speed \cite{durmus2016efficient}. We have empirically observed that in order to prevent a significant bias it is necessary to set $\lambda \in (0,2)$. Within this range, we prefer larger values of $\lambda$ to improve convergence speed, at the expense of some bias. We recommend using $\lambda = \min(\Ly^{-1}, 2)$, as setting $\lambda \gg \Ly^{-1}$ increases asymptotic bias without improving convergence speed because of the effect of $\Ly$ on $\gamma$.
	
	\paragraph{Selecting $\gamma'$ and $\lambda'$} Since $\Ly$ does not affect the  kernel $\bRker_{\gamma',\lambda',\theta}$ targeting the prior, the stability range for $\gamma'$ is $0 < \gamma'<\lambda'$. In our experiments we set $\gamma' = 0.98 \lambda^\prime$. We usually set $\lambda' = \lambda$ to have the same level of smoothing in both chains, however one can also use $\lambda' \gg \lambda$ if $\bRker_{\gamma',\lambda',\theta}$ is much slower than $\Rker_{\gamma',\lambda',\theta}$ (see \Cref{append:ssec:tips-for-2-chains} and \Cref{append:ssec:convergence-speed} for more details).

	\paragraph{Selecting $(\delta_n,m_n)_{n \in \nset}$} For simplicity and computational efficiency, we recommend using
	a single ($m_n=1$) Monte Carlo sample per iteration. A single sample is sufficient to construct a convergent SAPG scheme see our companion paper
	\cite{vidal:et:al:2019b}. 
	We recommend setting $\delta_n = \czero n^{- \powerdelta}$ with $\powerdelta \in [0.6, 0.9]$, and use $\delta_n =  \czero ~n^{- 0.8}$ in our experiments, which is a standard choice in the literature \cite{bottou2012stochastic}. For $\czero$ we recommend, for the case where $\theta$ is scalar, starting with $\czero =(\theta_0 \dim)^{-1}$ and then adjust if necessary. Although the choice of $\czero$ is asymptotically irrelevant {(see \Cref{fig:res-num-synth-delta-c0} (b))}, if the initial step-size is too large the iterate $\theta_n$ will be bouncing on the limits of the interval for a long transient regime, whereas convergence will be slow if $\czero$ is too small. For this reason, we recommend adjusting $\czero$ so that the step-size is of the order of the projection interval $\Theta$. When $\theta$ is not scalar, one can use different scales for each component of $\theta$. More details are provided in the  \Cref{append:ssec:vectorial-theta}.}

\paragraph{Selecting $(\omega_n)_{n \in \nset}$ and $N$}
While it is possible to construct other estimates, we recommend using the average
\begin{equation}
\thetaEB =  \left. \sum_{n=0}^{N-1} \omega_{n} \theta_n \middle/ \sum_{n=0}^{N-1} \omega_n \right.  \eqsp,
\end{equation}
with $(\omega_n)_{n \in \nset}$ given by       		
\begin{equation}
\omega_n =
\begin{cases}
0 \eqsp, & \text{if } n < N_0  \eqsp, \\
1 \eqsp, & \text{if } N_0 \leq n \leq N_1  \eqsp, \\
\delta_n  & \text{ otherwise} \eqsp,
\end{cases}
\end{equation} 
\noindent where $N_0, N_1 \in \nset$, $N_0 <
N_1$. This choice of $(\omega_n)_{n \in \nset}$,
defines three distinct phases: i) a burn-in phase
where the first $N_0$ iterations of the algorithm are
discarded to reduce the non-asymptotic bias (this is
particularly important when using a small number of
iterations); ii) a uniform averaging phase
$N_0 \leq n \leq N_1$ where the smoothing effect
associated with averaging improves convergence speed
and reduces estimation variance; iii) a refinement
phase where we use decreasing weights to improve the
precision of the estimator (see
\cite{vidal:et:al:2019b} for accuracy guarantees).

We have empirically observed that imaging problems do not usually require highly accurate estimates of $\theta$. Therefore, in the interest of computational efficiency, in our experiments we omit the third phase and stop when $N_1=N$.          
Moreover, rather than using the theoretical accuracy guarantees of \cite{vidal:et:al:2019b} to set $N$, we monitor $\absLigne{\thetaavg_{N+1} - \thetaavg_{N}}/\thetaavg_N$ and stop the algorithms when $\absLigne{\thetaavg_{N+1} - \thetaavg_{N}}/\thetaavg_N <\tau$ for a prescribed tolerance $\tau > 0$ (e.g., $\tau = 10^{-3}$).         
\paragraph{Selecting $\Theta$}When selecting the projection
interval, the lower bound should be as small as necessary but
not zero, as this may render the algorithm unstable (the
gradient  depends on $\theta^{-1}$ and diverges as
$\theta_n \rightarrow 0$). If possible, use tight bounds
to improve convergence speed.
\paragraph{Selecting $\theta_0$ } The choice of $\theta_0 \in \Theta$ is theoretically asymptotically irrelevant (see, e.g., \Cref{fig:res-num-synth-delta-c0} (a)). However, in some cases a very bad initialisation can prevent the algorithm from converging, e.g., by introducing large numerical errors in the computation of proximal operators. We observed this in the TGV denoising experiment \Cref{sssec:experiments.tgv-denois} when using the extreme initialisation $\theta^1_0=\theta^2_0=100$ (not reported in the paper).
       
\subsubsection{Other implementation considerations}\label{sssec:other-imp-guidelines}
\paragraph{Implementation in logarithmic scale} The proposed algorithms to estimate $\theta$ often exhibit better numerical convergence properties when they are implemented in a logarithmic scale, which is a standard strategy for scale parameters \cite{Atchade2011}. Accordingly, we introduce the change of variables $\eta = \log(\theta)$, obtain an estimate $\hat{\eta}$ by using one of the proposed algorithms to maximise the marginal likelihood $p(y|\eta)$, and then set $\hat{\theta} = \rme^{\hat{\eta}}$. This is equivalent to maximising $p(y|\theta)$ because of the invariance to re-parametrisation property of the maximum likelihood estimator. This change of variables requires a minor modification in the computation of the gradients, which have to be multiplied by $\rme^\eta_n$ to satisfy the chain rule. For example, step 9 in \Cref{algo:MCMC_single_chain} becomes $\eta_{n+1} = \Pi_{\Theta^\eta}\left[\eta_n + \rme^{\eta_n}\frac{\delta_{n+1}}{m_n} \sum_{k=1}^{m_n} \defEns{ \frac{\dim\rme^{-\eta_n}}{\alpha} - \g(X_k^{n})}\right]$, where $\Theta^\eta = \{\log(\btheta): \btheta \in \Theta\}$ denotes the range of admissible values of $\eta$ taking the logarithm component-wise. Similarly, step 11 of \Cref{algo:MCMC_double_chain} becomes $\eta_{n+1} = \Pi_{\Theta^\eta}\left[\eta_n + \rme^{\eta_n}\frac{\delta_{n+1}}{m_n} \sum_{k=1}^{m_n} \defEns{g(\bX^n_{k}) - \g(X^n_{k})}\right]$.

\paragraph{Initialisation of the Markov kernels} We strongly recommend warm-starting the Markov chains by running $T_0 \in \nset$ iterations with fixed $\theta=\theta_0$ before starting to update the value of $\theta$; an appropriate value for $T_0$ can be easily determined by monitoring the statistic $g(X_k)$. 

\paragraph{Alternative implementations of MYULA}
We want to point out that \eqref{EQ: MYULA_explicit_posterior} is not the only possible way of implementing MYULA to sample from $p(x|y,\theta) = p(x,y|\theta)/p(y|\theta)$. If some of the functions $g_i$ are Lipschitz differentiable, it might be convenient to incorporate them through their gradient, and reserve proximal operators for the non-differentiable terms in $g$. Moreover, the above implementation requires computing the proximal operator of the global function $\theta^{\top} \g$. In some cases it might be easier to use the proximal operators of each individual $g_i$ independently (see \cite{durmus2016efficient}). Which implementation of MYULA is the most convenient, will mostly depend on the tools available to the practitioner. Many optimisation algorithms for MAP estimation \eqref{EQ: mapEstim} also use the operators $\nabla \f$ and either  $\prox_{\theta^{\top} \g}^{\lambda}$ or $ \prox_{\theta^i \g_i}^{\lambda} $ \cite{combettes2011proximalsplitting,green2015bayesian}, making the implementation of our methods straightforward for problems currently solved with such tools. Our theoretical analysis \cite{vidal:et:al:2019b} also considers implementations of \Cref{algo:MCMC_single_chain}, \Cref{algo:MCMC_single_chain_separable} with other proximal Markov kernels. Moreover, although for simplicity here we use constant values for $\gamma, \lambda, \gamma^\prime$, and $\lambda^\prime$, the theory presented in \cite{vidal:et:al:2019b} allows implementing the algorithms with iteration-dependent values $(\gamma_n, \lambda_n, \gamma^\prime_n, \lambda^\prime_n)_{n \in \nset}$.

\paragraph{Estimation bias}Lastly, as mentioned previously, \Cref{algo:MCMC_single_chain}, \Cref{algo:MCMC_single_chain_separable}, and \Cref{algo:MCMC_double_chain} can exhibit some asymptotic estimation bias. This error arises from the fact that the MYULA kernels used do not target the posterior or prior distributions exactly but rather an approximation of these distributions \cite{durmus2016efficient}. This error is controlled by $\lambda, \lambda', \gamma$ and $\gamma'$, and can be made arbitrarily small at the expense of additional computing time, see \cite{vidal:et:al:2019b}. The bias can be completely removed by using iteration-dependent smoothing and step-size parameters and letting them vanish as $n \rightarrow \infty$, but this can dramatically deteriorate the non-asymptotic convergence properties. The bias can also be completely removed by combining MYULA with Metropolis-Hastings (MH) steps, as discussed in detail in \cite{pereyra2016proximal}. However, 
it is difficult to calibrate high-dimensional MH steps within a SAPG scheme to achieve the required acceptance rates, as the target densities change at each iteration, so we do not explore this any further. A high-dimensional MH correction can also deteriorate the non-asymptotic convergence properties of the algorithms and significantly increase computing times \cite{durmus2016efficient}.

\subsection{Connections to hierarchical Bayesian approaches}
\label{ssec:connectBayesApproach}
As we mentioned earlier, the Bayesian framework provides two main paradigms to select $\theta$ automatically: the empirical (already discussed in \Cref{ssec:empirical-bayes}) and the hierarchical, which is currently the predominant Bayesian approach in data science (see \cite{pereyra2013estimating,EUSIPCO} for examples in imaging sciences). We now discuss connections between the two paradigms and stress advantages and disadvantages.

In the hierarchical Bayesian paradigm, $\theta$ is modelled as an additional unknown quantity and it is assigned a prior distribution $p(\theta)$. This leads to an augmented posterior given for any $\theta \in \Theta$ and $x \in \rset^{\dim}$ by 
\begin{equation}
p(\x,\theta|\y)=p(\y|\x,\theta)p(\x|\theta)p(\theta) /p(\y) \eqsp .
\label{EQ: HB_augmented_posterior}
\end{equation}
There are two main ways of employing this augmented posterior. The first, and most popular, is to remove $\theta$ from the model by marginalisation, followed by inference on $\x|\y$ with the marginal posterior given for any $x \in \rset^{\dim}$ by 
\begin{equation}
p(\x|\y) = \int_{\rset^{\dimtheta}} p(\x,\ttheta|\y) \textrm{d}\ttheta \eqsp .
\label{EQ: HB_marginal_posterior_x}
\end{equation}
The marginal posterior is then often used to perform minimum mean squared error (MMSE) estimation by computing 
\begin{equation}
\hat{x}_{\MMSE} = \int_{\rset^{\dim}} \tx \, p(\tx|\y)~ \textrm{d}\tx.
\label{EQ: HB_marginal_MMSE}
\end{equation}
This can be achieved with a standard MCMC algorithm when $\rmZ(\theta)$ is tractable, \eg \ Gibbs sampling,   or with specialised algorithm that allows circumventing the evaluation of $\rmZ(\theta)$ at the expense of significant additional computational cost (see \cite{pereyra2013estimating} for details). For some specific models it is also possible to compute an approximate marginal MMSE solution by using a deterministic variational Bayesian algorithm (e.g., see \cite{babacan2008parameter,marnissi2017variational}),  but such algorithms have not yet been widely adopted because their implementation and performance remains very problem-specific.  

Alternatively, for the class of models considered in \Cref{ssec:proposedMethod-scalarTheta-homogeneous}, one can also efficiently compute the marginal MAP estimator
\begin{equation}
\hat{\x}_{\MAP} \in \underset{\x \in \rset^{\dim}}{\mathrm{argmin}}~ p(\x|\y) \eqsp ,
\label{EQ: mapEstim_marginalisation}
\end{equation}
by using the majorisation-minimisation algorithm proposed in \cite{EUSIPCO}. In some of our experiments we report comparisons with this the hierarchical Bayesian method, as it can be broadly applied to same models as \cref{algo:MCMC_single_chain} and \cref{algo:MCMC_single_chain_separable}.

In order to understand the connection between this hierarchical Bayesian approach and the empirical Bayesian strategy used in this paper it is useful to express $p(\x|\y)$ as follows
\begin{equation}
p(\x|\y) = \int_{\Theta} p(\x|\y,\ttheta) p(\ttheta|\y) \textrm{d}\ttheta \eqsp ,
\end{equation}
where we observe that $x \mapsto p(\x|\y)$ is effectively a weighted average of all the posteriors $x \mapsto p(\x|\y,\theta)$ parametrised by $\theta \in \rset^{\dimtheta}$, with weights given by the marginal posterior $p(\theta|\y)$, which represents the uncertainty in $\theta$ given the observed data $\y$. If instead of $p(\theta|\y)$ we perform the integration of $\theta \mapsto p(\x|\y, \theta)$ with respect to the Dirac distribution $\updelta_{\thetaStar}$, we obtain the empirical Bayesian pseudo-posterior $x \mapsto p(\x|\y,\thetaStar)$ considered in this paper. 

Note that in imaging problems the marginal posterior $p(\theta|\y) \propto p(\y|\theta)p(\theta)$ will be dominated by the marginal likelihood $p(\y|\theta)$ because of the dimensionality of $\y$. Therefore most of the mass of $p(\theta|y)$ will be close to $\thetaStar$. As a result, we expect that both the hierarchical and the empirical approaches will deliver broadly similar results. For models that are correctly specified both strategies should perform well, and hierarchical Bayes should moderately outperform empirical Bayes as it is decision-theoretically optimal \cite{robert2007bayesian}. 

However, most imaging models are over-simplistic and hence somewhat misspecified. Our experiments suggest that in this case the empirical Bayesian approach can outperform the hierarchical one. More precisely, we practically observe that the marginal posterior $p(\theta|\y)$ typically has its maximum at a good value for $\theta$, but struggles to concentrate and spreads its mass across a much wider range of values of $\theta$. Consequently, $\theta \mapsto p(\theta|\y)$ fails to sufficiently penalise poor models, which are given too much weight in $x \mapsto p(\x|\y)$ as a result. In this situation, the pseudo-posterior $x \mapsto p(\x|\y,\thetaStar)$ often delivers better inferences than the marginal posterior $x \mapsto p(\x|\y)$. In the context of inverse problems, this phenomenon is particularly clear in problems that are poorly conditioned and where the misspecification of the prior has a stronger effect on the inferences. This behaviour is observed in all the imaging problems reported in \Cref{sec:experiments} and is particularly clear in the hyperspectral unmixing problem.

It is also worth mentioning at this point that there is another hierarchical Bayesian approach where $\x$ and $\theta$ are estimated jointly from $\y$, without marginalisation \cite{zibetti2008determining,EUSIPCO}. For example, one can perform joint MAP estimation
\begin{equation}
(\hat{x}_{\star},\hat{\theta}_{\star}) = \underset{\x \in \rset^{\dim}, \ \theta \in \Theta }{\mathrm{argmax}}~ p(\x,\theta|\y) \eqsp .
\label{EQ: HB_joint_MAP}
\end{equation}

This strategy has been studied in detail in \cite{calvetti2019hierachical,calvetti2020sparse} in the context of hierarchical Bayesian sparse regularisation  models. More precisely, these works cleverly exploit the conditional structure of certain hierarchical Gaussian models with random prior covariance matrices to propose a simple iterative alternating scheme to compute the joint MAP estimator of $x$ and the prior covariance. This kind of scheme yields good results for the class of imaging models in those works, both in terms of accuracy and computational complexity. The generalisation of the ideas of \cite{calvetti2019hierachical,calvetti2020sparse} to other imaging models, particularly the class of models considered in this paper, is very interesting but highly non-trivial. 
\section{Numerical experiments}
\label{sec:experiments}
In this section we validate the proposed methodology with a range of imaging inverse problems, which we have selected to illustrate a variety of observation models and regularisation functions. The first experiment, presented in \Cref{ssec:experiments.synthetic}, is carried out using synthetic images for which the exact generative statistical model is known, as this allows assessing the performance of the proposed method in a case where the regularisation parameter has a true value, and where there is no model misspecification. We also use this experiment to explore the robustness of the method towards mild likelihood misspecification (e.g., when there is a mismatch in the statistical properties of the noise).

Following on from this, in  \Cref{ssec:experiments.nat-img-deblur} we demonstrate the method by estimating a scalar-valued regularisation parameter in a non-blind image deconvolution model with different kinds of prior distributions, such as total variation and $\ell_1$-wavelet priors. This allows comparing our method to some state-of-the-art approaches that are limited to scalar-valued regularisation parameters. We also use one of these experiments to explain how to address problems in which the noise variance is unknown by jointly estimating $\theta$ and the variance of the noise by marginal MLE. 

This is then followed by two challenging problems involving multivariate regularisation parameters. In \Cref{ssec:experiments.hyperunmix} we apply our method to a sparse hyperspectral unmixing problem combining an $\ell_1$ and a total variation regularisation, and where we report comparisons with the hierarchical Bayesian approach of \cite{EUSIPCO}. Lastly, in \Cref{sssec:experiments.tgv-denois} we apply our method to a total generalised variation denoising model that has two unknown regularisation parameters exhibiting strong dependencies, and which requires using \Cref{algo:MCMC_double_chain} with two parallel Markov chains. 

In all the experiments we first compute $\thetaEB$, see \eqref{eq:thetaavg}, and then calculate a MAP estimator using the empirical Bayesian posterior $x \mapsto p(\x|y,\thetaEB)$ by convex optimisation (solver details are provided in each experiment). All experiments were carried out on an Intel  ${\textrm{i9-8950HK@2.90GHz}}$ workstation running MATLAB R2018a. In all experiments, $\theta$ was estimated on a logarithmic scale by using the change of variables discussed in Section \Cref{sssec:other-imp-guidelines}, except for the first experiment in \Cref{ssec:experiments.synthetic} where we used a linear scale.

\subsection{Performance on synthetic images - denoising}
\label{ssec:experiments.synthetic}

We first demonstrate the performance of the algorithm on a very simple image denoising problem, where we work with synthetic test images to have access to the true value of the regularisation parameter.  We consider a wavelet-based image denoising under a synthesis formulation where we assume that the coefficients $x$ of the true image  in an orthogonal 4-level Haar basis $\Psi$ follow a Laplace distribution. That is the model \eqref{EQ: posterior} is given for any $x \in \rset^{\dim}$ by $ f_\y(x)= \norm{\y-\Psi x}^2_2/(2\sigma^{2}) $ and $\g(x) = \Vnorm[1]{x}$.
In our experiments, $y$ has dimension $\dimy = 256 \times 256$ pixels, and we set $ \theta=1 $ to generate the synthetic test images. The variance of the added noise $\sigma^2$ is chosen for every case such that the signal-to-noise-ratio (\SNRname) is ${20 ~\text{$\mathrm{dB}$}}$, ${30 ~\text{$\mathrm{dB}$}}$, or ${40 ~\text{$\mathrm{dB}$}}$. In all cases we compute the empirical Bayes estimator $\thetaEB$  by implementing \Cref{algo:MCMC_single_chain} with the MYULA kernel \eqref{EQ: MYULA_explicit_posterior}.

To study the statistical behaviour of the method, we repeat each experiment 500 times by generating 500 random observations $\y$,  each one coming from a different random $\x$; then, for every observation $\y$,  we estimate $\thetaEB$. 
\def\widthfig{\linewidth}
\def\widthminipage{0.33333\linewidth}
\begin{figure}[!h]			
	\begin{minipage}[t]{\widthminipage}
		\centering
		\centerline{	\includegraphics[width=0.95\widthfig]{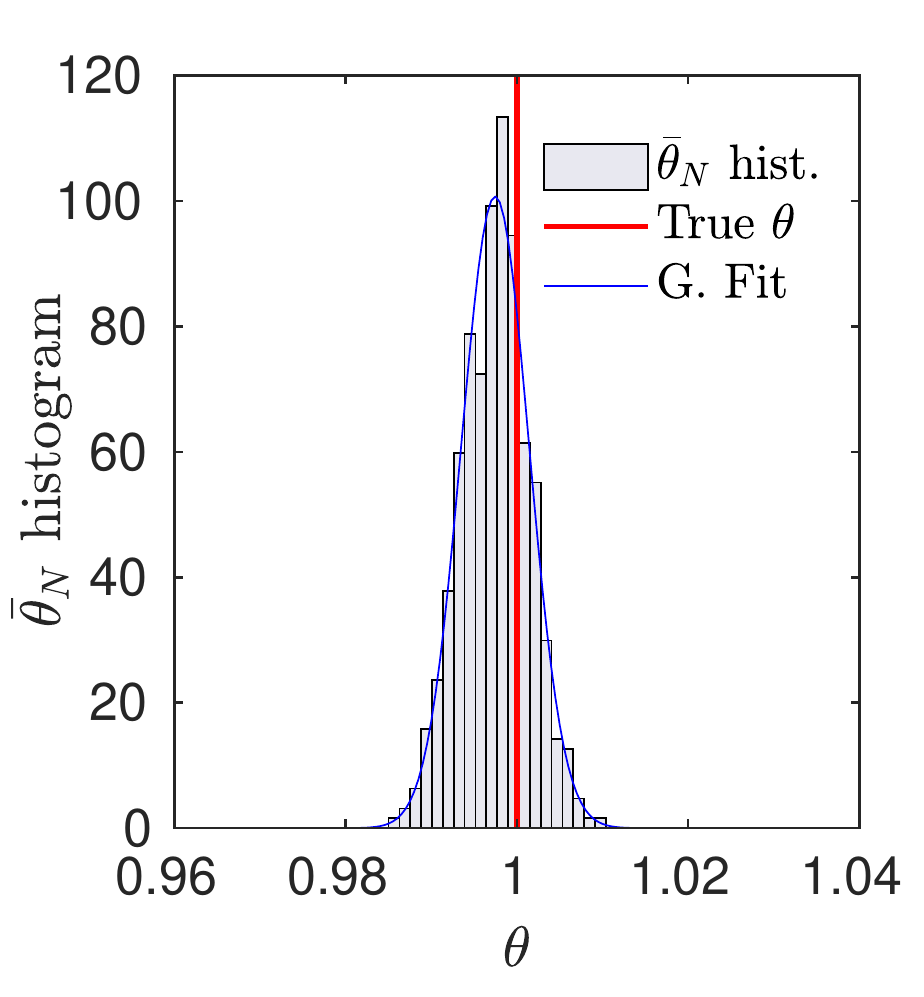}}
		\centerline{
			\includegraphics[width=\widthfig]{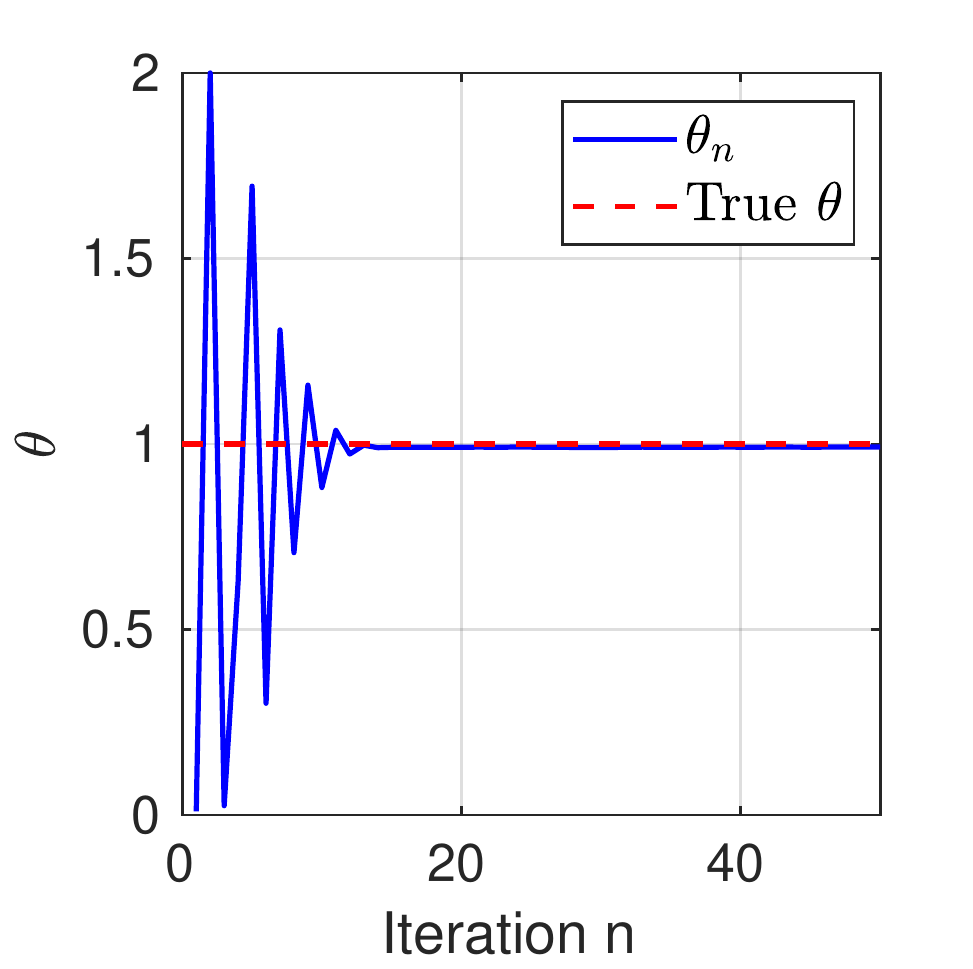}}
		\centerline{(a) \SNRname=20 $\mathrm{dB}$}
	\end{minipage}\hfill
	\begin{minipage}[t]{\widthminipage}
		\centering	
		\centerline{	
			\includegraphics[width=\widthfig]{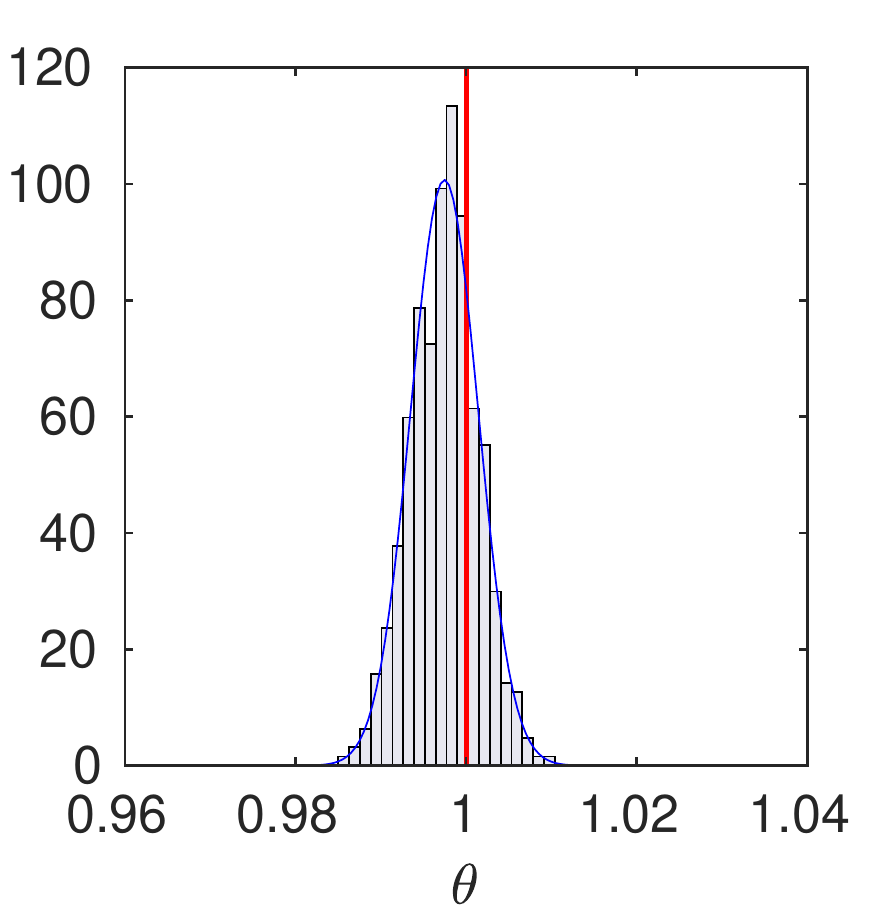}}
		\centerline{
			\includegraphics[width=\widthfig]{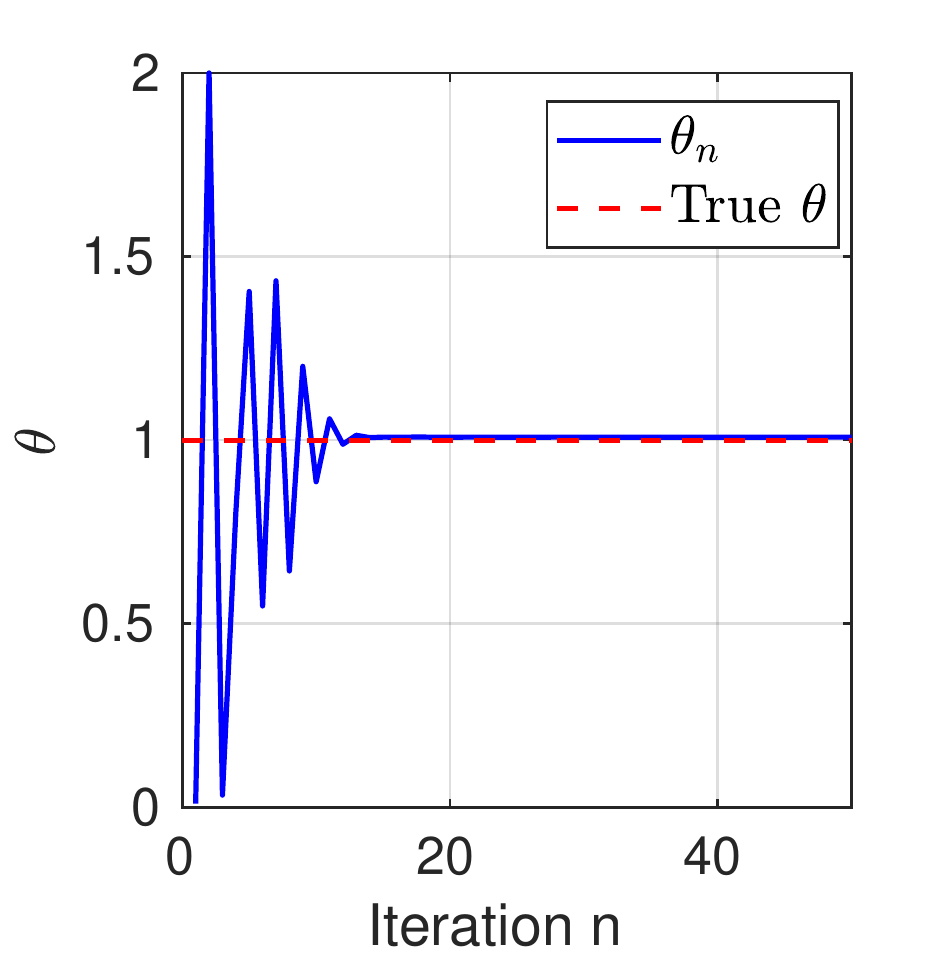}}
		\centerline{(b) \SNRname=30 $\mathrm{dB}$}
	\end{minipage}\hfill
	\begin{minipage}[t]{\widthminipage}
		\centering
		\centerline{	\includegraphics[width=\widthfig]{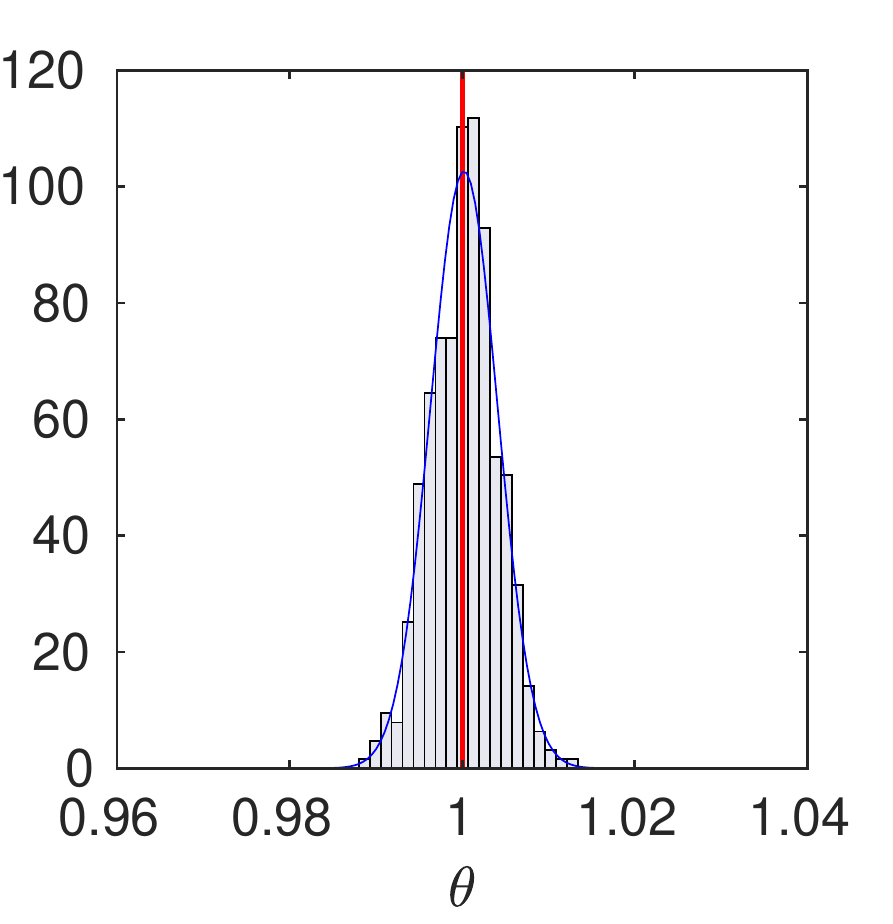}}
		\centerline{
			\includegraphics[width=\widthfig]{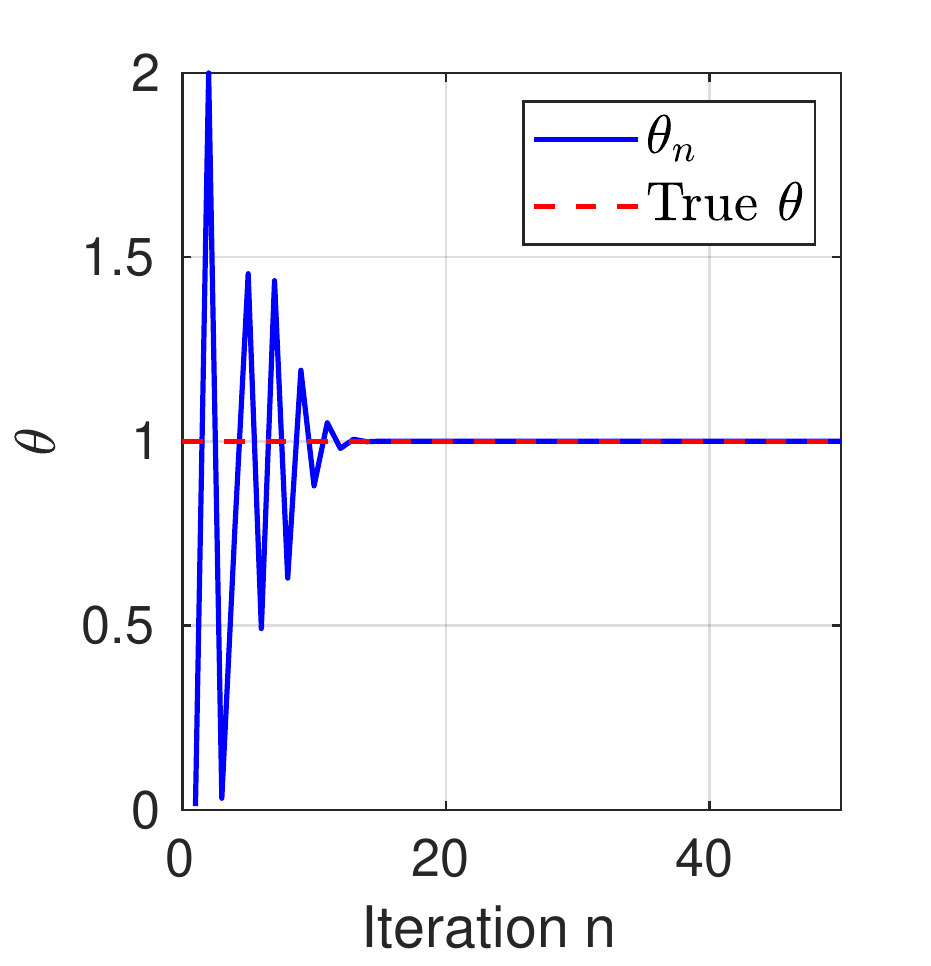}}
		\centerline{(c) \SNRname=40 $\mathrm{dB}$}
	\end{minipage}\hfill
	\caption{	\small		
		Denoising with synthesis-$\ell_1$ prior. Histograms of the estimated $\thetaEB$ for 500 repetitions and evolution of the iterates $(\theta_n)_{n \in \nset}$. Results for \SNRname \ of ${20 ~\text{$\mathrm{dB}$}}$, ${30 ~\text{$\mathrm{dB}$}}$ and ${40 ~\text{$\mathrm{dB}$}}$.\normalsize}
	\label{fig:res-num-synth-histograms}
\end{figure}

\Cref{fig:res-num-synth-histograms} shows the histograms obtained from the 500 estimated $\thetaEB$ values for each experiment (${20 ~\text{$\mathrm{dB}$}}$, ${30 ~\text{$\mathrm{dB}$}}$, and ${40 ~\text{$\mathrm{dB}$}}$). For completeness, we also present in \Cref{fig:res-num-synth-histograms} one example of a generated sequence of iterates $(\theta_n)_{n \in \nset}$ for each experiment. Observe that the estimation error is close to Gaussian and close to the true value of the regularisation parameter, as expected for a maximum likelihood estimator. The algorithm converges in approximately 15 iterations, possibly with some very small bias of the order of $0.1\%$.

To illustrate the robustness of the methods w.r.t. the
initialisation parameters $\theta_0$ and $\delta_0 = c_0$, we show
in \Cref{fig:res-num-synth-delta-c0} (a) and (b) the evolution of
the iterate $\theta_n$ for different values of $\theta_0$ and $c_0$,
respectively. Observe that all sequences converge to the same
solution, and that with an appropriate choice of $\Theta$ the scheme
is robust to the choice of $\theta_0$. Also observe that a poor
choice of $c_0$ can dramatically reduce convergence speed.
\def\widthfig{\linewidth} \def\widthminipage{0.48\linewidth}
\begin{figure}[!h]	
	\begin{minipage}[t]{\widthminipage}	
		\centering
		\centerline{\includegraphics[width=\linewidth,trim={0cm 0cm 0cm 0cm},clip]{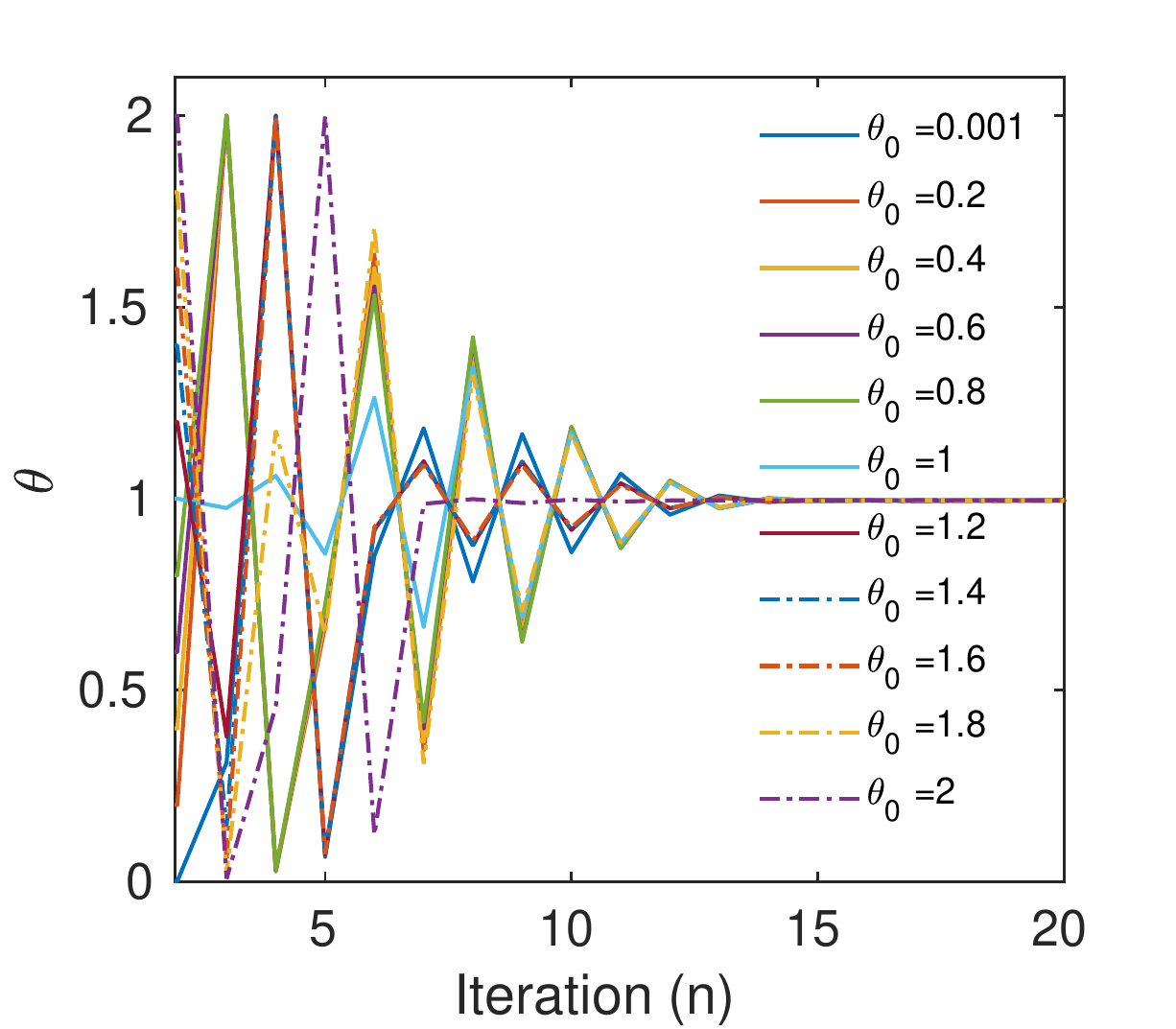}}
		\centerline{(a)}
	\end{minipage}\hfill
	\begin{minipage}[t]{\widthminipage}	
		\centerline{\includegraphics[width=\linewidth,trim={0cm 0cm 0cm 0cm},clip]{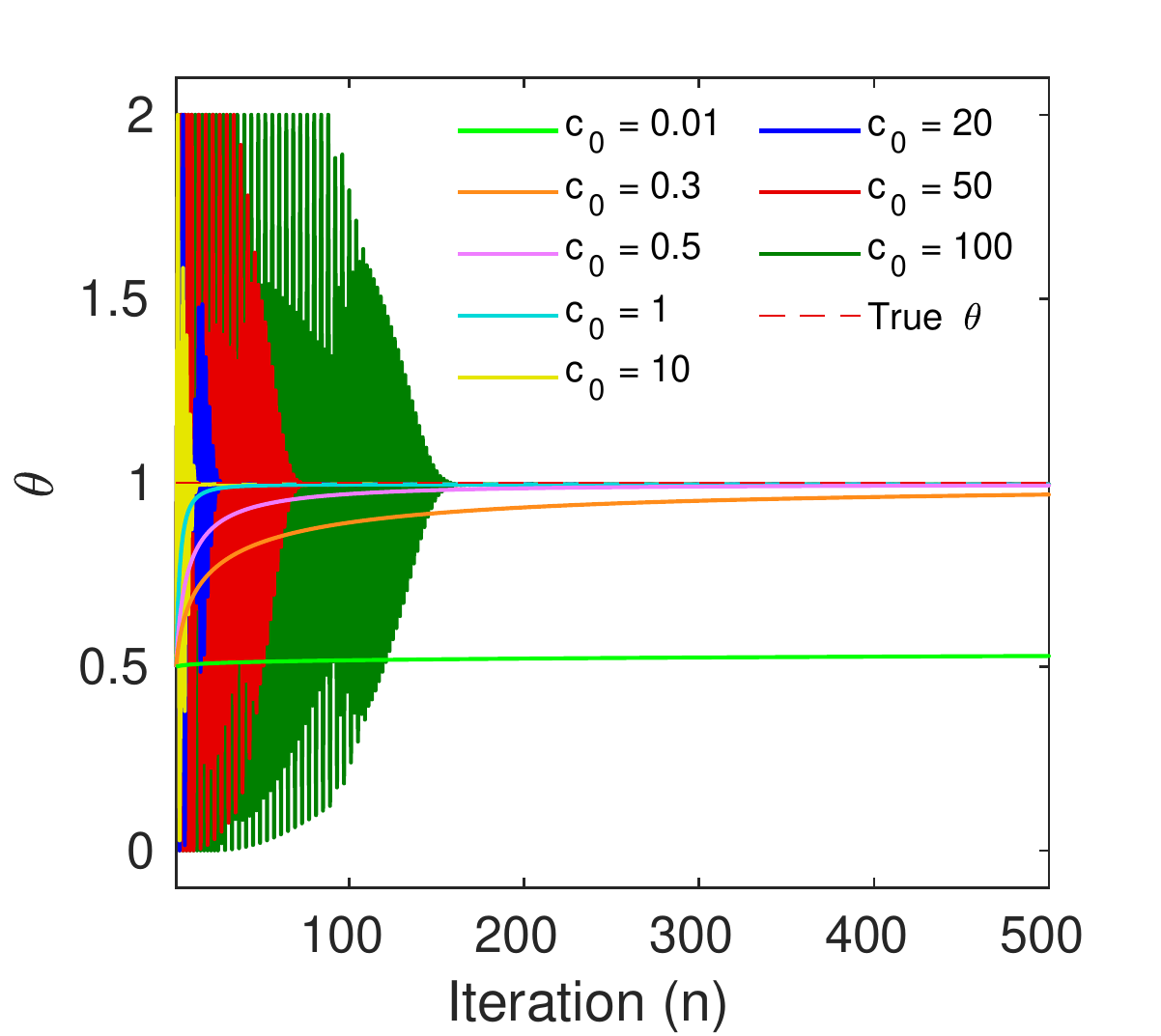}}
		\centerline{(b)}
	\end{minipage}\hfill
	\caption{	\small		
		Denoising with synthesis-$\ell_1$ prior for \SNRname=$20~\mathrm{dB}$. Evolution of the iterates $(\theta_n)_{n \in \nset}$ for different initialisation parameters: (a) $\theta_0$, (b) $\delta_0 = c_0$. The algorithms are robust to different initialisations, but a poor choice of $c_0$ can deteriorate the convergence speed.\normalsize}
	\label{fig:res-num-synth-delta-c0}
	\normalsize
\end{figure}

Moreover, to explore the behaviour of the method with other noise distributions, we repeat the previous experiment using Laplace noise instead of Gaussian noise. Since the Laplace distribution involves a non-smooth $\ell_1$ term, we adopt a proximal MCMC approach and implement the algorithms using the gradient of its $\lambda$-Moreau-Yosida envelope. The results are reported in \Cref{fig:res-num-synth-histograms-laplacenoise} (a).

\def\widthfig{\linewidth}
\def\widthminipage{0.31\linewidth}
\begin{figure}[!h]
	\begin{minipage}[t]{0.022\linewidth}			
		\begin{turn}{90}
			\hspace{-4.4cm}\small{	(b) Gaussian noise model	\hspace{1.1cm} (a) Laplace noise model}
		\end{turn}					
	\end{minipage}	\hfill	
	\begin{minipage}[t]{\widthminipage}
		\centerline{	\includegraphics[width=\widthfig,trim={0cm 0cm 0.29cm 0cm},clip]{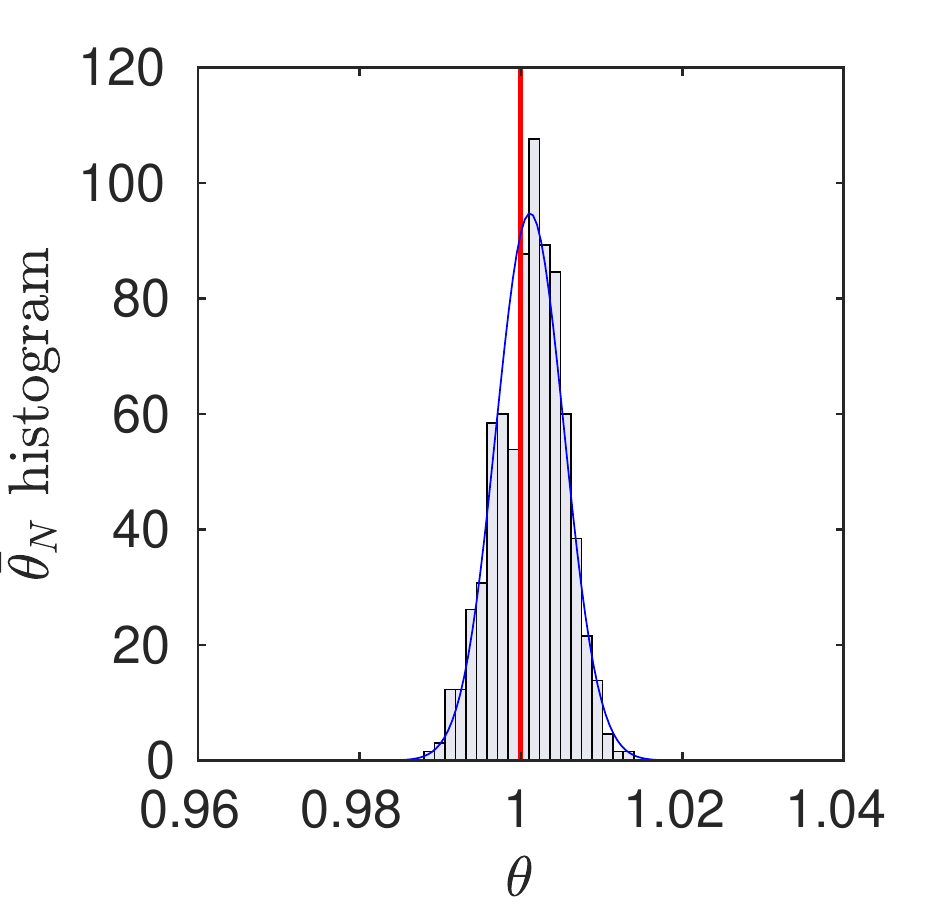}}
		\centerline{
			\includegraphics[width=\widthfig,trim={0cm 0cm 0.29cm 0cm},clip]{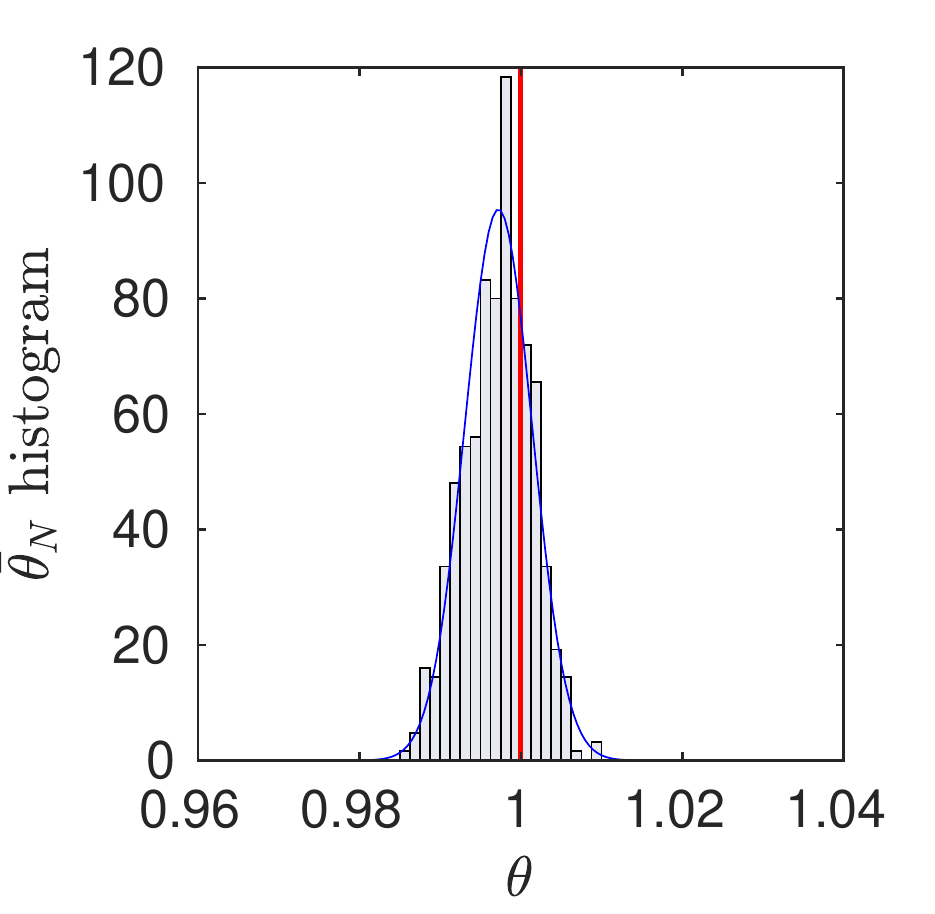}}
		\centerline{ \SNRname=20 $\mathrm{dB}$}		
	\end{minipage}\hfill
	\begin{minipage}[t]{\widthminipage}
		\centering	
		\centerline{	
			\includegraphics[width=\widthfig,trim={0cm 0cm 0.29cm 0cm},clip]{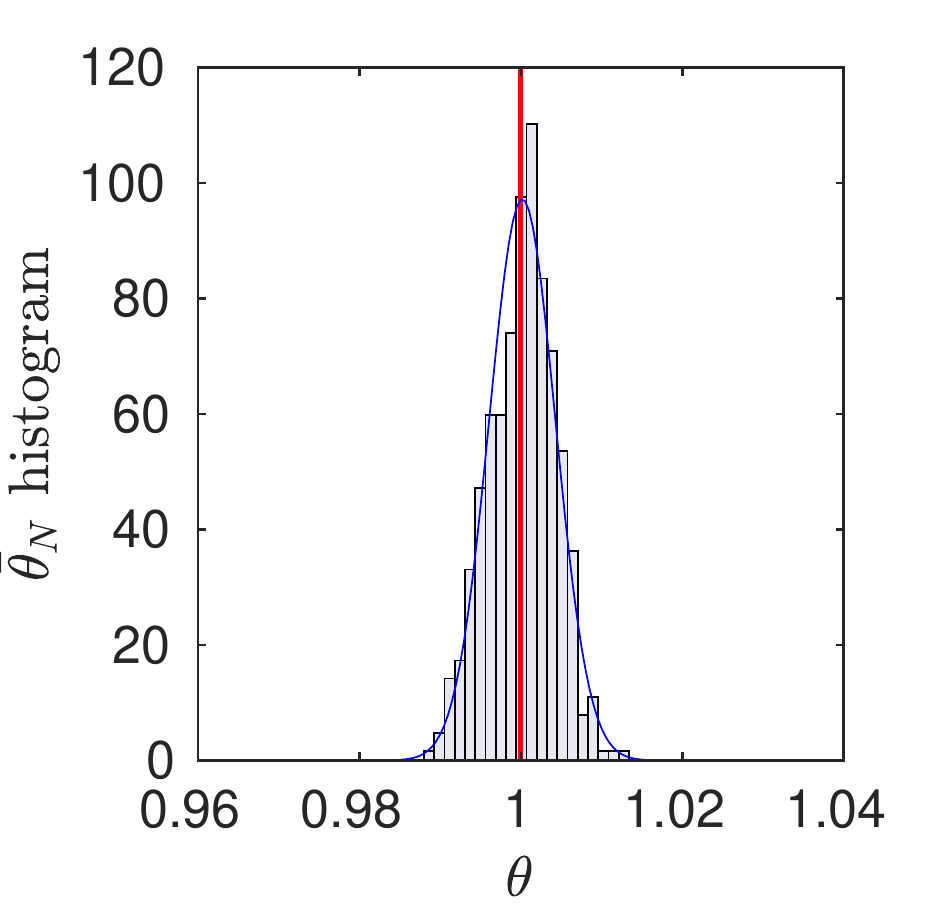}}
		\centerline{
			\includegraphics[width=\widthfig,trim={0cm 0cm 0.29cm 0cm},clip]{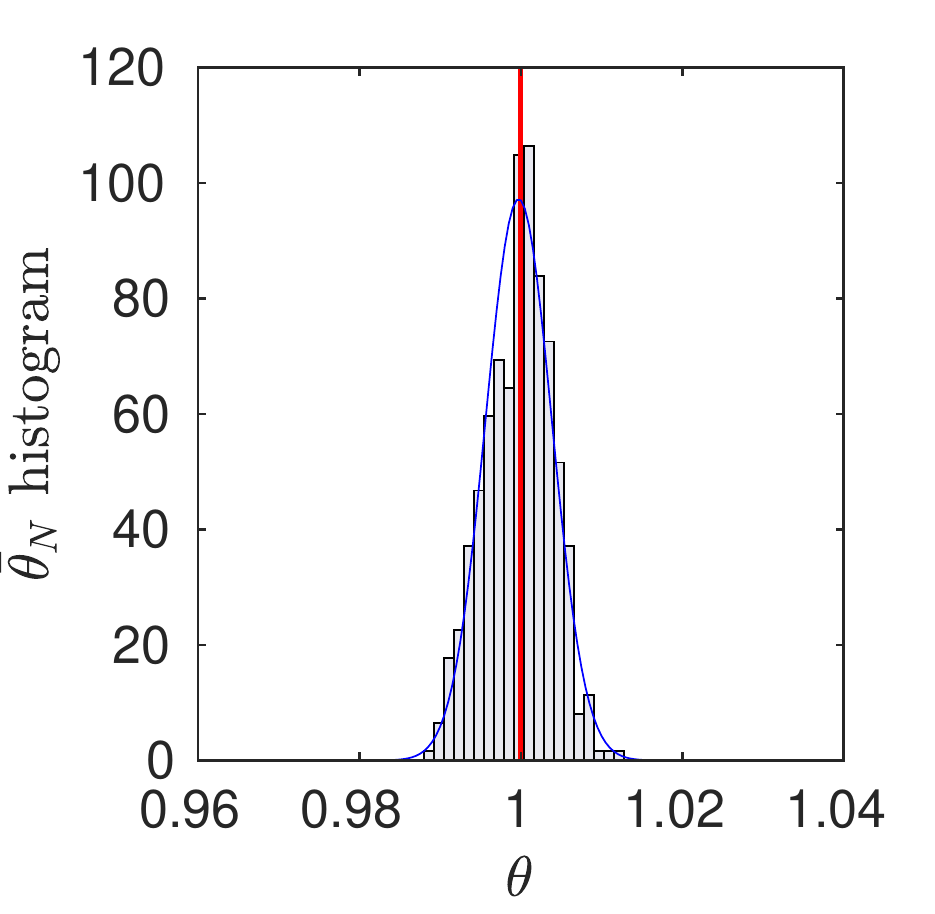}}		
		\centerline{ \SNRname=30 $\mathrm{dB}$}
	\end{minipage}\hfill
	\begin{minipage}[t]{\widthminipage}
		\centering
		\centerline{	\includegraphics[width=\widthfig,trim={0cm 0cm 0.29cm 0cm},clip]{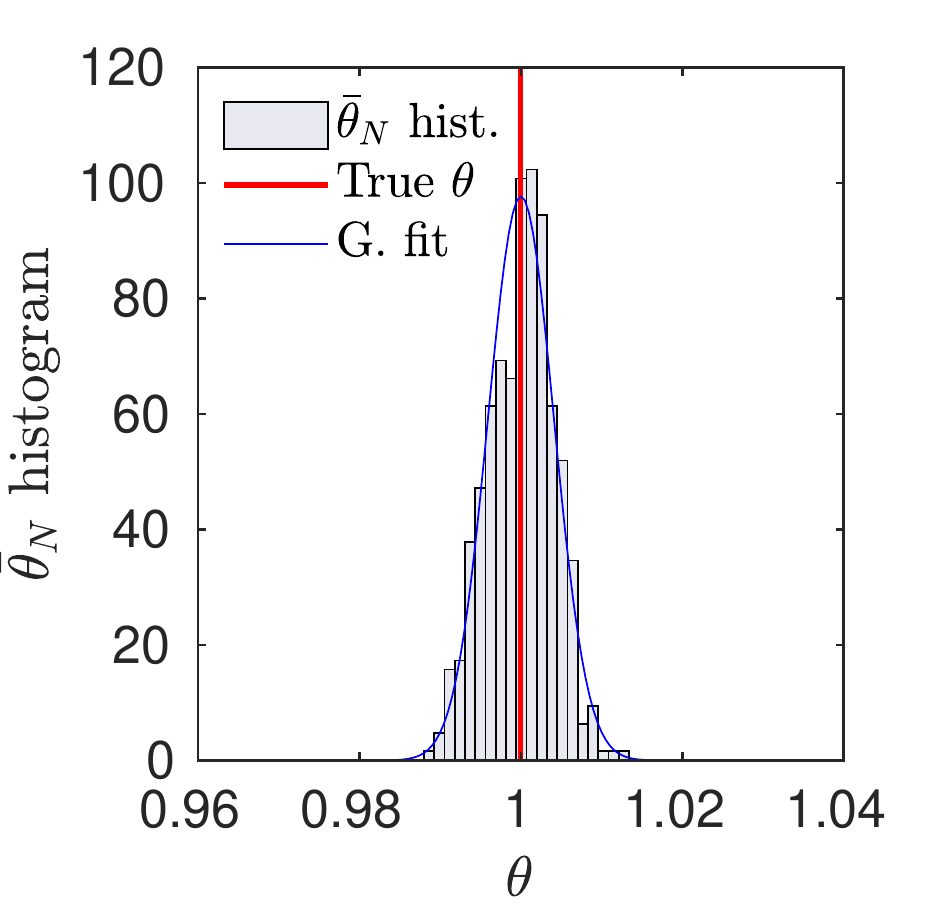}}
		\centerline{
			\includegraphics[width=\widthfig,trim={0cm 0cm 0.29cm 0cm},clip]{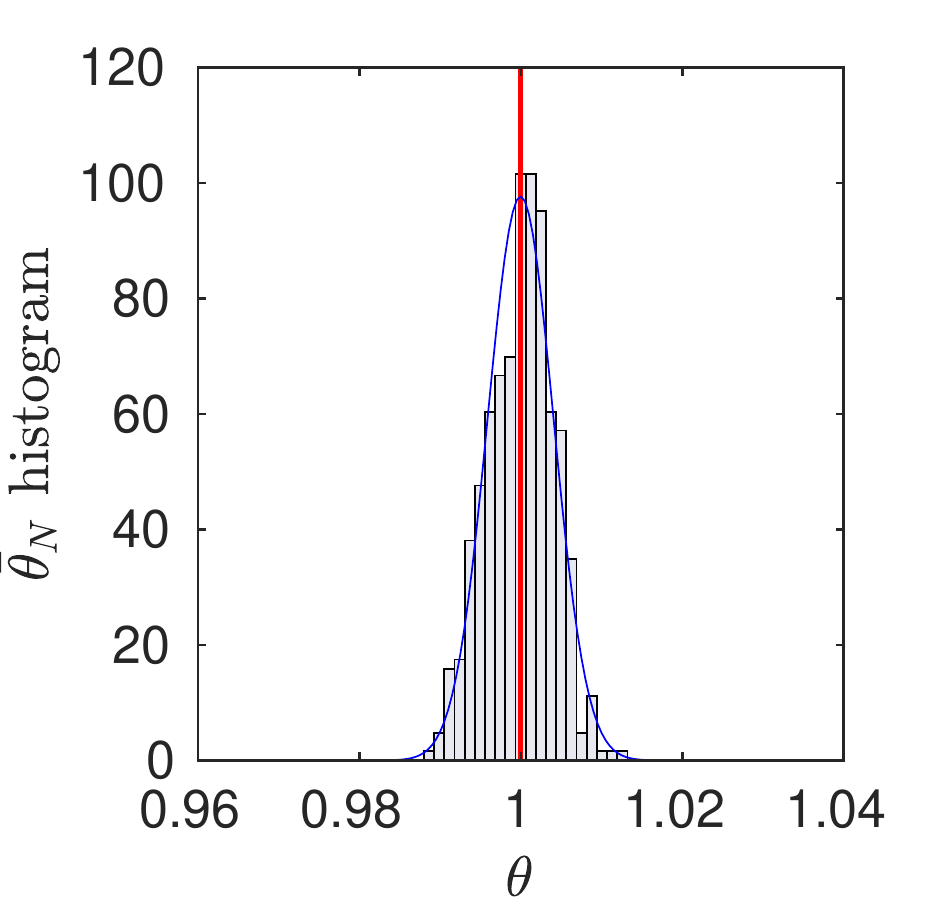}}
		\centerline{ \SNRname=40 $\mathrm{dB}$}
	\end{minipage}\hfill
	\caption{	\small		
		Denoising with synthesis-$\ell_1$ prior with Laplace noise. Histograms of the estimated $\thetaEB$ for 500 repetitions of the empirical Bayes algorithm using (a) correct Laplace noise model and (b)~incorrect Gaussian noise model. Results for \SNRname \ of ${20 ~\text{$\mathrm{dB}$}}$, ${30 ~\text{$\mathrm{dB}$}}$ and ${40 ~\text{$\mathrm{dB}$}}$.\normalsize}
	\label{fig:res-num-synth-histograms-laplacenoise}
\end{figure}

Lastly, to explore the robustness of the method towards mild misspecification of the likelihood, we have also repeated the same experiment with Laplace noise but using an incorrectly specified Gaussian noise model. These results are reported in \Cref{fig:res-num-synth-histograms-laplacenoise} (b), where we see that the method is robust to mild likelihood misspecification, as the differences between using a correctly specified likelihood or an incorrectly Gaussian likelihood only become noticeable for the lowest SNR of 20~dB.

\subsection{Non-blind natural image deconvolution}
\label{ssec:experiments.nat-img-deblur}
We now illustrate the proposed methodology with an application to image deblurring using two different kinds of prior distributions: the total variation (\tvname) prior and a wavelet-based synthesis-$\ell_1$ prior. For comparison, we also report the results obtained with \SUGARname \ \cite{sugar2014stein} (only when using a \tvname \ prior), joint MAP estimation \cite{EUSIPCO}, discrepancy principle \cite{engl1996regularization,morozov2012methods}, and by using the oracle value $\thetaMSE$ that minimises the estimation mean squared error (\MSEname), \ie \
\begin{equation}
\thetaMSE = \argmin_{\theta \in \Theta} \defEns{\norm{\xGT - \argmax_{x \in \rset^{\dim}} p(x | y, \theta)}_2} \eqsp ,
\end{equation}
where $\xGT$ is the ground truth.
We want to highlight that carrying out such a comparison is not a trivial task because some algorithms are solver-dependent while some others are completely independent of the solver used to compute the MAP estimator. For this reason, the comparison was done with extreme care, and we include a detailed explanation of how we compare the results in \Cref{append:method-comparison}.

In non-blind image deblurring, the aim is to recover an unknown image $ \x \in \rset^d $ from a blurred and noisy observation $ \y=\rmA\x+\rmw $, where $\rmA \in \rset^d \times \rset^d$ is a blur matrix, and $\rmw$ is a {$\dim\textrm{-dimensional}$} Gaussian random variable with zero mean and covariance matrix $\sigma \Id$ with ${\sigma >0}$. In our experiments, $x$ and $y$ are of size $\dim = 512 \times 512$ pixels, $\rmA$ is a known circulant uniform blur of size $9 \times 9$ pixels, and $\sigma^2$ is chosen such that the blurred signal-to-noise-ratio (\SNRname) is ${20 ~\text{$\mathrm{dB}$}}$, ${30 ~\text{$\mathrm{dB}$}}$, or ${40 ~\text{$\mathrm{dB}$}}$.  We perform all experiments on ten standard test images (\small{\texttt{barbara, boat, bridge,  flintstones, goldhill, lake, lena, man, mandrill and wheel}}). \normalsize

For each image, noise level, and $\theta$ selection method, we first obtain an estimate for $\theta$ and then use it to compute the MAP estimator $\hat{\x}_{\MAP}$ (given by \eqref{EQ: mapEstim}). In the case of the joint MAP method \cite{EUSIPCO}, we carry out joint MAP estimation of $\theta$ and $\hat{\x}_{\MAP}$. We compute the MAP estimator by using a highly efficient proximal convex optimisation algorithm, SALSA \cite{salsa2010fast}, which is an instance of Alternative Direction Method of Multipliers (ADMM). We then assess the resulting performance by computing the \MSEname \ between the \MAPname \ estimator and the ground truth. 

\subsubsection{Deconvolution with Total Variation prior}
\label{sssec:experiments.TV-deblur}

In this experiment we use model \eqref{EQ: posterior} where for any $x \in \rset^d$ we have $ \f(\x) =  \norm{\y-\rmA\x}^2_2/2\sigma^{2}$,  ${\g(\x) = \mathrm{TV}(\x)}$, and follow the previously explained procedure.   Here $\mathrm{TV}(\x)$  is the isotropic total variation pseudo-norm  given by  $\mathrm{TV}(\x)=\sum_i\sqrt{(\Delta_i^h \x)^2+(\Delta_i^v \x)^2}$ where $\Delta_i^v$ and $\Delta_i^h$ denote horizontal and vertical first-order local difference operators. 
To compute $\thetaEB$ we use \Cref{algo:MCMC_single_chain}. The prior associated with the total variation pseudo-norm is not proper
, so the effective dimension is $\dim-1$. We evaluated the proximal operator of $\mathrm{TV}(\x)$ using the  primal-dual algorithm from \cite{chambolle2011first}  with 25 iterations.  

The algorithm parameters are chosen following the recommendations provided in \Cref{append:ssec:setting-algo-parameters}; we consider $300$ warm-up iterations and set  $\btheta_0=0.01$, $X_0^0=y$, $m_n=1$ and $\delta_n = 0.1 \times n^{- 0.8}/\dim$ for any $n \in \nsets$;  we set  $\lambda = \min\left( 5 \Ly^{-1}, \lambda_{\max}\right)$ with $ \lambda_{\max}=2 $ and $\Ly=(0.99/\sigma)^2$, and ${\gamma = 0.98 \times (\Ly +1/\lambda)^{-1}}$.  As suggested in \Cref{append:ssec:setting-algo-parameters}, we set $(\omega_n)_{n \in \nset}$ to have  $N_0 = 25$ burn-in iterations and compute $\thetaavg_N$ using  \eqref{eq:thetaavg}.

In addition, instead of setting a fixed number of iterations, we stop the algorithm when the relative change $\absLigne{\thetaavg_{N+1} - \thetaavg_{N}}/\thetaavg_N$ is smaller than $10^{-3}$.  It would be possible to use a tolerance of $10^{-5}$ and get a slight improvement of the \MSEname \ ($<0.02$ $\mathrm{dB}$), but this would lead to computing times that are five times longer. We use SALSA with the following parameters: ${\texttt{inneriters} = 1}$, $\texttt{outeriters}=500$, $\texttt{tol}=10^{-5}$ and  ${\texttt{mu}=\thetaEB/10}$.

\newcommand{\snrTV}{30}
\def\widthfig{\linewidth}
\def\widthminipage{0.24\linewidth}
\newcommand\incGoldhillFull[1]{\includegraphics[width=\widthfig,trim={1pt 1pt 1pt 1pt},clip]{#1}}
\newcommand\incManFull[1]{\includegraphics[width=\widthfig,trim={1pt 1pt 1pt 6pt},clip]{#1}}
\begin{figure}[!htb]
	\centering	
	\begin{minipage}[t]{\widthminipage}
		\centering		
		\centerline{\incManFull{img/num-res-deblurTV/snr\snrTV/EB/man_degraded}\vspace{0.1cm}}
		
		\centerline{(a) Degraded}
	\end{minipage}	
	\begin{minipage}[t]{\widthminipage}
		\centering
		\centerline{\incManFull{img/num-res-deblurTV/snr\snrTV/EB/man_estim}\vspace{0.1cm}}
		
		\centerline{(b) EB}
	\end{minipage}
	\begin{minipage}[t]{\widthminipage}
		\centering			
		\centerline{\includegraphics[width=\widthfig,trim={0pt 0pt 0pt 5pt},clip]{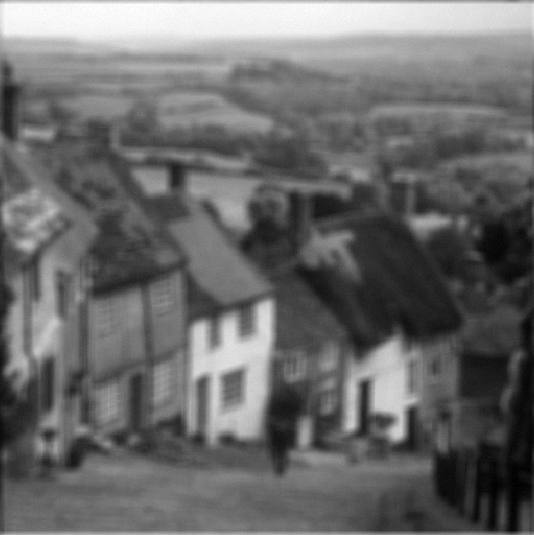}}
		\centerline{(a) Degraded}
	\end{minipage}	
	\begin{minipage}[t]{\widthminipage}
		\centering		
		\centerline{\incGoldhillFull{img/num-res-deblurTV/snr\snrTV/EB/goldhill_estim_clamp}}
		\centerline{(b) EB}
	\end{minipage}	
	\caption{	\small		
		Deblurring with \tvname \ prior for \texttt{man} and \texttt{goldhill} test images: (a) blurred and noisy (\SNRname=\snrTV~$\mathrm{dB}$) observation $ \y $, (b) \MAPname \ estimator obtained using $\thetaavg_N $ computed with empirical Bayes.\normalsize}
	\label{fig:num-res-TV-xMAP-full}
\end{figure}

For illustration, \Cref{fig:num-res-TV-xMAP-full} shows the results obtained for two of the test images (\texttt{man}  and \texttt{goldhill}) using the proposed method. The displayed images correspond to the ${30 ~\text{$\mathrm{dB}$}}$ \SNRname \ setup. In \Cref{fig:num-res-TV-xMAP} we compare the \MAPname \  estimates obtained by using each of the  considered methods. In this case we display a close-up on \texttt{man}  and \texttt{goldhill} selecting a region that contains fine details and sharp edges.   In   \Cref{fig:num-res-TV-MSEplot} and \Cref{fig:num-res-TV-theta}  we provide further details for the same two images, showing a plot of the \MSEname \ obtained with each method and the evolution of the iterates $ (\theta_n)_{n \in \nset} $ for the empirical Bayesian method.

\def\widthfig{0.99\linewidth}
\def\widthminipage{0.165\linewidth}	
\begin{figure}[!h]\centering		
	\newcommand\incGoldhill[1]{\includegraphics[width=\widthfig,trim={2cm 2cm 5cm 5cm},clip]{#1}\vspace{0.1cm}}
	\newcommand\incnes[1]{\includegraphics[width=\widthfig,trim={7cm 11.2cm 5cm 0.8cm},clip]{#1}}
	\newcommand\incMan[1]{\includegraphics[width=\widthfig,trim={5cm 11.2cm 7cm 0.8cm},clip]{#1}}
	\begin{minipage}[t]{\widthminipage}
		\centering
		\centerline{\incMan{img/num-res-deblurTV/snr\snrTV/EB/man_orig}}
		\centerline{\incGoldhill{img/num-res-deblurTV/snr\snrTV/EB/goldhill_orig}}
		\centerline{(a) Original}
	\end{minipage}\hfill
	\begin{minipage}[t]{\widthminipage}
		\centering		
		\centerline{\incMan{img/num-res-deblurTV/snr\snrTV/EB/man_degraded}}
		\centerline{\incGoldhill{img/num-res-deblurTV/snr\snrTV/EB/goldhill_degraded}}		
		\centerline{(b) Degraded}
	\end{minipage}\hfill
	\begin{minipage}[t]{\widthminipage}
		\centering
		\centerline{\incMan{img/num-res-deblurTV/snr\snrTV/EB/man_estim}}
		\centerline{\incGoldhill{img/num-res-deblurTV/snr\snrTV/EB/goldhill_estim_clamp}}
		\centerline{(c) EB}
	\end{minipage}\hfill
	\begin{minipage}[t]{\widthminipage}
		\centering
		\centerline{\incMan{img/num-res-deblurTV/snr\snrTV/HB/man_estim}}
		\centerline{\incGoldhill{img/num-res-deblurTV/snr\snrTV/HB/goldhill_estim}}
		\centerline{(d) HB}
	\end{minipage}\hfill
	\begin{minipage}[t]{\widthminipage}
		\centering
		\centerline{\incMan{img/num-res-deblurTV/snr\snrTV/discPrinciple/man_estim}}
		\centerline{\incGoldhill{img/num-res-deblurTV/snr\snrTV/discPrinciple/goldhill_estim}}
		\centerline{(e) DP}
	\end{minipage}\hfill
	\begin{minipage}[t]{\widthminipage}
		\centering
		\centerline{\incMan{img/num-res-deblurTV/snr\snrTV/SUGAR/man_estim}}
		\centerline{\incGoldhill{img/num-res-deblurTV/snr\snrTV/SUGAR/goldhill_estim}}
		\centerline{(f) SUGAR}
	\end{minipage}			
	\caption{	\small		
		Deblurring with \tvname \ prior. Close-up on \texttt{man} and \texttt{goldhill} test images: (a) True image $ \x $, (b) blurred and noisy (\SNRname=\snrTV~$\mathrm{dB}$) observation $ \y $, (c)-(f) \MAPname \ estimators obtained through empirical Bayes, hierarchical Bayes, discrepancy principle and \SUGARname \ methods,  respectively.\normalsize}
	\label{fig:num-res-TV-xMAP}
\end{figure}

\begin{figure}[!h]	\centering		
	\begin{minipage}[l1]{.48\linewidth}
		\centering
		\centerline{\includegraphics[width=\textwidth]{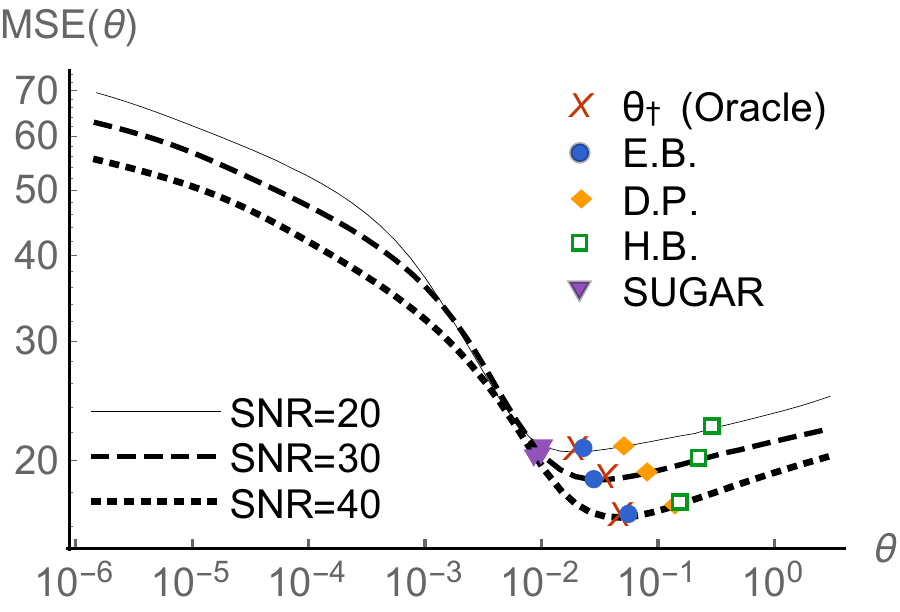}}		
		\centerline{(a) \texttt{man}}
	\end{minipage}
	\begin{minipage}[l2]{.48\linewidth}
		\centering
		\centerline{\includegraphics[width=\textwidth]{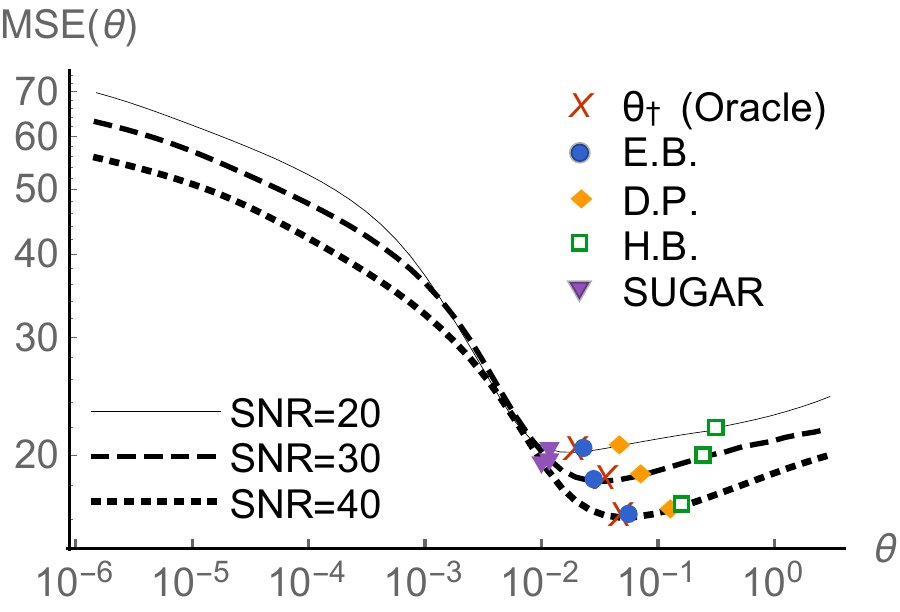}}		
		\centerline{(b) \texttt{goldhill}}	
	\end{minipage}
	\caption{\small Deblurring with \tvname \ prior. Mean squared error (\MSEname) obtained for (a)~\texttt{man}  and (b)~\texttt{goldhill} for different values of $\theta$. We compare the values obtained with empirical Bayes, discrepancy principle, hierarchical Bayes, \SUGARname, and  the optimal value $\thetaMSE$ that minimises the \MSEname.  }
	\label{fig:num-res-TV-MSEplot}
	\normalsize
\end{figure}

\begin{figure}[!h]			
	\begin{minipage}[l1]{.49\linewidth}
		\centering
		\centerline{\includegraphics[width=0.9\textwidth]{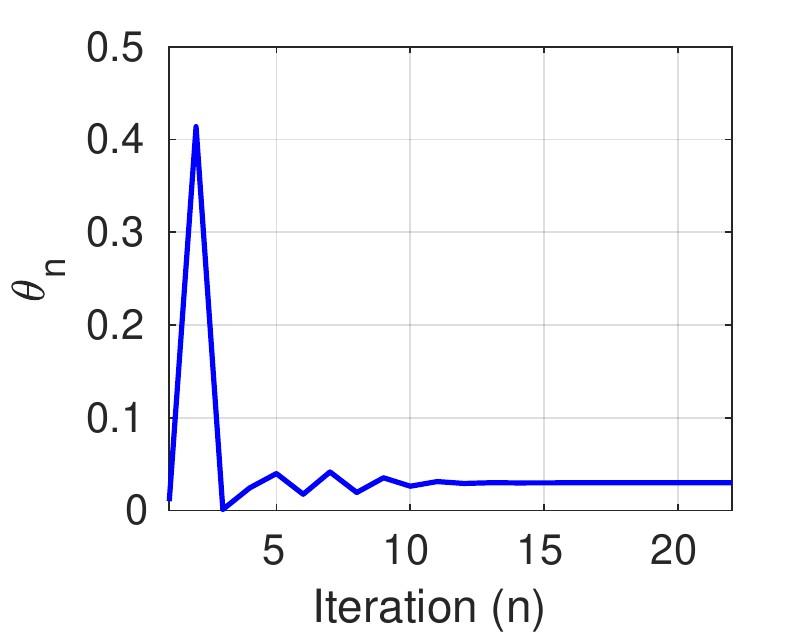}}		
		\centerline{(a) \texttt{man}}
	\end{minipage}\hfill
	\begin{minipage}[l2]{.49\linewidth}
		\centering
		\centerline{\includegraphics[width=0.9\textwidth]{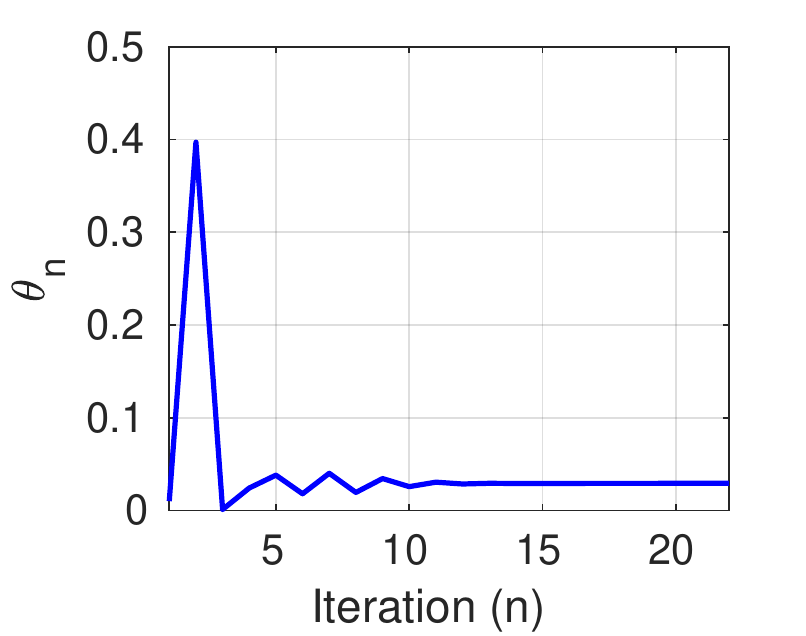}}		
		\centerline{(b) \texttt{goldhill}}	
	\end{minipage}
	\caption{\small Deblurring with \tvname \ prior. Evolution of the sequence of iterates $(\theta_n)_{n \in \nset}$ for the proposed method for \texttt{man} and \texttt{goldhill} test images (\SNRname=\snrTV~$\mathrm{dB}$).  }
	\label{fig:num-res-TV-theta}
	\normalsize
\end{figure}
Observe in \Cref{fig:num-res-TV-MSEplot} that the proposed empirical Bayesian algorithm yields close-to-optimal results, both for high and low \SNRname \ values. The method based on the discrepancy principle and the hierarchical Bayesian method overestimate the amount of regularisation required. Conversely, \SUGARname \ underestimates $\theta$ (this can also be observed in the recovered image in \Cref{fig:num-res-TV-xMAP} (f), where the MAP estimate presents some ringing artefacts due to high-frequency noise amplification); this is in agreement with the results reported in \cite{lucka2018risk}.

\Cref{table_deblur_TV} reports the average \MSEname \ values and average computing times obtained for each method.  
We can see that the proposed method performs close to the oracle performance, generally outperforming the other approaches from the state of the art with  very competitive computing times. In particular, observe that the proposed method performs remarkably for all \SNRname \ values. At high \SNRname \ values (${40 ~\text{$\mathrm{dB}$}}$) discrepancy principle and  joint MAP \cite{EUSIPCO} perform similarly, whereas for low \SNRname \ values (${20 ~\text{$\mathrm{dB}$}}$)  discrepancy principle outperforms joint MAP. Also, \SUGARname \ performs well for low \SNRname, but fails to find good values of $\theta$ when the \SNRname \ is higher. This might be due to the fact that \SUGARname \ minimises a surrogate of the \MSEname \ that works well for denoising but degrades in problems that are ill-posed or ill-conditioned. 
We emphasise at this point that the exact computing times of each algorithm depend on the specific stopping criteria and implementation details, so rather than claiming that one method is faster than the others, what we wish to illustrate is that the computing times are all within the same order of magnitude, with \SUGARname \ being moderately slower for this particular experiment. As we mentioned before, if we had selected a tolerance of $10^{-5}$ to stop our algorithm, the computing times would have increased with almost negligible changes in the \MSEname. Also note that we compute the optimal $\theta$ for the discrepancy principle method by continuation, but one could also use a different proximal splitting strategy (see  \cite{combettes2011proximalsplitting} for instance).

\begin{table}[!ht]	
	\renewcommand{\arraystretch}{1.2}
	\begin{minipage}[!t]{.59\linewidth}
		{\small
			\centering
			\begin{tabular}{rcccccc}\hline		
				Method & \multicolumn{2}{c}{\SNRname=20 $\mathrm{dB}$}       & \multicolumn{2}{c}{\SNRname=30 $\mathrm{dB}$}     & \multicolumn{2}{c}{\SNRname=40 $\mathrm{dB}$}                               \\
				&\MSEname &Time& \MSEname   &Time & \MSEname &Time\\
				\hline
				$ \thetaMSE$ & 23.29      &           & 21.39           &       & 19.06            &              \\
				E.B.   & 23.50        & 0.84         & 21.45            & 0.85        & 19.24     & 0.85       \\			
				D.P.   & 23.73         & 0.70	           & 21.87         &    1.52	         & 19.78          & 3.92 \\
				H.B.   & 25.07                   & 0.58                           & 22.84                   & 1.27                           & 19.84                   & 3.27                           \\
				\SUGARname  & 23.66                    & 3.64                           & 23.16                   & 5.00                          & 23.05                   &  5.63  \\   
				\hline                     
			\end{tabular}\vspace{0.4cm}
		}	
		
		\normalsize	
		\centerline{(a)}
	\end{minipage}\hfill	
	\begin{minipage}[!t]{.39\linewidth}
		\vspace{.4cm}
		\centering
		\includegraphics[width=\textwidth]{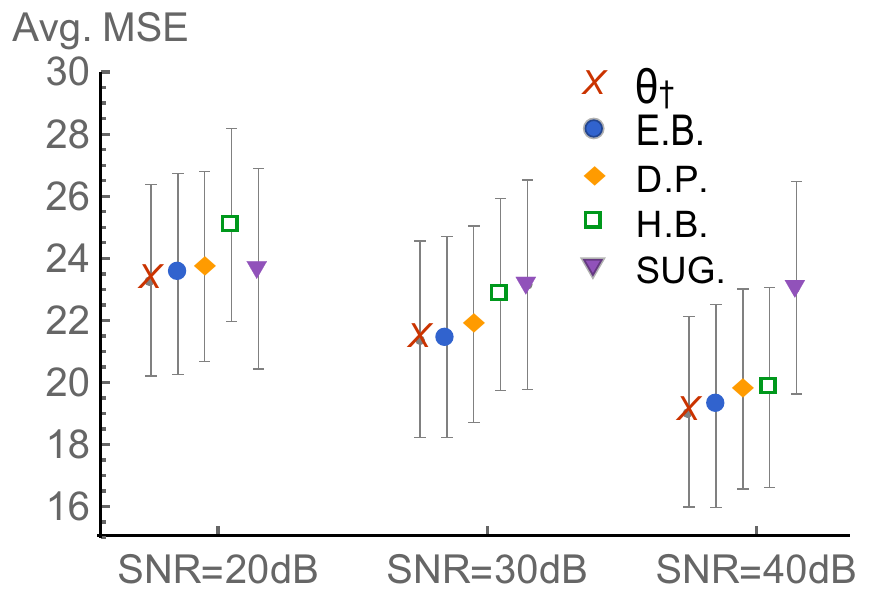}
		\centerline{(b)}
	\end{minipage}
	\caption{\small Deblurring with \tvname \ prior. (a) Table with average mean squared error obtained for ten images with different algorithms. Average execution times expressed in minutes. In (b) we  summarise the content of the table and show the standard deviation with error bars.}\normalsize	
	\label{table_deblur_TV}						
\end{table}		
\normalsize
\subsubsection{Deconvolution with Total Variation prior and unknown noise variance}
\label{sssec:experiments.TV-deblur-unknown-sigma}
In this section we consider the  same experiment as in \Cref{sssec:experiments.TV-deblur}, but we now suppose that the noise variance is unknown and explain how to modify our methodology to estimate this quantity jointly with $\theta$ by marginal MLE. This is beyond the scope of the theoretical results we present in our companion paper \cite{vidal:et:al:2019b}. However, we believe that the theory could be generalised to provide some (albeit weaker) guarantees for this case and other blind and semi-blind problems, and this is an important perspective for future work.

More precisely, we can use the proposed scheme to compute 
\begin{equation}
(\thetaStar,\sigma^2_\star) \in \underset{\theta \in \Theta, \, \sigma^2 \in {[\sigma^2_{min},\sigma^2_{max}]}}{\mathrm{argmax}}~ p(\y|\theta,\sigma^2)\,,
\end{equation}	
where $0 < \sigma^2_{min} < \sigma^2_{max} < \infty$ define a minimum and maximum admissible variance values. To obtain an estimate of $\frac{\textrm{d}}{\textrm{d}{\sigma^2}} \log p(\y|\theta,\sigma^2)$ in \Cref{algo:MCMC_single_chain} we differentiate $\log p(\x,\y|\theta,\sigma^2)$ w.r.t. $\sigma^2$ and obtain
\begin{equation}
\frac{\textrm{d}}{\textrm{d}{\sigma^2}}  \log p(\x,\y|\theta,\sigma^2) =  \dfrac{\norm{\y-\rmA\x}^2_2}{2(\sigma^2)^2}-\dfrac{\dim}{2\sigma^2}\, .
\end{equation}	
We summarise the resulting scheme for jointly estimating $\theta$ and $\sigma^2$ in \Cref{algo:MCMC_single_chain_uknown_sigma} below. 
\begin{algorithm}
	\caption{SAPG algorithm - Scalar $\theta$ and unknown noise variance $\sigma^2$}
	\label{algo:MCMC_single_chain_uknown_sigma}
	\begin{algorithmic}[1]						
		\STATE Input: initial $\{\theta_0, X_0^0\}$, $(\delta_n,\delta'_n,\omega_n,m_n)_{n \in \nset}$, $ \Theta $,  kernel parameters $\gamma,\lambda$, iterations $N$.
		\FOR{$n = 0$ to $N-1$}
		\IF{$n>0$}
		\STATE Set
		$X_0^n = X_{m_{n-1}}^{n-1}$, 
		\ENDIF
		\FOR{$k = 0$ to $m_n-1$}
		\STATE Sample $X^n_{k+1} \sim \Rker_{\gamma,\lambda,\theta_n,\sigma_n^2}(X^n_{k}, \cdot)$,
		\ENDFOR
		\STATE Set $\theta_{n+1} = \Pi_{\Theta}\left[\theta_n + \frac{\delta_{n+1}}{m_n} \sum_{k=1}^{m_n} \defEns{ \frac{\dim}{\alpha \theta_n} - \g(X_k^{n})}\right]$.
		\STATE Set $\sigma^2_{n+1} = \Pi_{[\sigma^2_{min},\sigma^2_{max}]}\left[\sigma^2_n + \frac{\delta'_{n+1}}{m_n} \sum_{k=1}^{m_n}\defEns{ {\norm{\y-\rmA X_{k}^n }^2_2}/{2(\sigma^2_n)^{2})}-{\dim}/{(2\sigma^2_n)}}\right]$.		
		\ENDFOR				
		\STATE Output: $\thetaEB =  \left. \sum_{n=0}^{N-1} \omega_{n} \theta_n \middle/ \sum_{n=0}^{N-1} \omega_n \right.$ and  $\bar{\sigma}^2_N =  \left. \sum_{n=0}^{N-1} \omega_{n} \sigma^2_n \middle/ \sum_{n=0}^{N-1} \omega_n \right..$	
	\end{algorithmic}
\end{algorithm}

One of the complications that stems from working with an unknown noise variance is that the Lipschitz constant $\Ly$ is unknown. This is a problem because $\Ly$ affects the maximum step-size $\gamma$ that we can use in the Markov Kernels while ensuring convergence; $\Ly$ is usually also used to set $\lambda$. To overcome this, we propose to initialise the algorithm by assuming the worst case scenario, i.e. $\sigma^2=\sigma^2_{min}$, which will lead to the largest $\hat{L}_y=(0.99^2/\sigma_{min}^2)$, and in turn lead to the smallest possible step-size $\gamma$ and a small $\lambda$. Since this value is usually very conservative, one can run some iterations of the algorithm until the value of $\sigma^2_n$ begins to stabilise, then refine $\hat{L}_y$ to update the algorithm parameters $\gamma$ and $\lambda$, and continue iterations with those updated values. Here we adopt this approach and run the algorithm in three stages, where we update $\gamma$ and $\lambda$ at the end of each stage by using the estimates of $\bar{\sigma}^2_N$ available at that point to refine $\hat{L}_y$. In accordance with the guidelines provided in \Cref{append:ssec:setting-algo-parameters}, we set $\lambda = \min\left( 5 \hat{L}_y^{-1}, \lambda_{\max}\right)$ with $\lambda_{\max}=2$ and $\gamma = 0.98 \times (\hat{L}_y +1/\lambda)^{-1}$. We have set $\sigma^2_{min}$ and $\sigma^2_{max}$ by assuming prior knowledge that the SNR is between $15$~dB and $45$~dB, but other values could be used without significantly impacting results. In each stage we use $300$ warm-up iterations, set $\btheta_0=0.01$, $\sigma^2_0=(\sigma^2_{min}+\sigma^2_{max})/2$, $X_0^0=y$, $m_n=1$, $\delta_n = 10 \times n^{- 0.8}/\dim$, and  $\delta'_n = 10 \times n^{- 0.8}/\dim$ for any $n \in \nsets$. At each stage, we use the same stopping criteria as in \Cref{sssec:experiments.TV-deblur}, with a tolerance of $10^{-3}$ both for $\theta_n$ and $\sigma_n$ (the algorithm progresses to the next stage (or is stopped) when both iterates meet the criteria).

\newcommand{\snrTVsigma}{30}
\def\widthfig{\linewidth}
\def\widthminipage{0.32\linewidth}	
\begin{figure}[!htb]
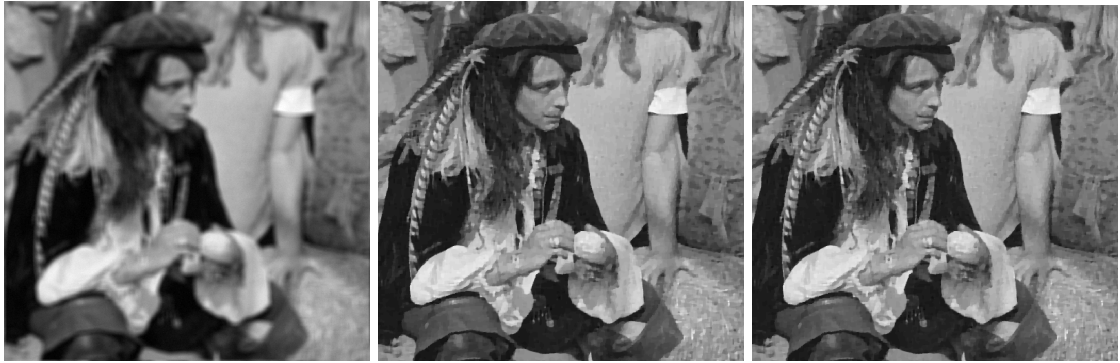

	\centering	
	\begin{minipage}[t]{\widthminipage}
		\centering		
		\centerline{\incManFull{img/num-res-deblurTV/snr\snrTVsigma/EB/man_degraded}\vspace{0.1cm}}
		\centerline{(a) Degraded}
	\end{minipage}	
	\begin{minipage}[t]{\widthminipage}
		\centering
		\centerline{\incManFull{img/num-res-deblurTV/snr\snrTVsigma/EB/man_estim}\vspace{0.1cm}}
		\centerline{(b) EB with known $\sigma$}
	\end{minipage}
	\begin{minipage}[t]{\widthminipage}
		\centering
		\centerline{\incManFull{img/num-res-deblurTV-estimSigma/snr\snrTVsigma/man_estim}\vspace{0.1cm}}
		\centerline{(c) EB with unknown $\sigma$}
	\end{minipage}
	\caption{	\small		
		Deblurring with \tvname \ prior for \texttt{man}: (a) blurred and noisy (\SNRname=\snrTVsigma~$\mathrm{dB}$) observation $ \y $, (b-c) \MAPname \ estimator with $\thetaavg_N $ computed with empirical Bayes using (b) true and (c) estimated $\sigma$.\normalsize}
	\label{fig:num-res-TV-estimSigma-xMAP-full}
\end{figure}

For illustration, \Cref{fig:num-res-TV-estimSigma-xMAP-full} shows the results obtained with \Cref{algo:MCMC_single_chain_uknown_sigma} for the \texttt{man} test image. For comparison, we also show the results of \Cref{sssec:experiments.TV-deblur} obtained by using the true value of $\sigma$. The displayed images correspond to the ${\snrTVsigma~\text{$\mathrm{dB}$}}$ \SNRname \ setup. Observe there is very little difference between the recovered image using the true value of $\sigma^2$ and the marginal MLE estimate $\bar{\sigma}^2_N$ obtained with \Cref{algo:MCMC_single_chain_uknown_sigma}. 

\begin{table}[!ht]	
	\renewcommand{\arraystretch}{1.2}
	\centering
	{\small
		\centering
		\begin{tabular}{rcccccc}\hline		
			Method & \multicolumn{2}{c}{\SNRname=20 $\mathrm{dB}$}       & \multicolumn{2}{c}{\SNRname=30 $\mathrm{dB}$}     & \multicolumn{2}{c}{\SNRname=40 $\mathrm{dB}$}                               \\
			&\MSEname &Time (min)& \MSEname   &Time (min) & \MSEname &Time (min)\\
			\hline
			$ \thetaMSE$ & 23.29      &           & 21.39           &       & 19.06            &              \\

			E.B. with known $\sigma$   & 23.50        & 0.84         & 21.45            & 0.85        & 19.24     & 0.85       \\			
			E.B. with unknown $\sigma$    & 23.53        & 1.02         & 21.52            & 1.35        & 19.27     & 1.77    \\
			\hline                     
		\end{tabular}\vspace{0.1cm}
	}					
	\caption{\small Deblurring with \tvname \ prior and unknown $\sigma$. Table with average mean squared error obtained for ten images for the experiment where $\sigma$ is estimated jointly with $\theta$. For reference we also include the results obtained with empirical Bayes when $\sigma$ is known and using the oracle value $\thetaMSE$ that minimises the MSE. }\normalsize			
	\label{table_deblur_TV_estimSigma}		
\end{table}	

\Cref{table_deblur_TV_estimSigma} presents a detailed comparison of the results obtained with \Cref{algo:MCMC_single_chain_uknown_sigma}. Again, observe that the quality of the restored images obtained with the marginal MLE estimate $\bar{\sigma}^2_N$ is comparable to that of the images obtained with the true value of $\sigma^2$, with a moderate overhead in the computing times when the three-stage approach is used. We conclude by presenting in   \Cref{fig:num-res-TV-estimSigma-theta-evol} the evolution of the iterates $(\theta_n)_{n \in \nset} $ and $ (\sigma^2_n)_{n \in \nset} $ for the last stage of the algorithm (the first two stages are discarded). Observe that the algorithm converges very quickly, similarly to the case when $\sigma^2$ is known.

\begin{figure}[!h]			
	\begin{minipage}[l1]{.49\linewidth}
		\centering
		\centerline{\includegraphics[width=0.99\textwidth]{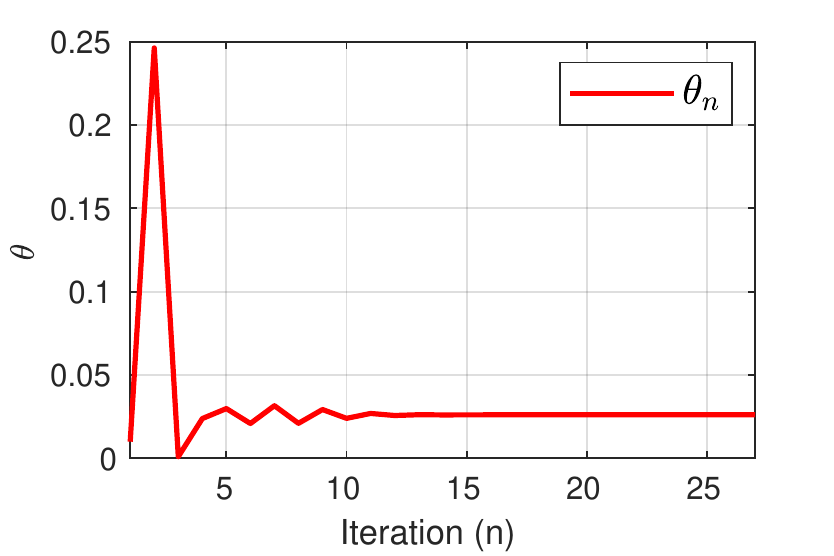}}			
	\end{minipage}
	\begin{minipage}[l2]{.49\linewidth}
		\centering
		\centerline{\includegraphics[width=0.99\textwidth]{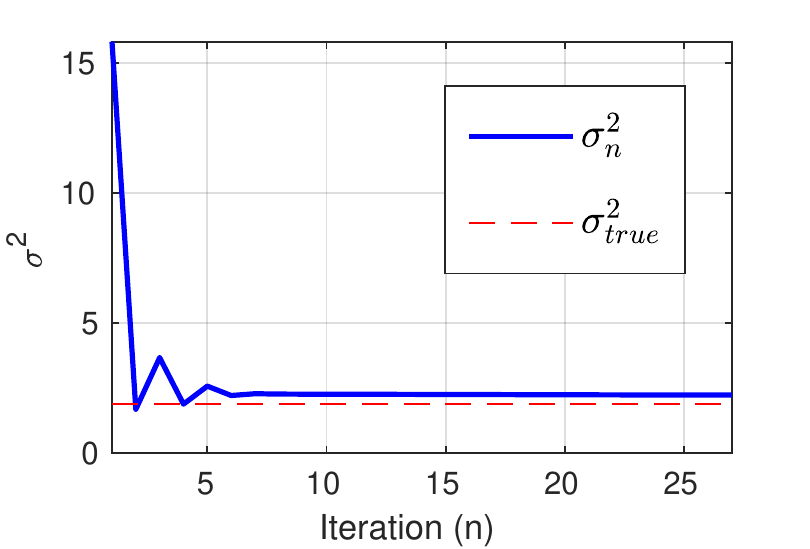}}				
	\end{minipage}
	\caption{\small Deblurring with \tvname \ prior and unknown noise variance $\sigma^2$. Evolution of the sequence of iterates $(\theta_n)_{n \in \nset}$ and  $(\sigma^2_n)_{n \in \nset}$ for the proposed method for the \texttt{man} test images (\SNRname=\snrTV~$\mathrm{dB}$).}
	\label{fig:num-res-TV-estimSigma-theta-evol}
	\normalsize
\end{figure}

\subsubsection{Wavelet deconvolution with synthesis prior}
\label{sssec:experiments.wt-deblur}
We now consider image deblurring under a wavelet synthesis formulation, where we assume that $x \in \rset^d$ represents the unknown image in an redundant 4-level Haar wavelet representation $\Psi$, with dimension $d = 10 \times d_y = 10 \times 512 \times 512$ coefficients. We assume a Laplace prior on the elements of $x$ with unknown parameter $\theta$. Accordingly, the posterior is of the form \eqref{EQ: posterior} with $ f_\y(x)= \norm{\y-\rmA \Psi x}^2_2/(2\sigma^{2}) $, $\g(x) = \Vnorm[1]{x}$. To obtain solutions we map $x$ to pixel domain by computing $\Psi^\top x$.

To compute $\thetaEB$ we use \Cref{algo:MCMC_single_chain}.  The algorithm parameters are chosen following the recommendations provided in \Cref{append:ssec:setting-algo-parameters}; we do not consider any warm-up iterations, and set  $\theta_0=0.01$, $X_0^0=y$, for any $n \in \nsets$, $m_n=1$, $\delta_n = 10 \times n^{- 0.8}/\dim$,  $\lambda = \min\left( 5 \Ly^{-1}, \lambda_{\max}\right)$ with $ \lambda_{\max}=2 $ and $\Ly=(0.98/\sigma)^2$. We use the same stopping criteria as in the previous experiment and we consider two different tolerance levels: i) we stop the algorithm when the relative change $\absLigne{\theta_{N+1} - \theta_N}$ is smaller than $10^{-4}$, and ii) when the relative change is smaller than $10^{-3}$. We set $(\omega_n)_{n \in \nset}$ to have  $N_0 = 20$ burn-in iterations and compute $\thetaavg_N$ using  \eqref{eq:thetaavg}. To compute the MAP estimate we use SALSA with the following parameters: ${\texttt{inneriters} = 1}$, \texttt{outeriters} = $1000$, \texttt{tol} = $10^{-5}$ and \texttt{mu}~=~$\thetaEB$.

\newcommand{\snrSynth}{20}
\def\widthfig{\linewidth}
\def\widthminipage{0.24\linewidth}	
\begin{figure}[!h]
	\centering		
	\begin{minipage}[t]{\widthminipage}
		\centering		
		\centerline{\includegraphics[width=\widthfig,trim={0pt 0pt 0pt 0pt},clip]{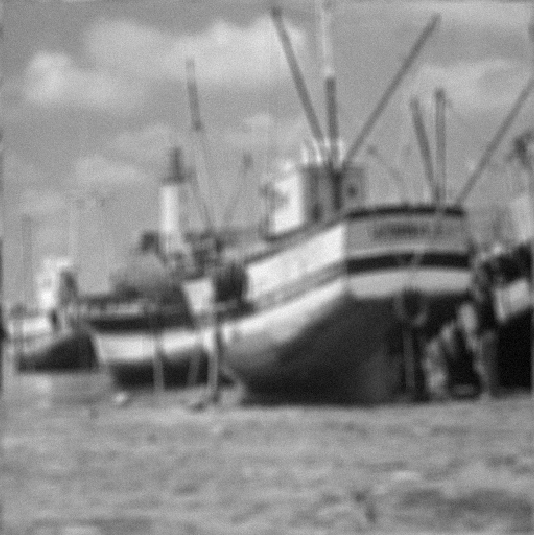}\vspace{0.1cm}}		\centerline{(a) Degraded}
	\end{minipage}
	\begin{minipage}[t]{\widthminipage}
		\centering
		\centerline{\includegraphics[width=\widthfig,trim={0pt 0pt 0pt 0pt},clip]{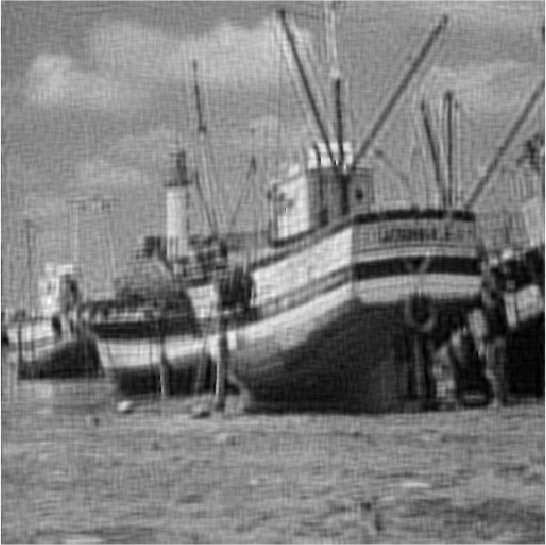}\vspace{0.1cm}}
		\centerline{(b) EB}
	\end{minipage}
	\begin{minipage}[t]{\widthminipage}
		\centering				
		\centerline{\includegraphics[width=\widthfig,trim={0pt 0pt 0pt 2pt},clip]{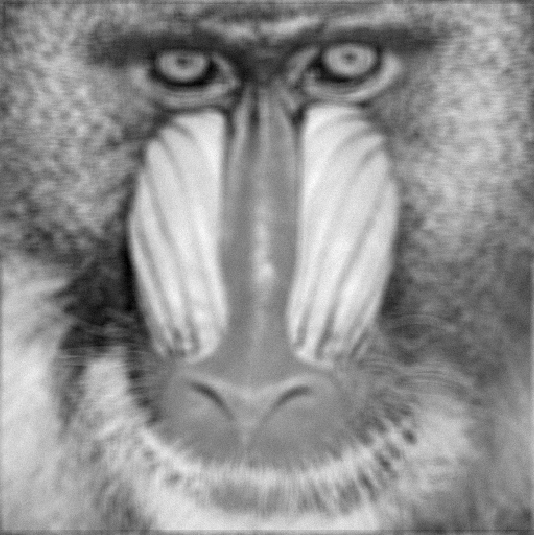}}
		\centerline{(c) Degraded}
	\end{minipage}
	\begin{minipage}[t]{\widthminipage}
		\centering	
		\centerline{\includegraphics[width=\widthfig,trim={0pt 0pt 0pt 0pt},clip]{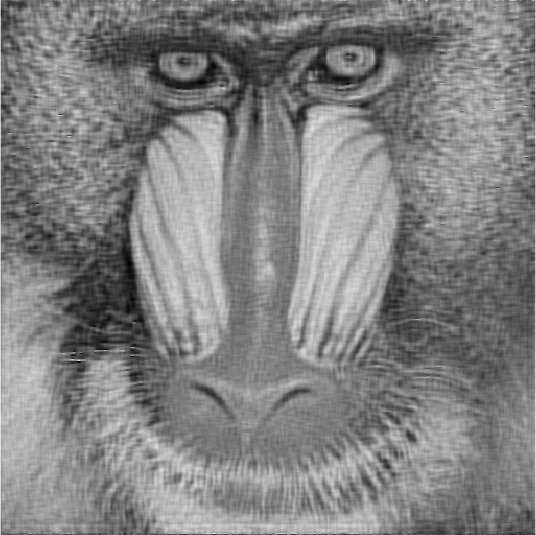}}
		\centerline{(d) EB}
	\end{minipage}
	\caption{	\small		
		Wavelet deconvolution with synthesis-$\ell_1$ prior for \texttt{boat} and \texttt{mandrill}: (a),(c)~blurred and noisy (\SNRname=\snrSynth~$\mathrm{dB}$) observation $ \y $, (b),(d) \MAPname \ estimator obtained with empirical Bayes.\normalsize}
	\label{fig:num-res-synthL1-xMAP-full}
\end{figure}

In \Cref{fig:num-res-synthL1-xMAP-full} we show the results obtained for two of the test images (\texttt{boat}  and \texttt{mandrill}) using the proposed method. The displayed images correspond to the ${20 ~\text{$\mathrm{dB}$}}$ \SNRname \ setup. In  \Cref{fig:num-res-synthL1-theta} we provide further details for the \texttt{boat} image, showing the evolution of the iterates $(\theta_n)_{n \in \nset}$ and the relative differences on its running average value $(\thetaavg_N)_{N \in \nset}$ throughout iterations.

\begin{figure}[!h]	
	\centering		
	\begin{minipage}[l1]{.47\linewidth}
		\centering
		\centerline{\includegraphics[width=\textwidth]{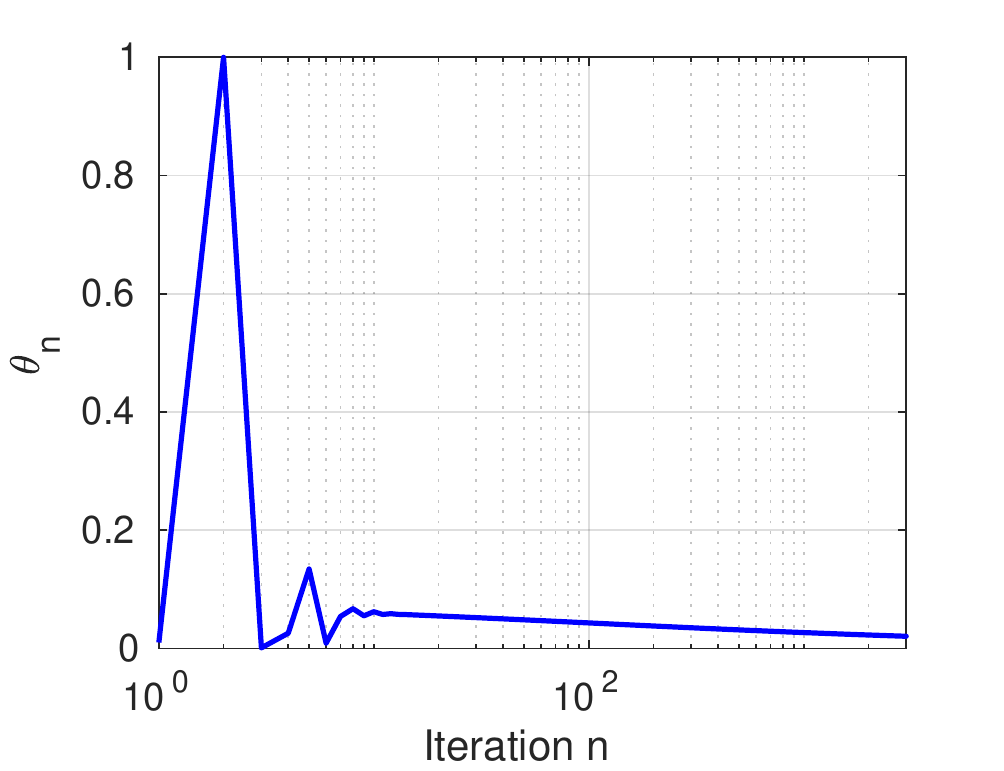}}		
		\centerline{(a)}
	\end{minipage}
	\begin{minipage}[l2]{.42\linewidth}
		\centering
		\centerline{\includegraphics[width=\textwidth]{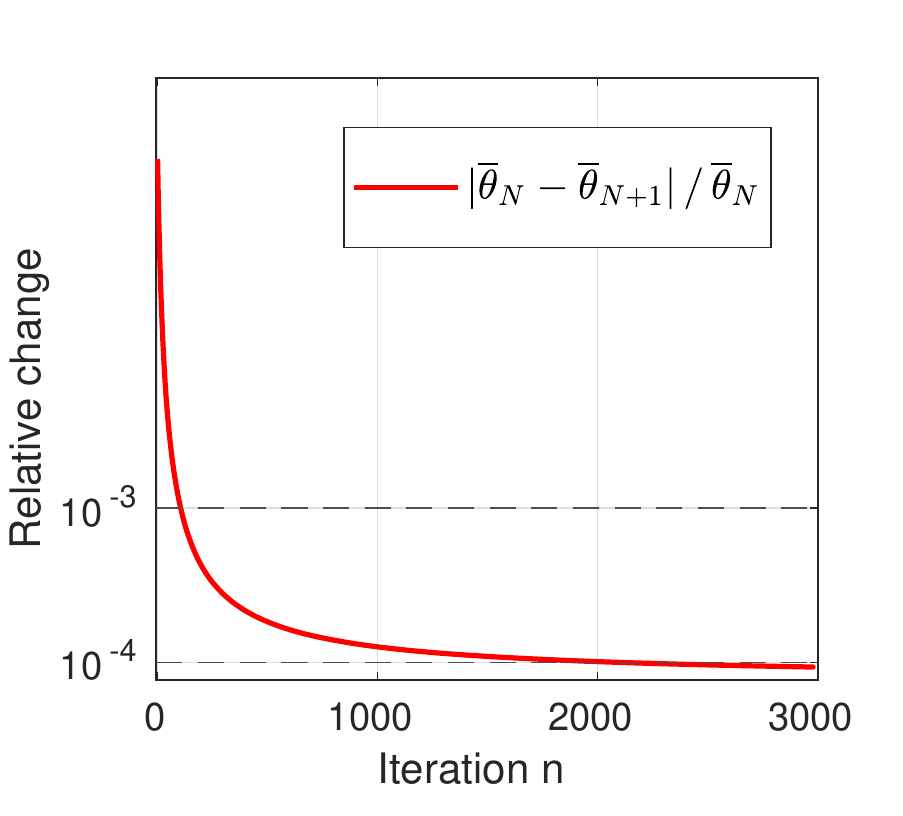}}		
		\centerline{(b)}	
	\end{minipage}
	\caption{\small Wavelet deconvolution with synthesis-$\ell_1$ prior for \texttt{boat}  image (\SNRname=\snrSynth~$\mathrm{dB}$). Evolution of (a) the iterates $ (\theta_n)_{n \in \nset} $ in log-scale and (b) the relative change in $(\thetaavg_N)_{N \in \nset}$ for the proposed method.}
	\label{fig:num-res-synthL1-theta}
	\normalsize
\end{figure}

In \Cref{fig:num-res-synthL1-xMAP}  we compare the results obtained by using each of the considered methods, showing a close-up on an image region that contains fine details and sharp edges.  \Cref{fig:num-res-synthL1-MSE} shows a plot of the \MSEname \ obtained with each method for the same two test images.

\def\widthfig{0.99\linewidth}
\def\widthminipage{0.165\linewidth}	
\begin{figure}[!h]		
	\newcommand\incBoat[1]{\includegraphics[width=\widthfig,trim={3cm 2cm 4cm 5cm},clip]{#1}\vspace{0.1cm}}
	\newcommand\incMandrill[1]{\includegraphics[width=\widthfig,trim={8cm 11.2cm 4cm 0.8cm},clip]{#1}}
	\begin{minipage}[t]{\widthminipage}
		\centering
		\centerline{\incBoat{img/num-res-deblurSynthL1/snr\snrSynth/EB/boat_orig}}
		\centerline{\incMandrill{img/num-res-deblurSynthL1/snr\snrSynth/EB/mandrill_orig}}
		\centerline{(a) Original}
	\end{minipage}\hfill
	\begin{minipage}[t]{\widthminipage}
		\centering		
		\centerline{\incBoat{img/num-res-deblurSynthL1/snr\snrSynth/EB/boat_degraded}}
		\centerline{\incMandrill{img/num-res-deblurSynthL1/snr\snrSynth/EB/mandrill_degraded}}
		\centerline{(b) Degraded}
	\end{minipage}\hfill
	\begin{minipage}[t]{\widthminipage}
		\centering
		\centerline{\incBoat{img/num-res-deblurSynthL1/snr\snrSynth/EB/boat_estim}}
		\centerline{\incMandrill{img/num-res-deblurSynthL1/snr\snrSynth/EB/mandrill_estim_clamp}}
		\centerline{(c) EB \tiny{tol $10^{-4}$}}
	\end{minipage}\hfill
	\begin{minipage}[t]{\widthminipage}
		\centering		
		\centerline{\incBoat{img/num-res-deblurSynthL1/snr\snrSynth/EB/boat_estim_tol}}
		\centerline{\incMandrill{img/num-res-deblurSynthL1/snr\snrSynth/EB/mandrill_estim_tol}}
		\centerline{(d) EB \tiny{tol $10^{-3}$}}
	\end{minipage}\hfill
	\begin{minipage}[t]{\widthminipage}
		\centering
		\centerline{\incBoat{img/num-res-deblurSynthL1/snr\snrSynth/HB/boat_estim}}
		\centerline{\incMandrill{img/num-res-deblurSynthL1/snr\snrSynth/HB/mandrill_estim}}
		\centerline{(e) HB}
	\end{minipage}\hfill
	\begin{minipage}[t]{\widthminipage}
		\centering
		\centerline{\incBoat{img/num-res-deblurSynthL1/snr\snrSynth/discPrinciple/boat_estim}}
		\centerline{\incMandrill{img/num-res-deblurSynthL1/snr\snrSynth/discPrinciple/mandrill_estim}}	
		\centerline{(f) DP}
	\end{minipage}			
	\caption{	\small		
		Wavelet deconvolution with synthesis-$\ell_1$ prior. Close-up on \texttt{boat} and \texttt{mandrill} images: (a)~True image, (b) blurred and noisy (\SNRname=\snrSynth~$\mathrm{dB}$) observation $ \y $, (c)-(f) MAP estimators obtained with Empirical Bayes (tol. $10^{-4}$ and $10^{-3}$), hierarchical Bayes and discrepancy  principle,  respectively.\normalsize}
	\label{fig:num-res-synthL1-xMAP}
\end{figure}

\begin{figure}[!h]
	\begin{minipage}[l1]{.48\linewidth}
		\centering
		\centerline{\includegraphics[width=\textwidth]{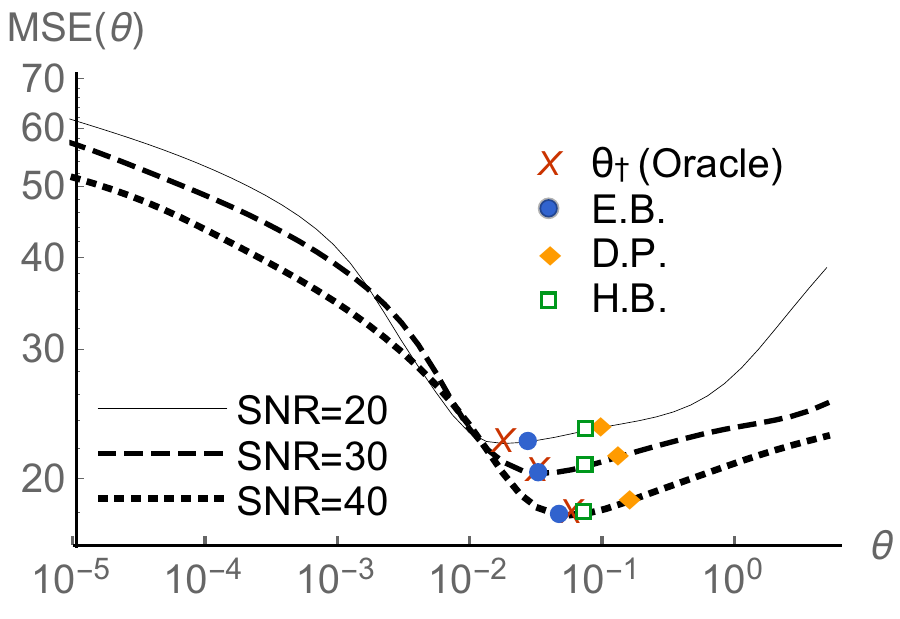}}
		\centerline{(a) \texttt{boat}}
	\end{minipage}
	\begin{minipage}[l2]{.48\linewidth}
		\centering
		\centerline{\includegraphics[width=\textwidth]{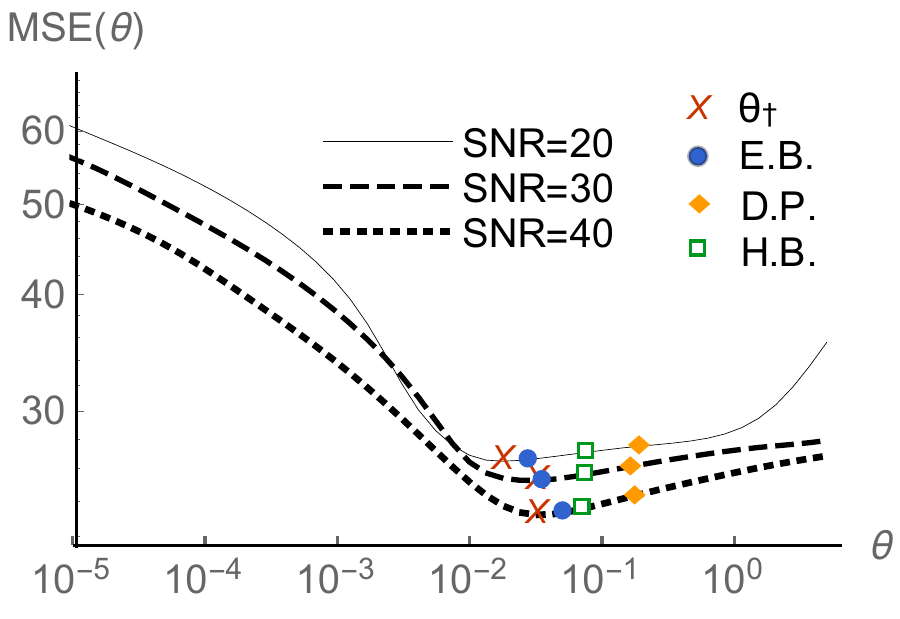}}
		\centerline{(b) \texttt{mandrill}}		
	\end{minipage}
	\caption{\small Wavelet deconvolution with synthesis-$\ell_1$ prior - Mean squared error (\MSEname) obtained for (a)~\texttt{boat}  and (b)~\texttt{mandrill} for different values of $\theta$. We compare the values obtained with empirical Bayes with tolerance $10^{-4}$, discrepancy principle,  hierarchical Bayes, and the optimal value $\thetaMSE$.} \label{fig:num-res-synthL1-MSE}
\end{figure}

\Cref{table_synth_l1} shows the average \MSEname \ values and average computing times obtained for each method. We observe once again that the empirical Bayesian method achieves the best results for all \SNRname \ values and is very close to the oracle performance. Reducing the tolerance leads to a small improvement in \MSEname, at the expense of a higher computing time. The discrepancy principle consistently overestimates the parameter leading to over-smoothed solutions.

\begin{table}[!h]
	\label{table_synth_l1}
	\renewcommand{\arraystretch}{1.3}
	\begin{minipage}[!t]{.61\linewidth}
		{\small
			\begin{tabular}{rcccccc}\hline
				Method & \multicolumn{2}{c}{\SNRname=20 $\mathrm{dB}$}   & \multicolumn{2}{c}{\SNRname=30 $\mathrm{dB}$}      & \multicolumn{2}{c}{\SNRname=40 $\mathrm{dB}$}     \\
				&\MSEname &Time & \MSEname   &Time & \MSEname &Time\\
				\hline
				$ \thetaMSE$ & 24.23      &        & 22.70     &       & 20.56     &           \\
				\tiny{$tol\, 10^{-4}$} \small E.B.   	  & 24.40     & 4.48    & 22.80    & 3.59      & 20.70   & 2.44    \\	
				\tiny{$tol\, 10^{-3}$} \small E.B.   & 24.70      & 0.36           & 22.90     & 0.25          & 20.80     & 0.09    \\				
				D.P.   & 25.09      &   13.93	      & 23.57     & 28.64	       & 21.38     & 61.03          \\   
				H.B.   & 25.01      & 11.61           & 23.23     & 23.87          & 20.89     & 50.86          \\
				\hline     	                  
			\end{tabular}
			
		}\normalsize
		\vspace{0.3cm}
		\centerline{(a)}	
	\end{minipage}
	\begin{minipage}[!t]{.37\linewidth}
		\vspace{.5cm}
		\centering
		\includegraphics[width=\textwidth]{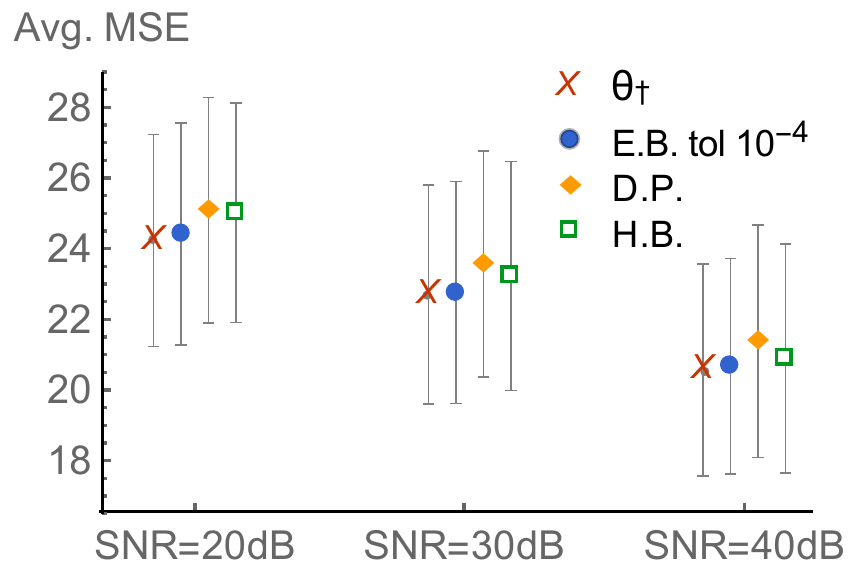}
		\centerline{(b)}
	\end{minipage}
	\caption{\small Wavelet deconvolution with synthesis-$\ell_1$ prior. (a) Table with average mean squared error obtained for ten images with different algorithms. Average execution times expressed in minutes. (b)  summarises the content of the table and shows the standard deviation with error bars.}\normalsize					
\end{table}	
For high \SNRname \ values, both Bayesian methods attain similar values of \MSEname, but the proposed empirical Bayes methodology is five times faster. We want to point out that these general conclusions depend a lot on the parameters used for the solver of the MAP estimation problem (in this case SALSA \cite{salsa2010fast}). We included a detailed analysis of this in \Cref{append:method-comparison}. 

\subsection{Hyperspectral Unmixing with \tvname-SUnSAL}
\label{ssec:experiments.hyperunmix}
Hyperspectral sensors acquire hundreds of narrow band spectral images in different frequency bands. These images are collected in a three-dimensional hyperspectral data cube for processing and analysis.  Although the  spectral resolution is high,  the spatial resolution is usually low, leading to the existence of “mixed” spectra in the acquired image pixels \cite{iordache2012total}. Hyperspectral unmixing is a source separation problem that aims at  decomposing each mixed pixel into its constituent spectra (the so-called end-members) and their corresponding fractional abundances or proportions. This is normally done under the assumption of a linear mixing model \cite{simoes2015convex}. In particular, linear unmixing techniques assume the availability of a library of spectral signatures and use a model $\y=\rmA \x + \rmw$, where  $\y \in \mathbb{R}^{\dimf \times \dimp}$ is the hyperspectral image with $\dimf$ frequency channels and $\dimp$ pixels; $\x \in \mathbb{R}^{\dimm \times \dimp} $ is the fractional abundance matrix compatible with the library $\rmA \in \mathbb{R}^{\dimf\times \dimm} $ containing the pure spectral signatures for $\dimm$ different materials; and $\rmw$ is a $\dimf \times \dimp$ Gaussian random variables with zero mean and covariance matrix $\sigma \Id$ and $\sigma >0$. In \cite{iordache2012total}, the unmixing problem is solved by using the regulariser $g$ given for any $x \in \rset^{\dimm} \times \rset^{\dimp}$ by
\begin{equation}
g(x)=(\mathrm{TV}(x), \norm{\x}_{1}) \quad \text{s.t.}  \quad x\geq0 \eqsp ,
\end{equation}
which is associated with a two dimensional regularisation parameter $\theta = (\theta^{\mathrm{TV}}, \theta^1) \in \rset^2$. $\theta^{\mathrm{TV}} \in \rset$ controls the spatial cohesion of the objects, and $\theta^1 \in \rset$ enforces sparsity on $\x$. In this experiment, \tvname \ is the vectorial isotropic total variation pseudo-norm given for any $x \in \rset^{\dimm} \times \rset^{\dimp}$ by
\begin{equation} \mathrm{TV}(x) = \sum_{i = 1}^{\dimp} \sum_{j \in \mathcal{V}_i} \norm{x_i - x_j}_1 \eqsp ,
\end{equation}
where for any $i \in \{1,\ldots,\dimp\}$, $x_i \in \mathbb{R}^{\dimm}$ denotes the $i$-th image pixel and $\mathcal{V}_i$ its vertical and horizontal neighbour pixels.

Although this regulariser is not separable and we would therefore have to use \Cref{algo:MCMC_double_chain} with two MCMC chains, we can use a pseudo-likelihood approximation estimate $\theta$ using a single MCMC chain together with the expression of $\nabla_{\theta} \log \rmZ(\theta)$ for the homogeneous case.  In this way we can achieve highly competitive computing times as well as compare our results with the hierarchical Bayesian method from \cite{EUSIPCO}, which we would otherwise not be able to apply to this problem. 

More precisely, we consider $[\partial \log \rmZ / \partial \theta^1](\theta) = \dim /  \theta^1$	and $[\partial \log \rmZ / \partial \theta^{TV}](\theta) = \dim /  \theta^{TV}$. 
{Although} $x \mapsto \mathrm{TV}(\x)$ and $x \mapsto \norm{\x}_{1} $ are not acting on independent subsets of $\x$, we have empirically observed that this  provides a good approximation and delivers excellent results.

We consider the experiment A-\textit{Simulated Data Sets} case 1) \textit{Simulated Data Cube 1} presented in \cite[Section 4]{iordache2012total}, particularly the case where $\rmw$ is a white Gaussian noise. In this experiment a synthetic hyperspectral image is generated by using five randomly selected spectral signatures. The image has $\dimp= 75 \times 75= 5625$ pixels and $\dimf=224$  frequency bands per pixel. For full details see \cite{iordache2012total}. We follow the exact same procedure as presented there, except for a modification in the spectral signature dictionary $\rmA$. In \cite{iordache2012total} they consider a dictionary $\rmA \in \mathbb{R}^{224 \times 240} $, which is a library generated from a random selection of 240 materials  from the USGS library\footnote{Available online: http://speclab.cr.usgs.gov/spectral.lib06}. Here we consider a simplified version where we only select $\dimm=12$ random materials, thus having $\rmA \in \mathbb{R}^{224 \times 12}$. Out of these 12 materials, only 5 are present in the synthetic image. 
The synthetic fractional abundances $\xGT$ are displayed in the first row of \Cref{fig:num-res-hypers-xMAP} (only the 5 present end-members are shown)

We use the proposed algorithm to estimate $ \theta^{\mathrm{TV}} $ and $ \theta^{1} $ for this setup using  \Cref{algo:MCMC_single_chain_separable} under three different noise levels: we consider a \SNRname \ of ${20 ~\text{$\mathrm{dB}$}}$, ${30 ~\text{$\mathrm{dB}$}}$ and ${40 ~\text{$\mathrm{dB}$}}$. For comparison, we also report the results obtained with the joint MAP method from \cite{EUSIPCO} and by using the oracle value $\thetaMSE$ that maximises the estimation signal-to-reconstruction-error (SRE) given by $\|\xGT\|_2^2/\|\xGT-\hat{\x}_{\MAP}\|_2^2$.

We evaluated the proximal operator of $x \mapsto \theta^{\mathrm{TV}} \mathrm{TV}(x)+ \theta^1\norm{\x}_{1}$ using SUnSAL solver from \cite{iordache2012total} with 20 iterations. We address the positivity constraint separately by using its Moreau-Yosida envelope, leading to the additional term $x \mapsto (x-\Pi_{+}(x))/\lambda$ where $\Pi_{+}$ is the projection operator onto $\coint{0,+\infty}^{\dimm} \times \coint{0,+\infty}^{\dimp}$, and $\lambda$ is the same smoothing parameter used for the other proximal operators.

To speed up the convergence, we use a gradient preconditioning technique explained in \Cref{append:ssec:convergence-speed}. Since we use the preconditioned gradient of $\f$ instead of the gradient of $\f$, the Lipschitz constant becomes $\Lt=1/\sigma^2$.
The algorithm parameters are chosen following the recommendations provided in \Cref{append:ssec:setting-algo-parameters}; we set  $\theta_0^1=10$, $\theta_0^{\mathrm{TV}}=10$, we initialised $X_0^0$ using the pseudo-inverse of $A$ and projecting on the space of positive matrices. In addition, we perform $200$ warm-up iterations and set for any $n \in \nsets$, $m_n=1$ and  $\delta_n = ~ n^{- 0.8}/(\dimp \dimm)$.

Special care was taken when setting $\gamma >0$ and $\lambda >0$ due to the preconditioning. We set $\gamma=1/(\Lt+2/\lambda)$ for any $n \in \nset$ 
and  $\lambda = 0.9  \times \lambda_{\rmA}/\Lt $, where $\lambda_{\rmA}$ is the largest eigenvalue of $(\rmA^{\transpose} \rmA)^{-1}$.
We run the algorithm for 50  iterations and compute $(\thetaavg)_{N \in \nset}$ as defined in \eqref{eq:thetaavg} with $(\omega_n)_{n \in \nset}$ set to have $N_0 = 30$ burn-in iterations.

\begin{figure}[h!]
	\label{fig:num-res-hypers-xMAP}
	\centering
	\includegraphics[width=0.9 \textwidth]{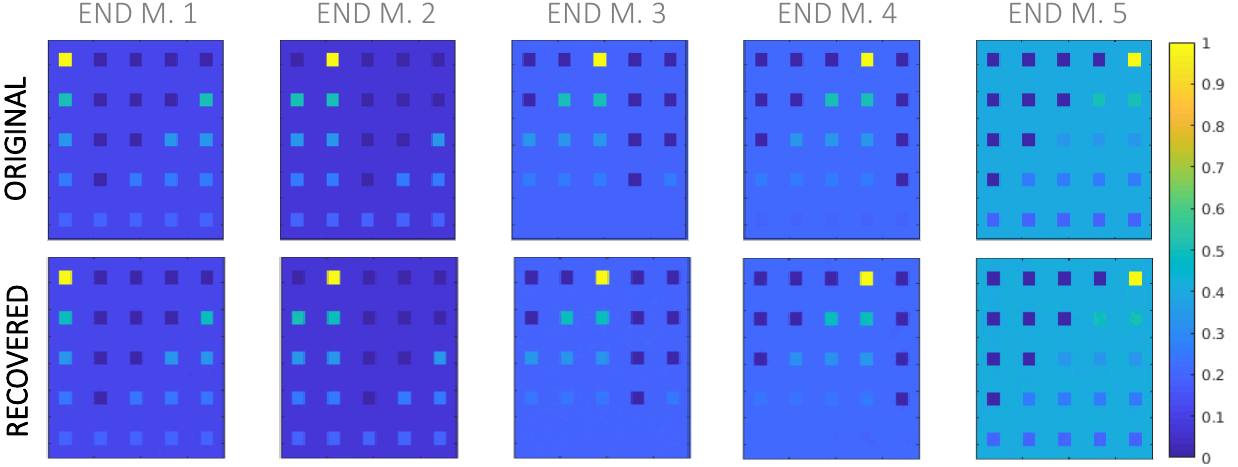}
	\caption{Hyperspectral Unmixing - Synthetic fractional abundances for 5 end-members. Original and \MAPname \  estimates for \SNRname=30 $\mathrm{dB}$ using the empirical Bayes posterior \eqref{EQ: EBposterior}.}
\end{figure}
In \Cref{fig:num-res-hypers-xMAP} we display the MAP recovery of the synthetic fractional abundances using the estimated values of $\theta^{\mathrm{TV}}$ and $\theta^1$ with the SUnSAL solver for \SNRname=30 $\mathrm{dB}$.
\Cref{fig:num-res-hypers-theta} shows the evolution of the iterates $(\theta^1_n)_{n \in \nset}$ and $(\theta^{\mathrm{TV}}_n)_{n \in \nset}$ and the relative change in the running averages $(|\bar{\theta}_{N+1}-\bar{\theta}_{N}|/\bar{\theta}_N)_{N \in \nset}$ throughout iterations for \SNRname=30 $\mathrm{dB}$. Observe the excellent convergence properties of the proposed scheme, which stabilises in as little as $25$ iterations.

\begin{figure}[h!]
	\centering
	\includegraphics[width=0.4\textwidth]{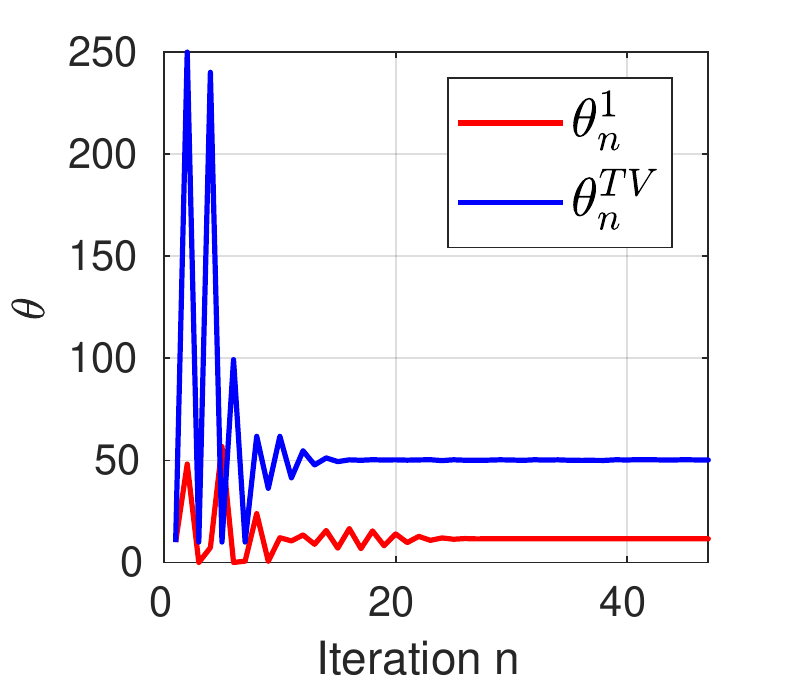}
	\includegraphics[width=0.42\textwidth]{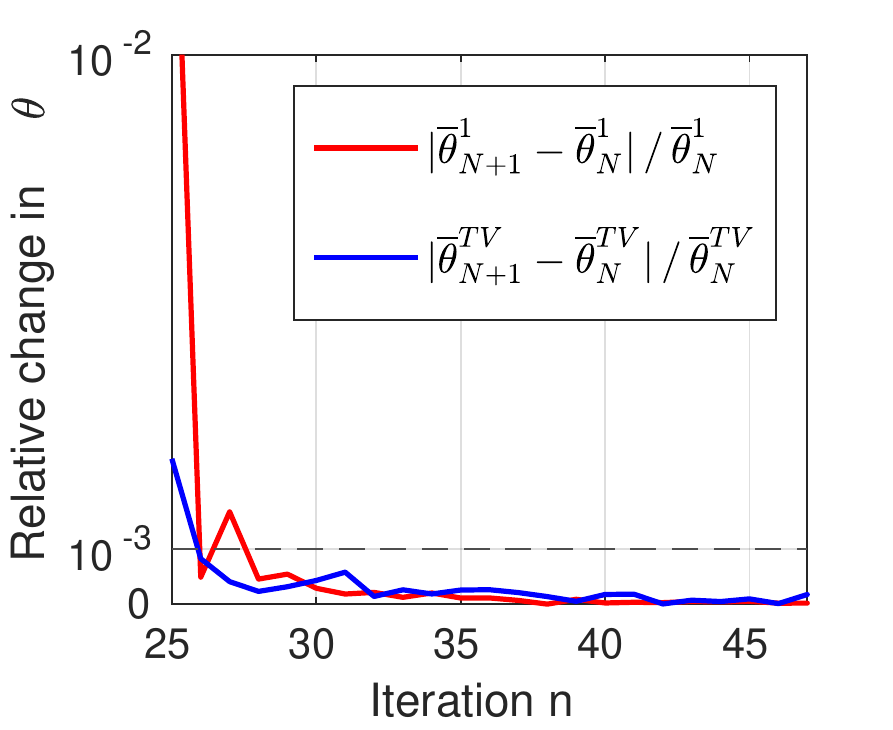}
	\caption{\small Hyperspectral Unmixing - Evolution of the iterates $(\theta^1_n)_{n \in \nset}$ and $(\theta^{\mathrm{TV}}_n)_{n \in \nset}$  (left) and of the relative successive differences $(|\bar{\theta}_{N+1}-\bar{\theta}_{N}|/\bar{\theta}_N)_{N \in \nset}$  (right) for the proposed method  with \SNRname=30 $\mathrm{dB}$. The relative change is computed after 25 burn-in iterations. }
	\label{fig:num-res-hypers-theta}
\end{figure}

The obtained results are reported in \Cref{table_hyper_unmix} and summarised in \Cref{fig:num-res-hypers-resultsCompare}, which shows the signal to reconstruction error (SRE) surfaces for different values of the regularisation parameters.  Observe that the empirical Bayesian method yields good results for all \SNRname 
\begin{table}[!h]
	\label{table_hyper_unmix}
	\centering
	\renewcommand{\arraystretch}{1.3}
	{\small
		\begin{tabular}{rccccccc}
			\hline  Method  &  & \multicolumn{2}{c}{\SNRname=20 $\mathrm{dB}$ } & \multicolumn{2}{c}{\SNRname=30 $\mathrm{dB}$ } & \multicolumn{2}{c}{\SNRname=40 $\mathrm{dB}$ }\\
			& Stop\,criteria &  SRE  &  Time (s)  &  SRE  &  Time (s)  &  SRE  &  Time (s)\\
			\hline  $\thetaMSE$ (Oracle)  & -- &  29.38  &  --  &  38.61  &  --  &  47.64  &  -- \\
			E.B.  & 50\,iters. &  27.46  &  36  &  38.42  &  37  &  45.68  &  42 \\			
			H.B. \cite{EUSIPCO} & 15\,iters. &  18.33  &  76  &  31.72  &  77  &  47.36  &  76\\
			\hline  
		\end{tabular}			
	}
	\caption{Hyperspectral unmixing - Signal to reconstruction error (SRE) obtained for different \SNRname \ values along with computing times expressed in seconds. }
\end{table}
\def\widthfig{\linewidth}
\def\widthminipage{0.33333\linewidth}
\newcommand\incHUplot[1]{\includegraphics[width=\widthfig,trim={0cm 0cm 0cm 0.6cm},clip]{#1}}	
\newcommand\incHUplotB[1]{\includegraphics[width=\widthfig,trim={3.6cm 8cm 4cm 1cm},clip]{#1}}	
\begin{figure}[!h]	
	\centering		
	\begin{minipage}[t]{\widthminipage}
		\centering
		\centerline{	\incHUplot{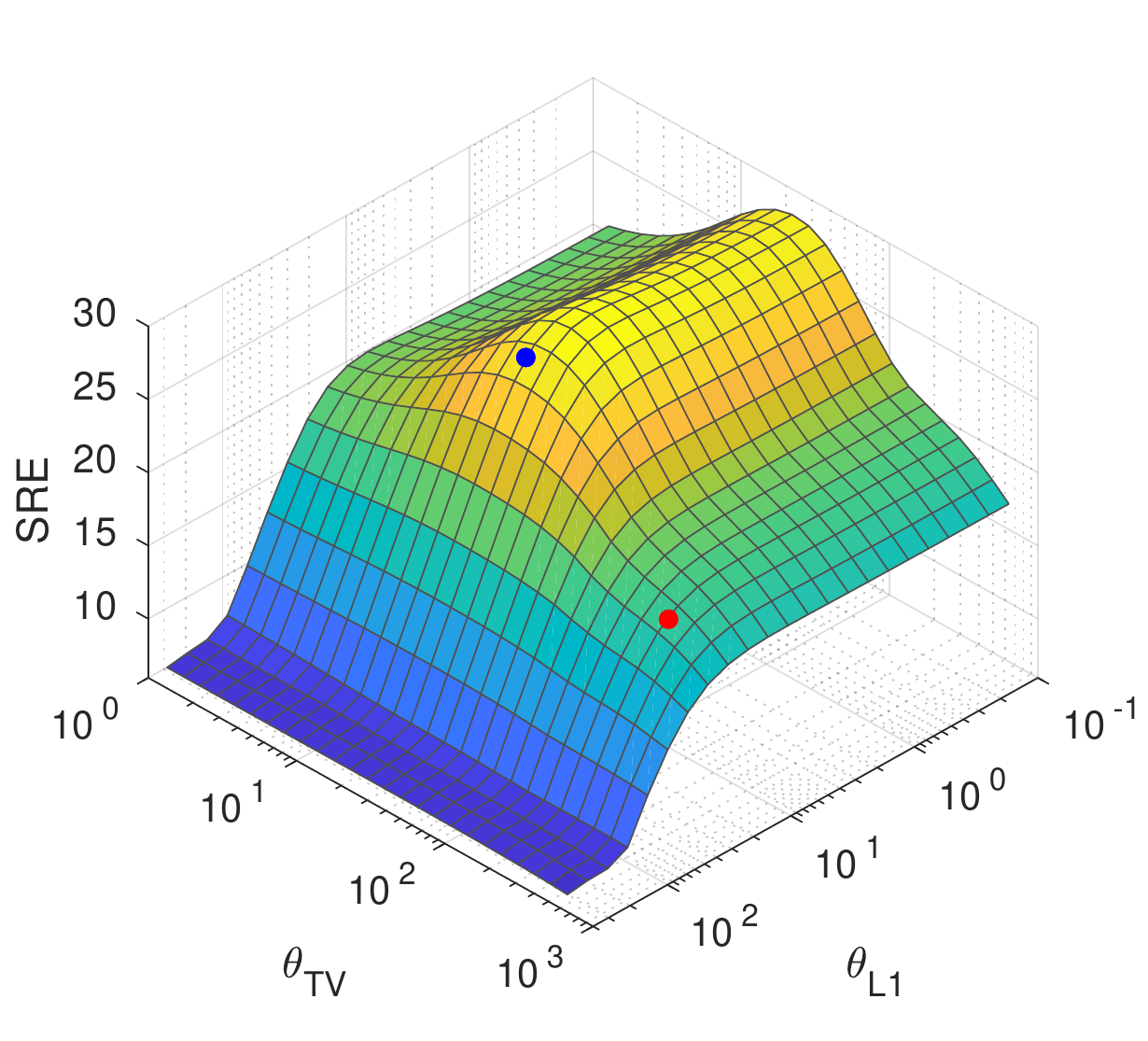}}
		\centerline{(a) \SNRname=20 $\mathrm{dB}$}
	\end{minipage}\hfill
	\begin{minipage}[t]{\widthminipage}
		\centering		
		\centerline{	\incHUplot{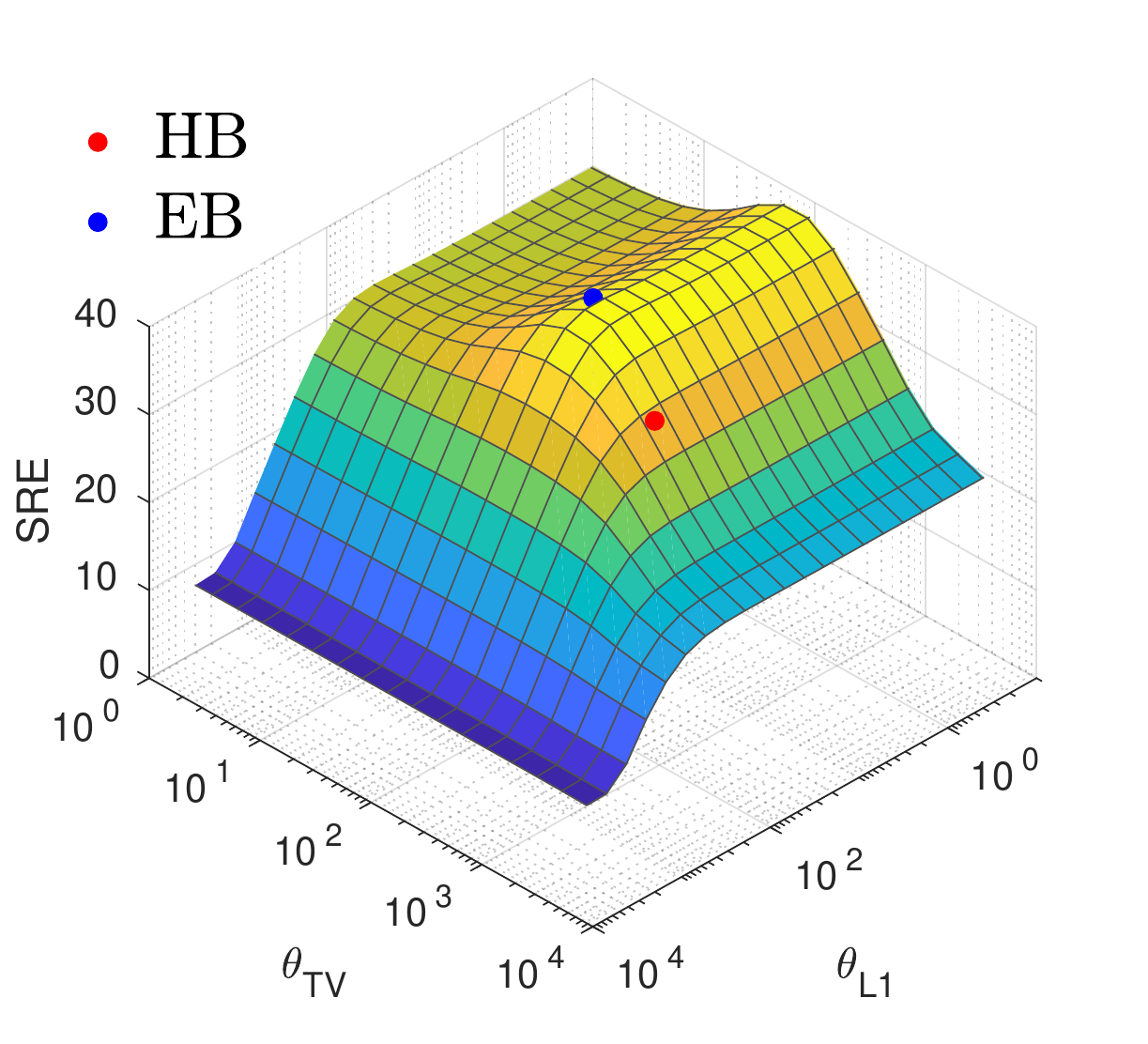}}		
		\centerline{(b) \SNRname=30 $\mathrm{dB}$}
	\end{minipage}\hfill
	\begin{minipage}[t]{\widthminipage}
		\centering
		\centerline{	\incHUplot{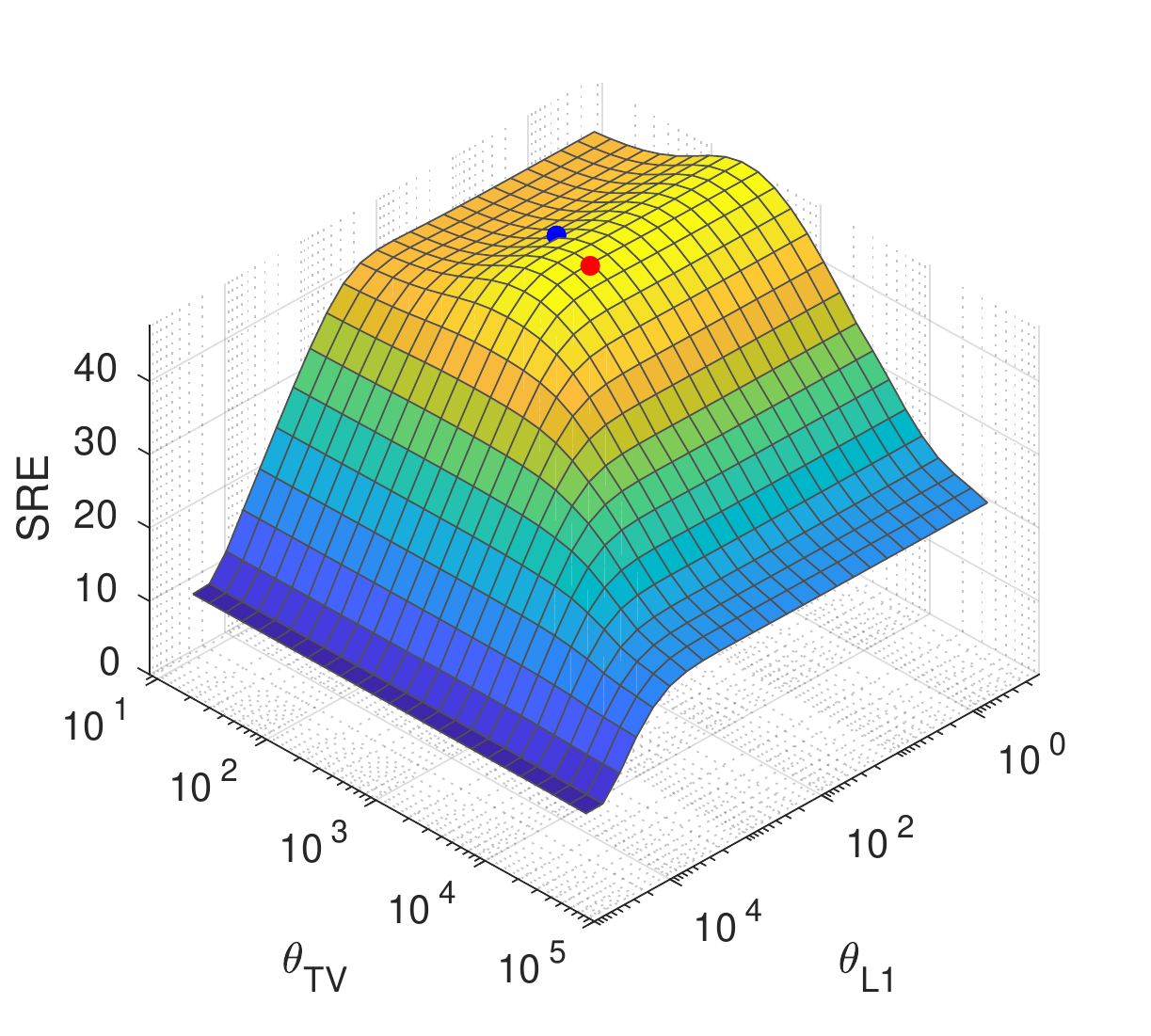}}		
		\centerline{(c) \SNRname=40 $\mathrm{dB}$}
	\end{minipage}
	\caption{	\small		
		Hyperspectral Unmixing - Signal to reconstruction error (SRE) surfaces for different \SNRname \ values expressed in $\mathrm{dB}$.  Comparison between parameters estimated with our empirical Bayesian algorithm (EB) and with the hierarchical Bayesian method (HB) from \cite{EUSIPCO}.\normalsize}
	\label{fig:num-res-hypers-resultsCompare}
\end{figure}
\noindent values, and clearly outperforms the hierarchical Bayesian method for low \SNRname \ values. For high \SNRname \ values the hierarchical method achieved slightly better results. As discussed in \Cref{ssec:connectBayesApproach}, we believe that this is due to the fact that, at high \SNRname \ values,  the likelihood $x \mapsto p(\y|\x)$ dominates the posterior  and mitigates errors related to the misspecification of the prior. More precisely, if the hyperprior that we set on $\theta$ assigns a high weight to values of $\theta$ that lead to bad models, \ie \ a misspecified prior $x \mapsto p(\x|\theta)$, the impact of this misspecification on the recovered estimates depends on the degree of concentration of the likelihood. At high SNR, the likelihood dominates the posterior thus concealing the possible prior misspecification and leading to good results.  Conversely, at low \SNRname \ values, the performance of the hierarchical model is degraded by model misspecification.

Also note in \Cref{table_hyper_unmix} that the  computing times for the empirical Bayesian method are approximately two times faster than the ones for the hierarchical method. 
\subsection{Denoising with a total generalised variation prior}
\label{sssec:experiments.tgv-denois}
In this last experiment, we apply the proposed methodology to a challenging problem that is beyond the scope of the considered class of models and our theoretical guarantees. We consider an image denoising problem where $y \sim \mathcal{N}(x,\sigma^2 \Id)$ with  $\sigma^2 > 0$ and where we use the following prior
$$
p(x|\theta^1,\theta^2) = \frac{1}{\mathrm{Z}(\theta^1,\theta^2)} \exp\{-\mathrm{TGV}_{\theta^1,\theta^2}^{2} (x) - \varepsilon \|x\|_2^2\} \eqsp ,
$$
where $\vareps >0$ and where $\mathrm{TGV}_{\theta^1,\theta^2}^{2} (x)$ is a second-order generalisation of the conventional total variation regulariser, given, for any $(\theta^1, \theta^2) \in \coint{0,+\infty}^2$ and $x \in \rset^{\dim}$, by
\begin{equation}\label{eq:Experiments-TGV-norm}
\mathrm{TGV}_{\theta^1,\theta^2}^{2} (x) =\underset{r \in \mathbb{R}^{2\dim}}{\mathrm{min}} \defEnsLigne{ \theta^1 \norm{r}_{1,2} + \theta^2 \norm{J(\Delta x-r)}_{1,\Frob.} } \eqsp .
\end{equation}
where $\Delta = (\Delta^v, \Delta^h)$ is the discrete image-gradient operator that computes the first-order vertical and horizontal pixel differences, and $J$ computes the Jacobian matrix of the image-gradient vector field to capture second-order information (i.e., $(J\Delta) (x)$ is a discrete image-Hessian operator) \cite{codeTGV}. This generalisation was first considered in \cite{chambolle97} and further studied in \cite{bredies2010total} as a means of incorporating second-order derivative information to eliminate the common staircasing artifacts associated with the conventional \tvname \ regulariser. 

A main difficulty associated with using the \tgvname \ regulariser is the need to correctly set the parameters $\theta^1$ and $\theta^2$, which control the strength as well as the characteristics of the regularisation enforced (as explained in \cite{codeTGV}, the \tgvname \ regularisation behaves like the standard \tvname \ regularisation for large $\theta^2$ values, whereas for small values it behaves like the $\ell_1$-Frobenius norm of the discrete image-Hessian). \Cref{fig:num-res-TGV-parrot} below illustrates the dramatic effect that these two parameters have on the quality of the recovered MAP estimate. Observe the strong coupling between $\theta^1$ and $\theta^2$, which makes setting their values particularly challenging.

\begin{figure}[h!]\centering
	\begin{minipage}[t]{0.49\textwidth}\vspace{0pt}
		
		{\includegraphics[width= \textwidth,trim={0cm 1.5cm 0cm 1cm},clip]{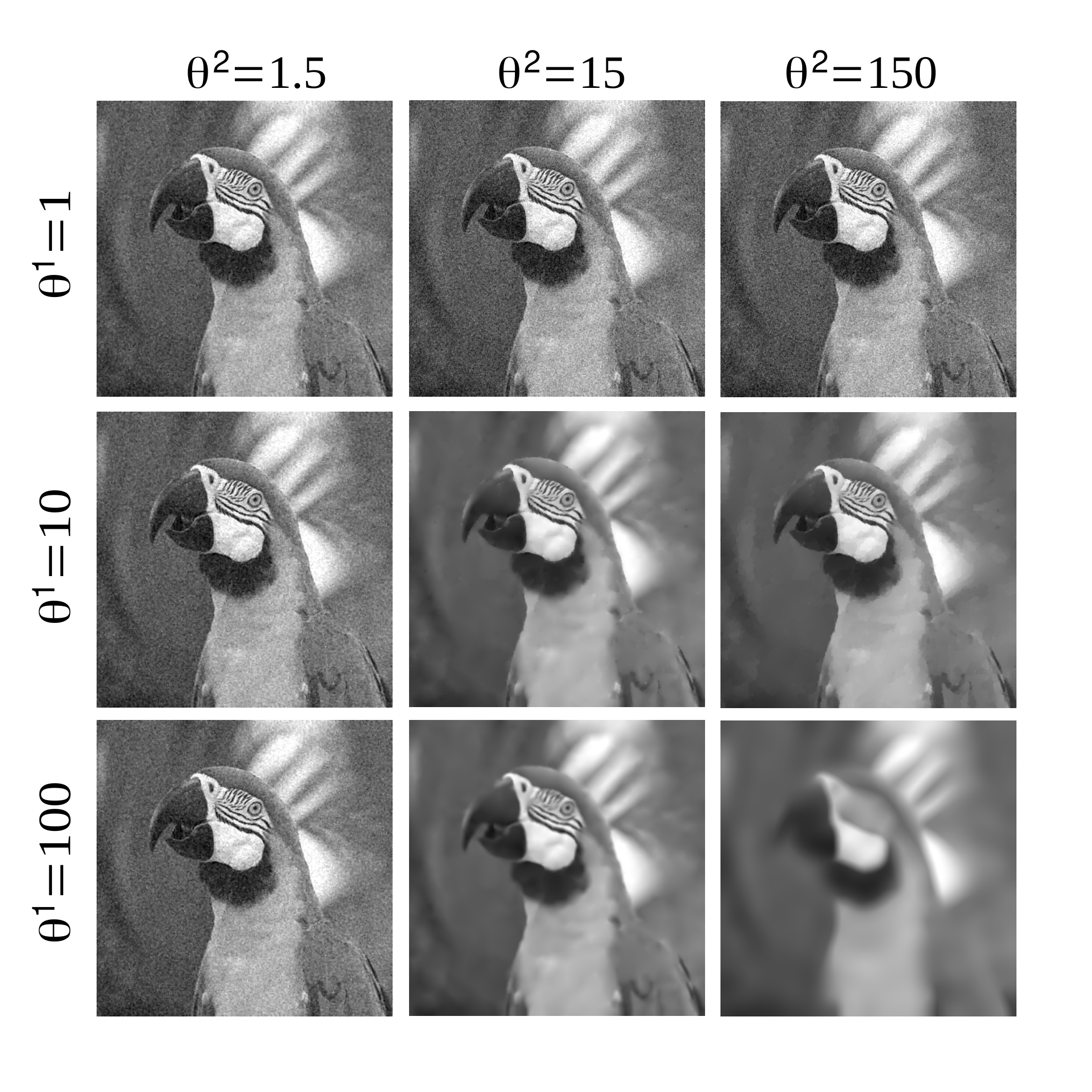}}			
	\end{minipage}
	\begin{minipage}[t]{0.49\textwidth}	\vspace{5pt}	
		\centerline{\includegraphics[width= 1.1\textwidth,trim={0cm 0cm 0cm 0cm},clip]{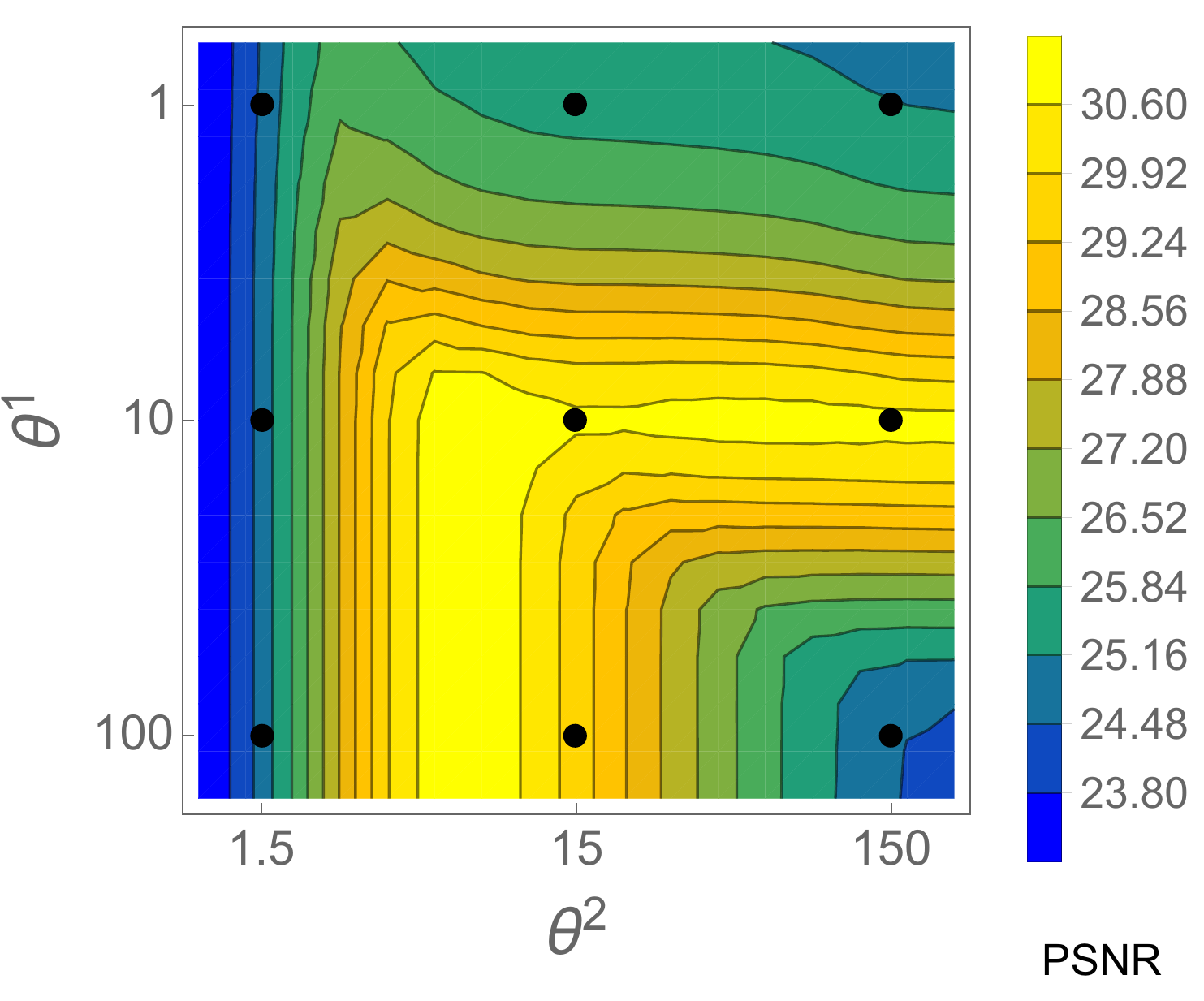}}								
	\end{minipage}
	\caption{Denoising with \tgvname \ prior. \MAPname \ estimates for different values of $\theta^1$ and $\theta^2$ for \texttt{parrot} image with \SNRname $=5.6~\mathrm{dB}$ (left).  \PSNRname \ for different values of $\theta^1$ and $\theta^2$ (right). The 9 black points on the right plot show the location of the parameter combinations used to compute the \MAPname \ estimates on the left. }
	\label{fig:num-res-TGV-parrot}
\end{figure}

However, this prior is not in the exponential family because $\theta^1$ and $\theta^2$ play a role in the definition of the statistic $\mathrm{TGV}_{\theta^1,\theta^2}^{2} (x)$. Therefore, our methodology and theory do not directly apply. Also note that the additional regularisation $\varepsilon \|x\|_2^2$ with $\varepsilon > 0$ is necessary to guarantee that $p(x)$ is proper, which is potentially important in order to apply the proposed methodology with two Markov chains (otherwise the auxiliary chain targeting $p(x)$ would not be ergodic - two chains are required because \eqref{eq:Experiments-TGV-norm} is not separable and homogeneous). We use $\varepsilon = 10^{-10}$.

In order to apply the proposed methodology to the estimation of $\theta^1$ and $\theta^2$ we use an approximation of the gradient  $\nabla_\theta \log p(x|\theta^1,\theta^2)$. More precisely, we express $p(x)$ as follows for any $x \in \rset^{d}$ and $\theta^1, \theta^2 >0$
$$
p(x|\theta^1,\theta^2) = \frac{1}{\mathrm{Z}(\theta^1,\theta^2)} \exp\parentheseDeux{-\theta^1 g_1(x,\theta^1,\theta^2) - \theta^2 g_2(x,\theta^1,\theta^2) - \varepsilon \|x\|_2^2} \eqsp ,
$$
with 
$$
g_1(x,\theta^1,\theta^2) =\norm{r(x,\theta^1,\theta^2)}_{1,2} \eqsp ,
$$ 
$$
g_2(x,\theta^1,\theta^2) = \norm{J(\Delta x-r(x,\theta^1,\theta^2))}_{1,\Frob.} \eqsp ,
$$
$$
r(x,\theta^1,\theta^2)=\underset{s \in \mathbb{R}^{2\dim}}{\mathrm{argmin}} \defEnsLigne{ \theta^1 \norm{s}_{1,2} + \theta^2 \norm{J(\Delta x-s)}_{1,\Frob.} } \eqsp ,
$$
and approximate the partial derivatives $\frac{\partial}{\partial \theta^1} \log p(x|\theta^1,\theta^2)$ and $\frac{\partial}{\partial \theta^2} \log p(x|\theta^1,\theta^2)$ by
$$
\frac{\partial}{\partial \theta^1} \log p(x|\theta^1,\theta^2) \approx \textrm{E}_{x|\theta^1,\theta^2}[g_1(x,\theta^1,\theta^2)] - g_1(x,\theta^1,\theta^2) \, ,
$$
$$
\frac{\partial}{\partial \theta^2} \log p(x|\theta^1,\theta^2) \approx \textrm{E}_{x|\theta^1,\theta^2}[g_2(x,\theta^1,\theta^2)] - g_2(x,\theta^1,\theta^2)\, .
$$
This approximation of the gradient, which arises from omitting the terms 
$$
\textrm{E}_{x|\theta^1,\theta^2}\left[\theta^1 \frac{\partial}{\partial \theta^1} g_1(x,\theta^1,\theta^2) + \theta^2 \frac{\partial}{\partial \theta^1} g_2(x,\theta^1,\theta^2)\right] -\theta^1 \frac{\partial}{\partial \theta^1} g_1(x,\theta^1,\theta^2) - \theta^2 \frac{\partial}{\partial \theta^1} g_2(x,\theta^1,\theta^2)
$$
and
$$
\textrm{E}_{x|\theta^1,\theta^2}\left[\theta^1 \frac{\partial}{\partial \theta^2} g_1(x,\theta^1,\theta^2) + \theta^2 \frac{\partial}{\partial \theta^2} g_2(x,\theta^1,\theta^2)\right] - \theta^1 \frac{\partial}{\partial \theta^2} g_1(x,\theta^1,\theta^2) - \theta^2 \frac{\partial}{\partial \theta^2} g_2(x,\theta^1,\theta^2)
$$
in the calculation of the partial derivatives $\frac{\partial}{\partial \theta^1} \log p(x|\theta^1,\theta^2)$ and $\frac{\partial}{\partial \theta^2} \log p(x|\theta^1,\theta^2)$, introduces an additional bias in the stochastic gradients driving \Cref{algo:MCMC_double_chain}\footnote{A rigorous analysis of this bias should also consider the points where $\mathrm{TGV}_{\theta^1,\theta^2}^{2} (x)$ is not differentiable w.r.t. $\theta^1$ and $\theta^2$. This can be achieved by using similar techniques to \cite{de2019efficient}.}. However, the numerical experiments reported below suggest that the algorithm is robust to this additional bias, in the sense that we empirically observe good convergence to useful estimates of $\theta^1$ and $\theta^2$.

In our experiments, we implement \Cref{algo:MCMC_double_chain} with this approximate gradient and follow the recommendations provided in \Cref{append:ssec:setting-algo-parameters} to set the algorithm parameters; we perform $25$ warm-up iterations and set  $\btheta^1_0=\btheta^2_{0}=10$, $X_0^0=\bX_0^0=y$,  for any $n \in \nsets$, $m_n=1$ and $\delta_n = 20 \times  n^{- 0.8}/\dim$, and we set  $\lambda = \min\left( 5 \Ly^{-1}, \lambda_{\max}\right)$ with $ \lambda_{\max}=2 $ and $\Ly=(0.95/\sigma)^2$. To stop the algorithm we consider three different cases: we stop the algorithm i) after $N=2000$ fixed iterations ii) when the relative change in $ \thetaavg_{N} $ is  $\normLigne{\thetaavg_{N+1} - \thetaavg_{N}}_{\infty} \leq 10^{-4}$ and iii) $\normLigne{\thetaavg_{N+1} - \thetaavg_{N}}_{\infty} \leq 10^{-3}$. Again, we compute $\thetaavg_N$ using  \eqref{eq:thetaavg}, setting $(\omega_n)_{n \in \nset}$ to have  $N_0 = 20$ burn-in iterations.

We also considered a thinning of $6$ iterations in the chain associated with the prior as its samples were roughly $6$ times more correlated than those coming from the chain targeting the posterior (i.e., we discard $5$ every $6$ samples as explained in \Cref{append:ssec:tips-for-2-chains}). To compute the $\mathrm{TGV}_{\theta^1,\theta^2}^{2}$ norm and proximal operator, we use the iterative primal-dual algorithm \cite{codeTGV}.

Applying \Cref{algo:MCMC_double_chain} to the entire image is too computationally expensive because of the complexity associated with evaluating the proximal operator of the TGV regulariser.  Therefore, in this experiment we estimate $\thetaEB$ from a representative patch of size $255 \times 255$ pixels, and then use the estimated $\theta^1$ and $\theta^2$ values to compute the MAP estimate of the entire image\footnote{For homogeneous regularisers,  $\theta$ is asymptotically independent of the dimension of $\x$ when $\dim$ is large, suggesting that it is possible to estimate its value from a representative image patch. Our empirical results suggest that this might hold for other models as well.}. We consider the same ten test images used in \Cref{ssec:experiments.nat-img-deblur} and we set the noise variance $\sigma^2$, such that the signal-to-noise-ratio (\SNRname) is ${8 ~\text{$\mathrm{dB}$}}$, ${12 ~\text{$\mathrm{dB}$}}$, or ${20 ~\text{$\mathrm{dB}$}}$. For each image and noise level, we first obtain an estimate for $\theta^1$ and $\theta^2$ and then use them to compute the MAP estimator $\hat{\x}_{\MAP}$ (given by \eqref{EQ: mapEstim}) using the same solver \cite{codeTGV} we use for the proximal operator. We measure estimation performance by computing the peak-signal-to-noise-ratio (PSNR) given by
$\textrm{PSNR}(x,\hat{\x}_{\MAP}) = - 10 \log_{10} {\|x - \hat{\x}_{\MAP}\|_2^2}/{d}$. All the \PSNRname \ plots shown in \Cref{fig:num-res-TGV-psnr-boat}, \Cref{fig:num-res-TGV-psnr-lake} and \Cref{fig:num-res-TGV-thetaInitSurf} were computed with the entire image.

\Cref{table_tgv} below summarises the average PSNR values and average computing times obtained for each SNR value for the three different stopping criteria. We observe that the proposed empirical Bayesian method achieves very good results for all SNR values and is very close to the oracle performance. Crucially, the stopping criteria has a strong impact on the computing times but not on the resulting PSNR values. Therefore, although convergence can take close to one hour with a strict convergence criterion, good results can be obtained in the order of a minute by using a weaker convergence criterion. 
\begin{table}[!hb]	
	\label{table_tgv}
	\centering
	\renewcommand{\arraystretch}{1.3}
	\begin{tabular}{rllllll}\hline
		Method & \multicolumn{2}{c}{\SNRname=8 $\mathrm{dB}$}                             & \multicolumn{2}{c}{\SNRname=12 $\mathrm{dB}$}                               & \multicolumn{2}{c}{\SNRname=20 $\mathrm{dB}$}                               \\
		&P\SNRname &Time  & PSNR   &Time  & PSNR &Time (min)\\
		\hline
		$ \thetaMSE$ (Oracle)               & 27.80 $\pm$ 2.35 &        & 30.21 $\pm$ 2.12 &       & 35.60 $\pm$ 1.77 &             \\
		\tiny{2000 iter} \small E.B.        & 27.11 $\pm$ 2.81 & 131.10 & 29.69 $\pm$ 2.33 & 96.41 & 35.48 $\pm$ 1.81  & 95.06       \\	
		\tiny{$tol\, 10^{-4}$} \small E.B.  & 27.09 $\pm$ 2.84 & 24.61  & 29.72 $\pm$ 2.33 & 23.27 & 35.47 $\pm$ 1.81  & 44.70         \\	
		\tiny{$tol\, 10^{-3}$} \small E.B.  & 27.00 $\pm$ 2.96 & 3.04   & 29.50 $\pm$ 2.71 & 2.18  & 35.57 $\pm$ 1.79  & 5.03         \\					
		\hline   		                
	\end{tabular}
	\caption{\small Denoising with \tgvname \ prior. Average mean squared error $ \pm $ standard deviation obtained for ten different images.  We show results for different stopping criteria, either with a fixed number of iterations or with a maximum tolerance for the relative change in the mean $\theta^{1}$ and $\theta^{2}$ estimates.}\normalsize
\end{table}

For illustration, \Cref{fig:num-res-TGV-xMAP} depicts the original image, the noisy observation and the recovered MAP estimates for the \texttt{boat} and \texttt{lake} test images with \SNRname \  $={8 ~\text{$\mathrm{dB}$}}$. 

\newcommand{\snrGTV}{8}
\def\widthfig{0.99\linewidth}
\def\widthminipage{0.30\linewidth}	
\begin{figure}[!h]
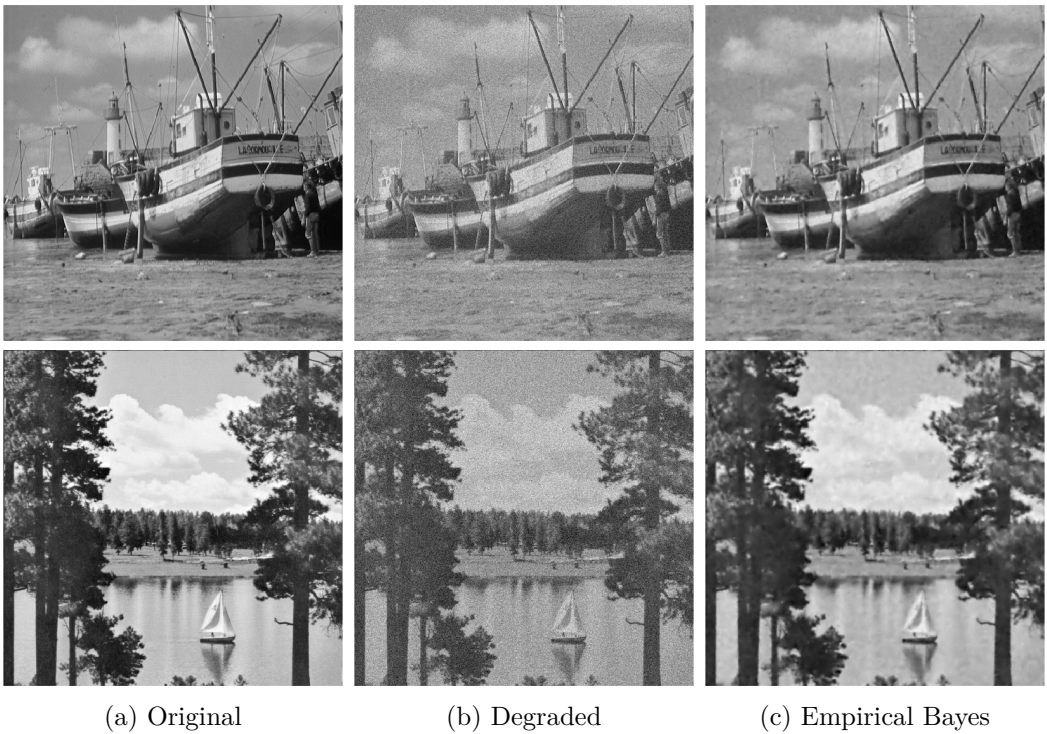
	
	\newcommand\incBoatGTV[1]{\includegraphics[width=\widthfig,trim={0cm 0cm 0cm 0cm},clip]{#1}\vspace{0.1cm}}
	\newcommand\incLake[1]{\includegraphics[width=\widthfig,trim={0cm 0cm 0cm 0cm},clip]{#1}\vspace{0.1cm}}
	\centering
	\begin{minipage}[t]{\widthminipage}
		\centering
		\centerline{\incBoatGTV{img/num-res-gtv/snr\snrGTV/boat_orig}}
		\centerline{\incLake{img/num-res-gtv/snr\snrGTV/lake_orig}}
		\centerline{(a) Original}
	\end{minipage}
	\begin{minipage}[t]{\widthminipage}
		\centering		
		\centerline{\incBoatGTV{img/num-res-gtv/snr\snrGTV/boat_degraded}}
		\centerline{\incLake{img/num-res-gtv/snr\snrGTV/lake_degraded}}		  
		\centerline{(b) Degraded}
	\end{minipage}
	\begin{minipage}[t]{\widthminipage}
		\centering
		\centerline{\incBoatGTV{img/num-res-gtv/snr\snrGTV/boat_estim}}
		\centerline{\incLake{img/num-res-gtv/snr\snrGTV/lake_estim}}
		\centerline{(c) Empirical Bayes}
	\end{minipage}	
	\caption{	\small		
		Denoising with \tgvname \ prior for \texttt{boat} and \texttt{lake} test images: (a) True image, (b) noisy observation $y$ (\SNRname=\snrGTV~$\mathrm{dB}$), (c) \MAPname \ estimators obtained with empirical Bayes.\normalsize}
	\label{fig:num-res-TGV-xMAP}
\end{figure}
More interestingly, \Cref{fig:num-res-TGV-psnr-boat} and \Cref{fig:num-res-TGV-psnr-lake} show the landscape of the \PSNRname \ as a function of $\theta^1$ and $\theta^2$ for the two test images, with the obtained solutions highlighted as a blue dot. Observe that the estimated solutions are extremely close to the optimal ones, which is remarkable given the difficulty of the problem and the fact that solutions are derived directly from statistical inference principles, without any form of ground truth.
\def\widthfig{0.99\linewidth}
\def\widthminipage{0.49\linewidth}	
\begin{figure*}[!h]	
	\centering	
	\begin{minipage}[t]{\widthminipage}
		\centerline{	\includegraphics[width=\textwidth]{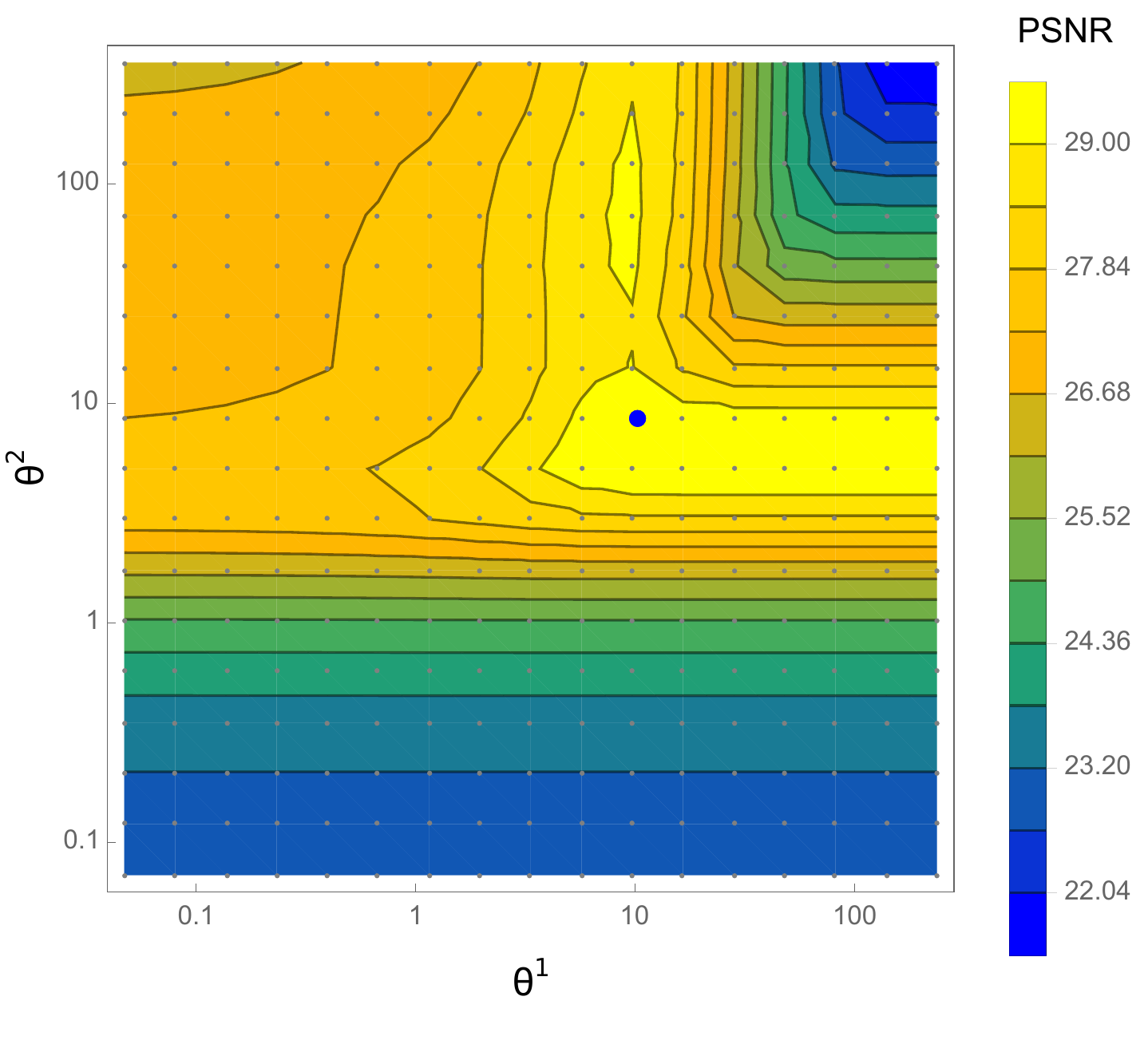}}
	\end{minipage}\hfill
	\begin{minipage}[t]{\widthminipage}
		\centerline{	\includegraphics[width=\textwidth]{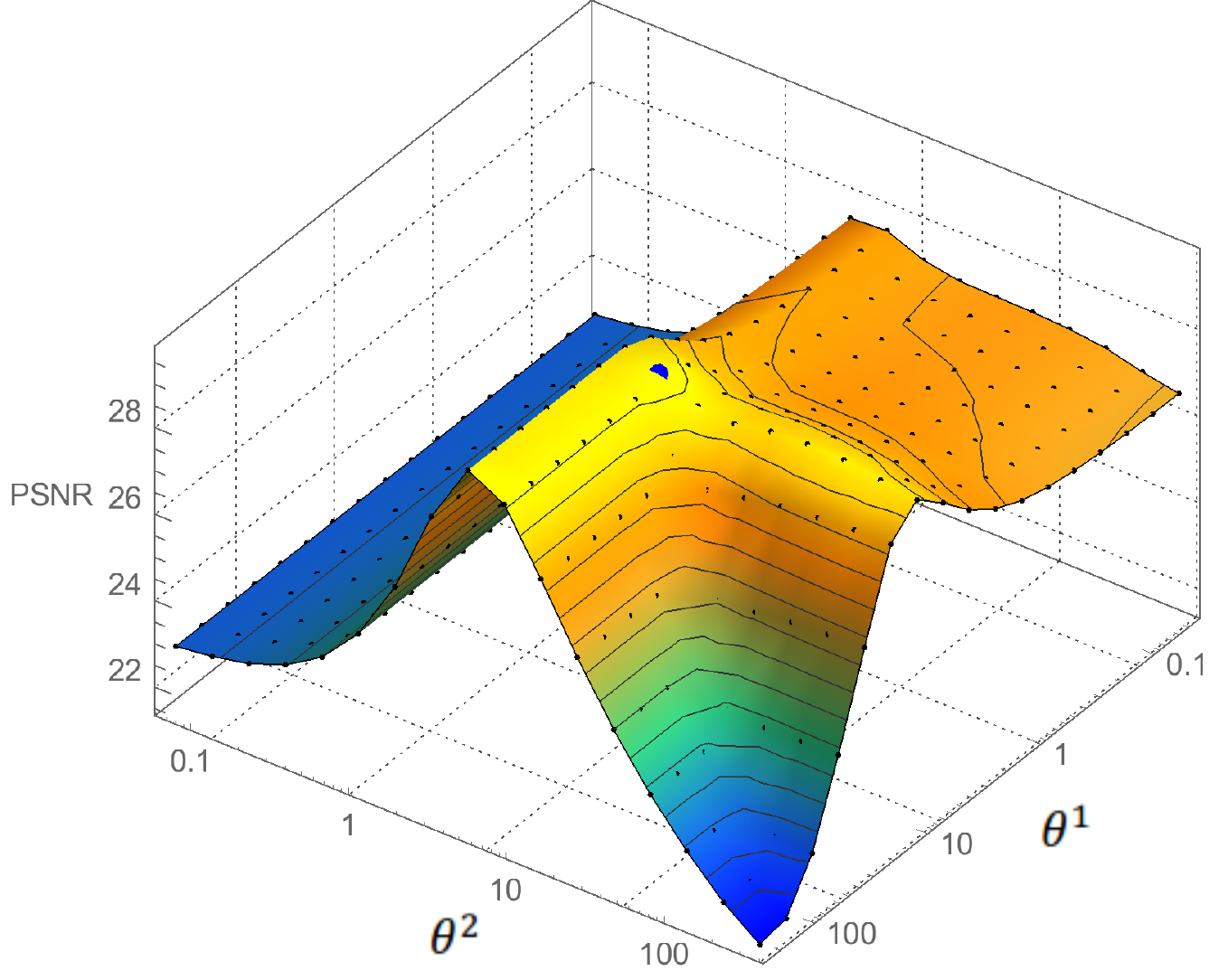}}	
	\end{minipage}
	\caption{	\small		
		Denoising with \tgvname \ prior on  \texttt{boat} image (\SNRname=\snrGTV~$\mathrm{dB}$). \PSNRname \ for different values of $\theta^1$ and $\theta^2$. Blue marker shows the location of $\thetaEB$ estimated with empirical Bayes using 2000 iterations.\normalsize}
	\label{fig:num-res-TGV-psnr-boat}
\end{figure*}

\begin{figure*}[!h]\centering		
	\begin{minipage}[t]{\widthminipage}
		\centerline{	\includegraphics[width=\textwidth]{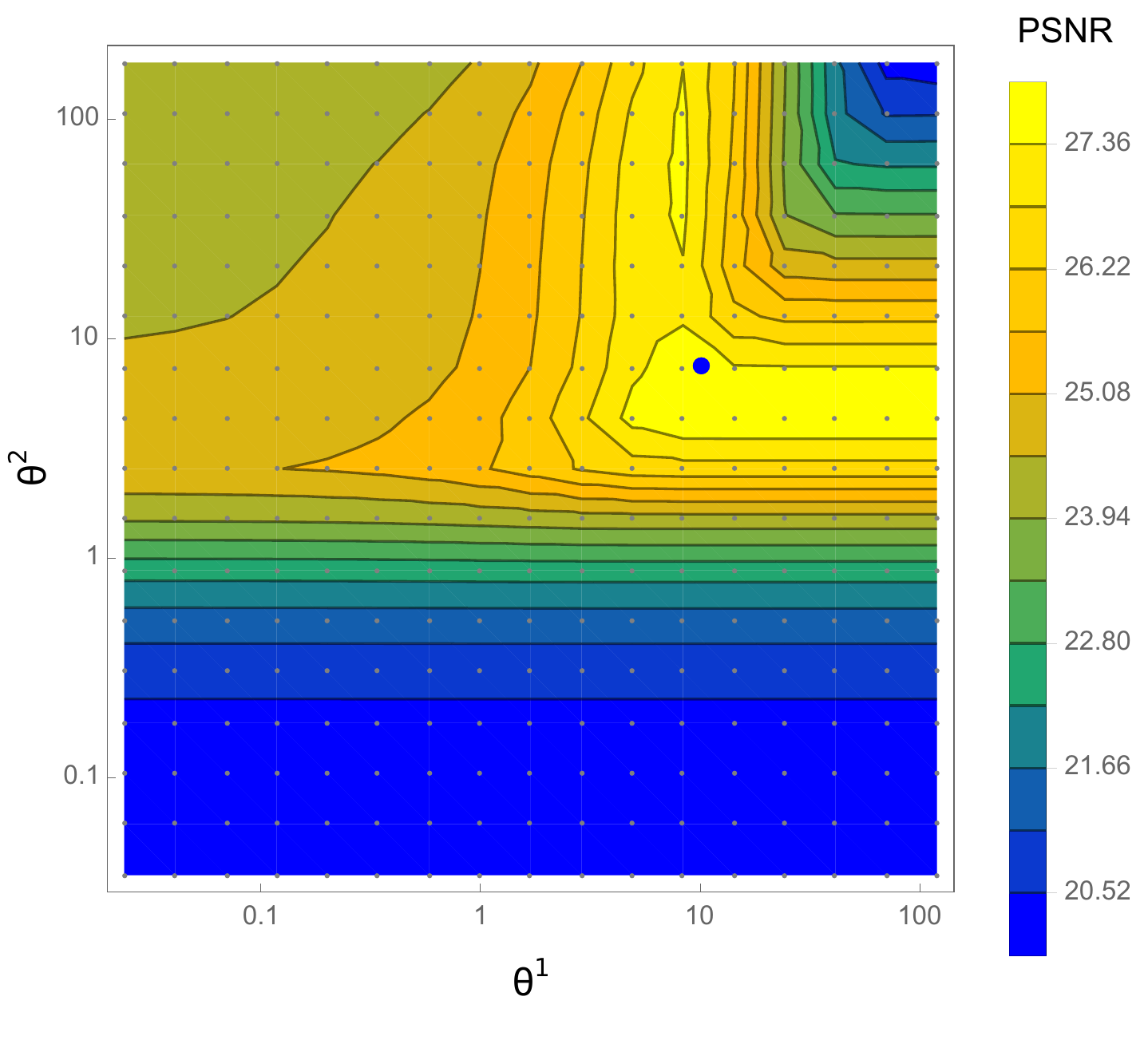}}
	\end{minipage}\hfill
	\begin{minipage}[t]{\widthminipage}
		\centerline{	\includegraphics[width=\textwidth]{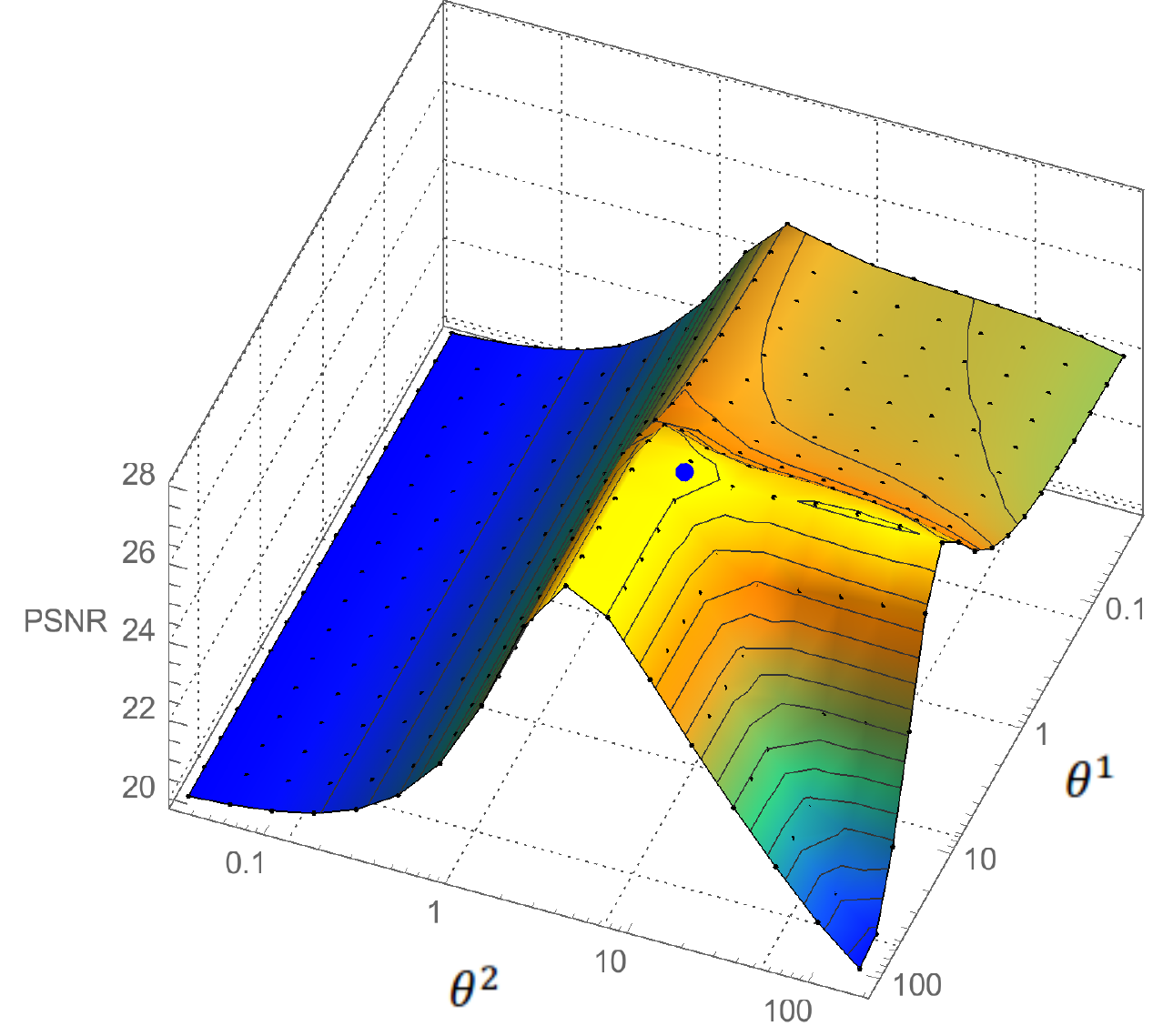}}	
	\end{minipage}
	\caption{	\small		
		Denoising with \tgvname \  prior on \texttt{lake} image (\SNRname=\snrGTV~$\mathrm{dB}$). \PSNRname \ for different values of $\theta^1$ and $\theta^2$. Blue marker shows the location of $\thetaEB$ estimated with empirical Bayes using 2000 iterations.\normalsize}
	\label{fig:num-res-TGV-psnr-lake}
\end{figure*}

Following on from this, \Cref{fig:num-res-TGV-theta} and \Cref{fig:num-res-TGV-relError} show, respectively, the evolution of the iterates and the relative change in the estimated values of $\theta^1$ and $\theta^2$, for the \texttt{lake} test image, and for \SNRname~$={8 ~\text{$\mathrm{dB}$}}$, \SNRname~$={12 ~\text{$\mathrm{dB}$}}$, and \SNRname~$={20 ~\text{$\mathrm{dB}$}}$. Observe that the algorithm converges quickly and can deliver a useful solution in approximately $50$ iterations if the weaker convergence criterion is used, or in approximately $500$ iterations if one uses a stricter convergence criterion.
\begin{figure}[h!]
	\centering
	\begin{minipage}[t]{0.3\textwidth}
		\centering
		\centerline{		\includegraphics[width=\textwidth]{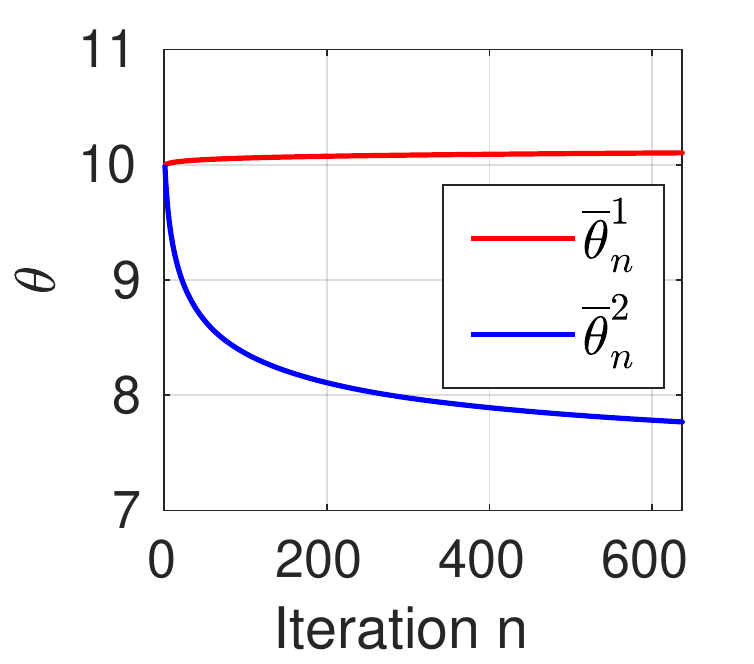}}
		\centerline{(a) \SNRname \ 8$\mathrm{dB}$}
	\end{minipage}
	\begin{minipage}[t]{0.3\textwidth}
		\centering
		\centerline{		\includegraphics[width=\textwidth]{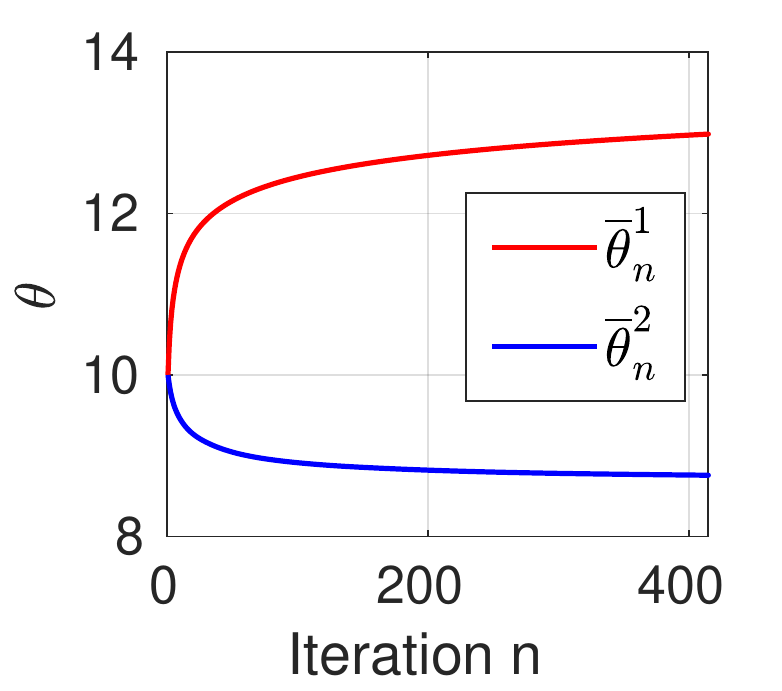}}
		\centerline{(b) \SNRname \ 12$\mathrm{dB}$}
	\end{minipage}
	\begin{minipage}[t]{0.3\textwidth}
		\centering
		\centerline{		\includegraphics[width=\textwidth]{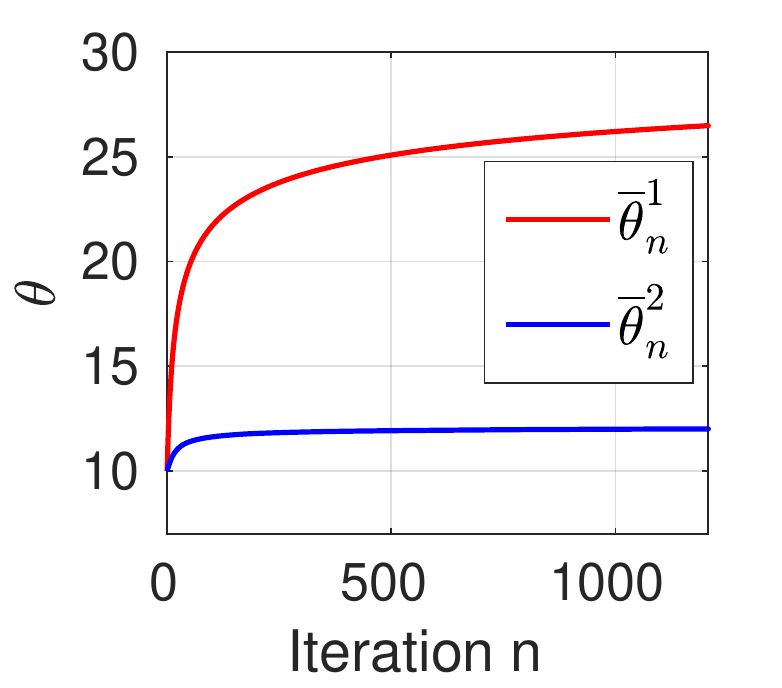}}
		\centerline{(c) \SNRname \ 20$\mathrm{dB}$}
	\end{minipage}
	\label{fig:num-res-TGV-theta}
	\caption{\small 	Denoising with \tgvname \ prior. Evolution of the iterates $(\theta^1_n)_{n \in \nset} $ and $ (\theta^2_n)_{n \in \nset}$ for the \texttt{lake} test image for different  \SNRname \ values.}
\end{figure}
\def\widthfig{0.99\linewidth}
\def\widthminipage{0.33\linewidth}	
\begin{figure*}[!h]		
	\newcommand\incRelErrorGTV[1]{\includegraphics[width=\widthfig]{#1}\vspace{0.1cm}}	
	\centering
	\begin{minipage}[t]{\widthminipage}
		\centering	
		\centerline{\incRelErrorGTV{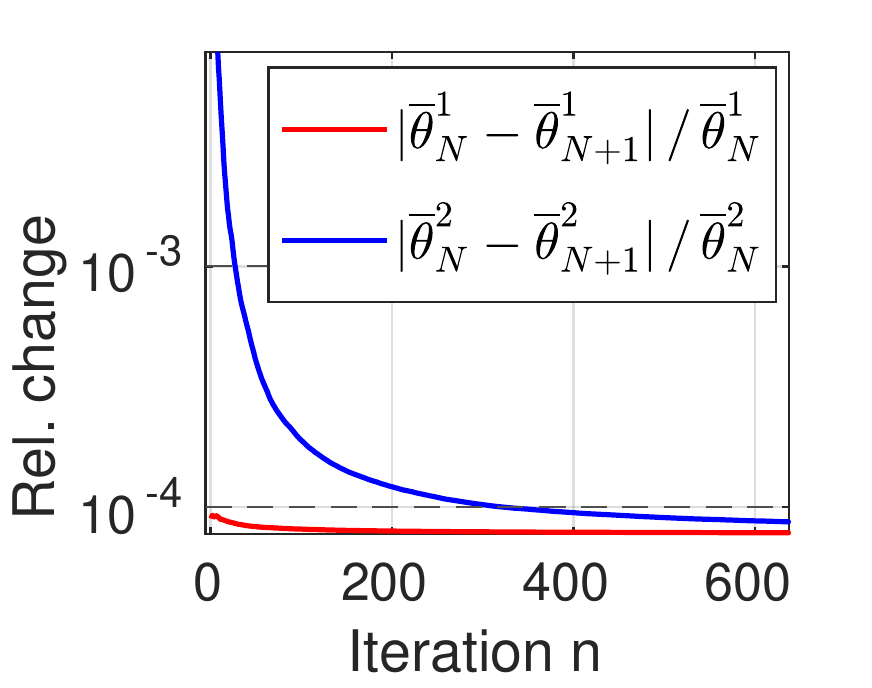}}
		\centerline{(a) \SNRname \ 8$\mathrm{dB}$}
	\end{minipage}\hfill
	\begin{minipage}[t]{\widthminipage}
		\centering			
		\centerline{\incRelErrorGTV{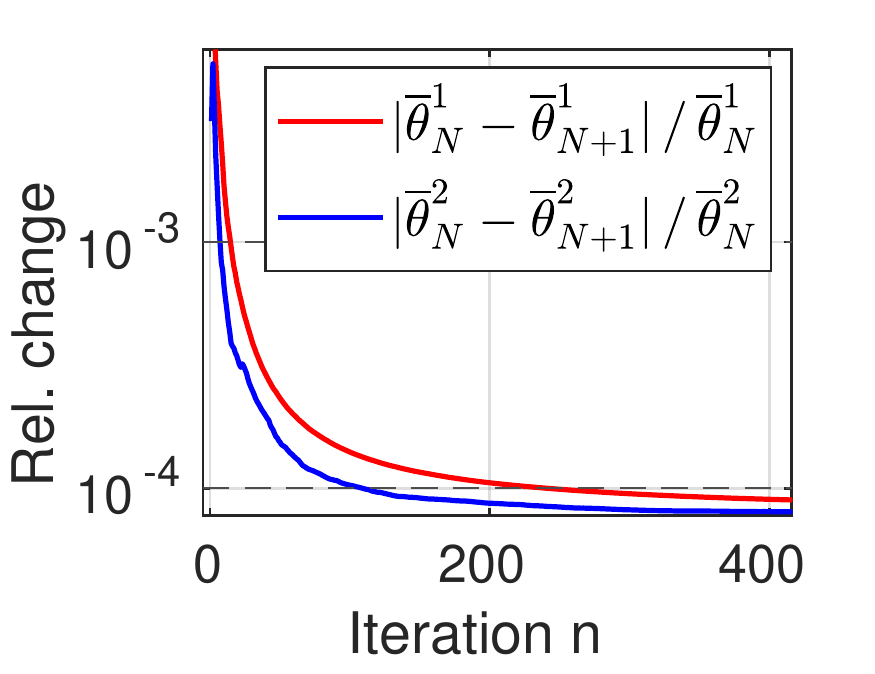}}	
		\centerline{(b) \SNRname \ 12$\mathrm{dB}$}
	\end{minipage}\hfill
	\begin{minipage}[t]{\widthminipage}
		\centering
		\centerline{\incRelErrorGTV{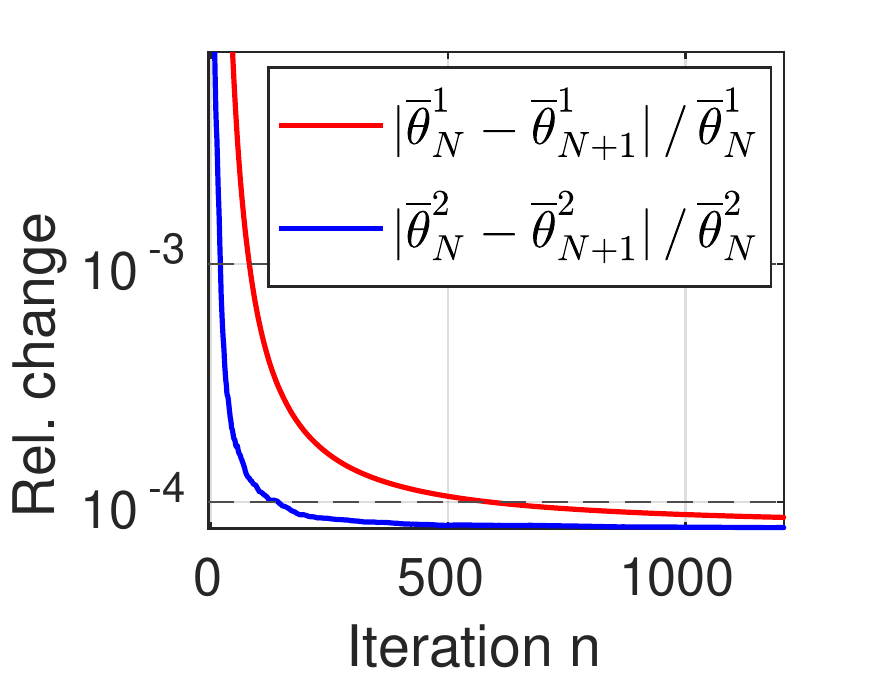}}		
		\centerline{(c) \SNRname \ 20$\mathrm{dB}$}
	\end{minipage}
	\caption{	\small		
		Denoising with \tgvname \ prior. Relative successive differences $|\thetaavg^i_N-\thetaavg^i_{N+1}|/\thetaavg^i_N$ with $i=1,2$ for the proposed method with the \texttt{lake} test image for different  \SNRname \ values.\normalsize}
	\label{fig:num-res-TGV-relError}
\end{figure*}

Lastly, \Cref{fig:num-res-TGV-thetaInitSurf} below explores the robustness to different initialisations by showing the evolution of the iterates on the landscape of PSNR values for the \texttt{flintstones} image with \SNRname~$={12 ~\text{$\mathrm{dB}$}}$. We consider three different initialisations, highlighted in colours red, green, and blue, and observe that in the three cases the algorithm quickly converges to values for the parameters $\theta^1$ and $\theta^2$ that are close-to-optimal in terms of the resulting PSNR. However, the algorithm is not fully robust to bad initialisation because of the non-convexity and the approximations involved. For example, initialising the algorithm in the corner of the PSNR landscape (e.g., $\theta^1_0=\theta^2_0=100$) does not lead to a satisfactory solution, indicating that a careful initialisation is required. Alternatively, one could also initialise the algorithm by performing a certain number of updates on $\theta^1$ with $\theta^2$ fixed to a small value - e.g. $\theta^2 = 1$ - to keep the model close to the conventional total variation regulariser, and then update both $\theta^1$ with $\theta^2$ until the convergence criterion is satisfied.

\begin{figure}[h!]
	\centering
	\includegraphics[width=\widthfig,trim={0cm 0.3cm 0cm 0cm},clip]{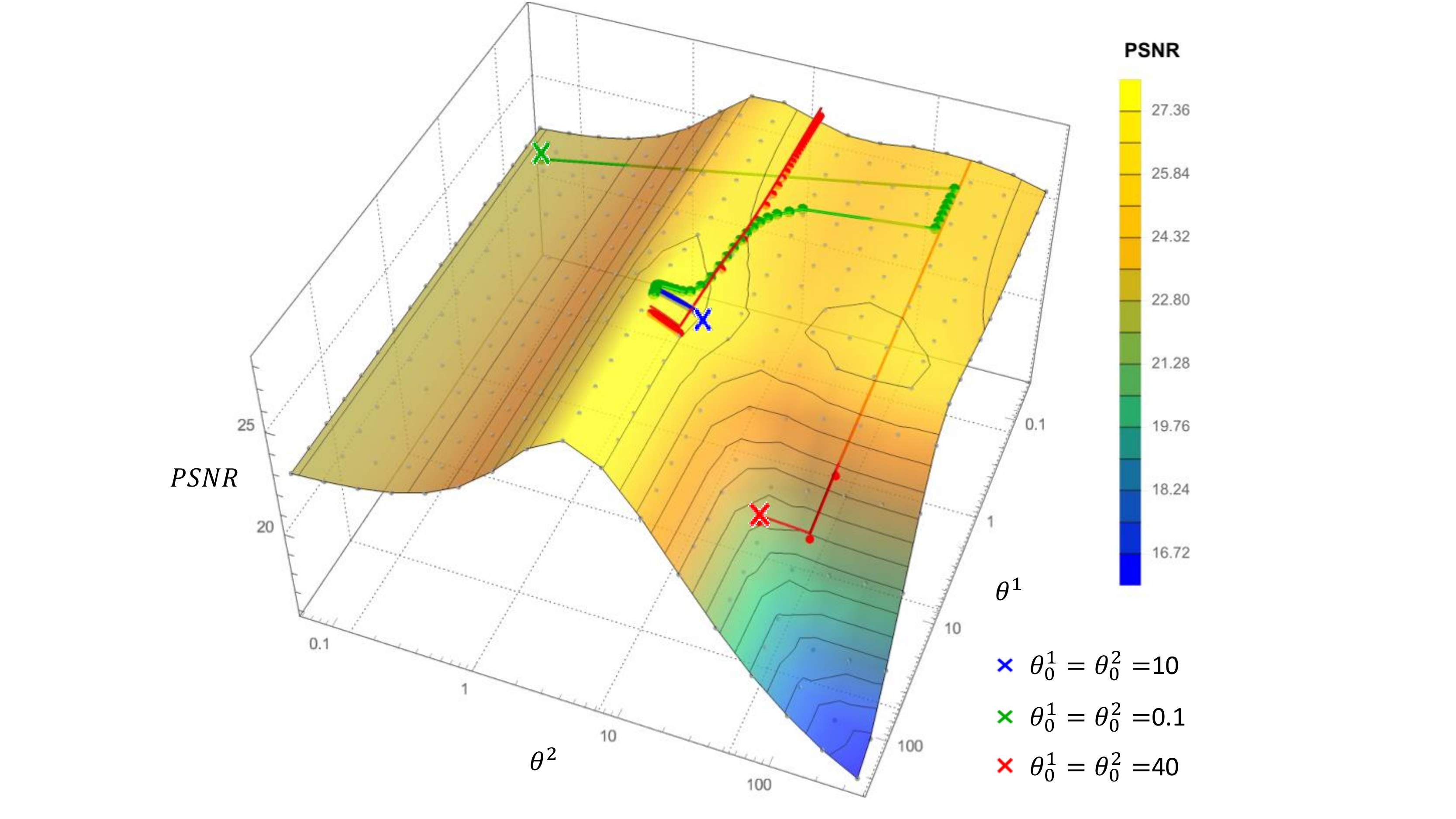}
	\label{fig:num-res-TGV-thetaInitSurf}
	\caption{Denoising with \tgvname \ prior on the \texttt{flintstones} image (\SNRname=12~$\mathrm{dB}$). Evolution of the iterates $ (\theta^1_n)_{n \in \nset} $ and $ (\theta^2_n)_{n \in \nset} $ for different initial values  $\theta^1_0$ and $\theta^2_0$. When initialising with $\theta^1_0=\theta^2_0=40$ (red) the algorithm converges to a different point with similar PSNR.}
\end{figure}

To conclude, we note that there are several other generalisations of the total variation regularisation (see \cite{bredies2010total}). We have chosen to perform our experiments with \eqref{eq:Experiments-TGV-norm} because of the availability of the efficient MATLAB implementation \cite{codeTGV}. However, we expect that \Cref{algo:MCMC_double_chain} will
also perform well for other generalisations of the total variation norm, particularly the second-order generalisation proposed in \cite{bredies2010total} that is very similar to \eqref{eq:Experiments-TGV-norm}.


\section{Conclusions}
\label{sec:conc}
This paper considered the automatic selection of regularisation parameters in imaging inverse problems, with a particular focus on problems that are convex w.r.t. the unknown image and possibly non-smooth, and which would be typically solved by maximum-a-posteriori estimation by using modern proximal optimisation techniques. We adopted an empirical Bayesian approach and proposed a computational method to efficiently and accurately estimate regularisation parameters by maximum marginal likelihood estimation. The considered marginal likelihood function is computationally intractable and we proposed to address this difficulty by using a stochastic proximal gradient optimisation algorithm that is driven by two proximal MCMC samplers, and which tightly combines the strengths of modern high-dimensional optimisation and Monte Carlo sampling techniques. Because the proposed method uses the same basic operators as proximal optimisation algorithms, namely gradient and proximal operators, it is straightforward to apply to problems that are currently solved by proximal optimisation. Moreover, we provided a detailed theoretical analysis of the proposed methodology, including easily verifiable conditions for convergence. In addition to being highly computational efficient and having strong theoretical underpinning, the proposed methodology is very general and can be used to simultaneously estimate multiple regularisation parameters, unlike some alternative approaches from the literature that can only handle a single or scalar parameter. 

We demonstrated the methodology with a range of imaging problems and models. We first considered image denoising and non-blind deblurring problems involving scalar regularisation parameters and showed that the method achieved close-to-optimal performance in terms of MSE and outperformed alternative approaches from the literature. We then successfully applied the method to two challenging problems involving bivariate regularisation parameters: a sparse hyperspectral unmixing problem with a total-variation plus sparsity prior, and a challenging denoising problem using a second-order total generalised variation regulariser. Again, the method delivered close-to-optimal results, as measured by estimation MSE.

Future work will focus on relaxing the convexity assumptions to provide theoretical convergence guarantees for non-convex problems, and on improving computational efficiency by using the recently proposed accelerated proximal Markov kernels \cite{Vargas2019Arxiv}. The application of the proposed methodology to challenging problems arising in medical and astronomical imaging is currently under investigation. Another important perspective for future work is to extend this methodology to semi-blind and blind imaging problems, as well as to problems involving space-varying regularisation parameters \cite{lanza2017space}.

\section{Acknowledgements} We are grateful to Dr. Charles Deledalle for
providing us with a SUGAR implementation for an ADMM solver available at
\url{https://github.com/deledalle/sugar/blob/master/solvers/admm.m}. AD acknowledges financial support from Polish National Science Center grant: NCN UMO-2018/31/B/ST1/00253. The work of MP is supported by UKRI/EPSRC under grant EP/T007346/1.
\bibliographystyle{plain}
\bibliography{refs}
	
 \appendix
 \section{Fisher's identity}
 \label{sec:fishers-identity}
 Fisher's identity is a standard result in the probability literature (e.g. see \cite[Proposition D.4]{douc2014nonlinear}). We reproduce its proof here for completeness.
 \begin{proposition}
 	\label{prop:fisher}  
 	For any $\theta \in \Theta \in \rset^{\dimtheta}$ and $\tx \in \rset^{\dim}$, let  $(x,y) \mapsto p(x,y | \theta)$  and $y \mapsto p(y|\tx)$ be positive probability density functions on $\rset^{\dim} \times \rset^{\dim_y}$ and $\rset^{\dimy}$. Assume that for any $x \in \rset^{\dim}$ and $\theta \in \interior(\Theta)$, $\theta \mapsto p(y, x|\theta)$ is differentiable. In addition, assume that for any $y \in \rset^{\dimy}$ and $\theta \in \interior(\Theta)$, there exist $\vareps >0$ and $\tg$ such that for any $\tilde{\theta} \in \cball{\theta}{\vareps}$ and $x \in \rset^{d}$, $\normLigne{\nabla_{\theta} p(y, x | \tilde{\theta})} \leq \tg(x)$ with $\int_{\rset^{\dim}} \tg(x) p(y|x) \rmd x < +\infty$. Then, for any $y \in \rset^{\dimy}$, $\theta \mapsto p(y|\theta)$ is differentiable over $\interior(\Theta)$ and we have for any $y \in \rset^{\dimy}$ and $\theta \in \interior(\Theta)$,
 	\begin{equation}
 	\nabla_{\theta} \log p(y | \theta) = \int_{\rset^{\dim}} p(x | y,  \theta) \nabla_{\theta} \log p(y, x | \theta)  \rmd x \eqsp .
 	\end{equation}
 \end{proposition}
 
 \begin{proof}
 	Let $y \in \rset^{\dimy}$. It is clear using the Leibniz integral rule that $\theta \mapsto p(y|\theta)$ is differentiable over $\interior(\Theta)$ and we have for any $\theta \in \interior(\Theta)$
 	\begin{align}
 		\nabla_{\theta} \log p(y|\theta) &= \left . \int_{\rset^{\dim}} p(y|x) \nabla_{\theta} p(y, x|\theta) \rmd x \middle/ p(y|\theta) \right. \\ &= \left . \int_{\rset^{\dim}} p(y, x |\theta) \nabla_{\theta} \log p(y, x|\theta) \rmd x \middle/ p(y|\theta) \right. = \int_{\rset^{\dim}} p(x | y , \theta) \nabla_{\theta} \log p(y, x|\theta) \rmd x \eqsp ,
 	\end{align}
 	which concludes the proof.
 \end{proof}
 
 \section{Practical implementation guidelines}
 \label{append:sec:practicalImplementGuidelines}
 
 In this section we provide some additional guidelines regarding the implementation and troubleshooting of the proposed methodology.  
 \subsection{Testing the MCMC sampler}
 \label{append:ssec:setup-mcmc-sampler}
 Before trying to adjust the value of $\theta \in \Theta$ with the algorithm, we strongly recommend starting by testing the MCMC sampler with a fixed value of $\theta$. A simple way to see if the Markov chain is working as expected, is to  plot the value of the log-probability of the samples.  
 
 As mentioned in \Cref{sec:probStatement}, there is a useful concentration phenomenon  studied in \cite[Theorem 1.2]{Bobkov2011} which implies that for high-dimensional log-concave densities $\pi$, a Markov chain targeting $\pi$ eventually start generating samples $X_n$ for which $\log \pi(X_n)$ is approximately constant (and close to the entropy). Therefore, if the MCMC sampling is successful the log-probability stabilises after some iterations and remains more or less constant. 
 \def\widthfig{0.99\linewidth}
 \def\widthminipage{0.49\linewidth}	
 \begin{figure*}[!htb]		
 	\begin{minipage}[t]{\widthminipage}
 		\centerline{	\includegraphics[width=\textwidth]{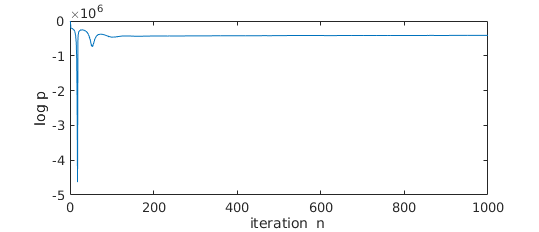}}
 		\centerline{(a) TV-deblurring}
 	\end{minipage}\hfill
 	\begin{minipage}[t]{\widthminipage}
 		\centerline{	\includegraphics[width=\textwidth]{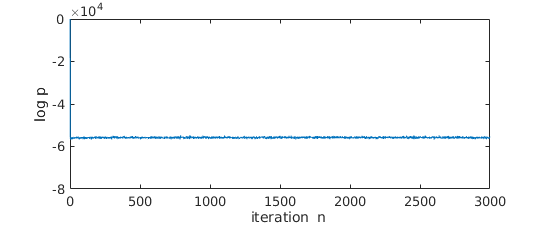}}	
 		\centerline{(b) TGV-denoising}
 	\end{minipage}
 	\caption{	\small Evolution of $(\log p(X_1^n|\y,\theta))_{n \in \nset}$ with $(X_1^n)_{n \in \nset}$ sampled using MYULA and targeting $p(\cdot|y, \theta)$. Results for (a) TV-deblurring with $\mathrm{SNR} = 40 \mathrm{dB}$ and (b) TGV-denoising with   $\mathrm{SNR} = 8 \mathrm{dB}$.\normalsize}
 	\label{fig:logPiCompare}
 \end{figure*}
 
 Conversely, if plots show that the chain is divergent or very unstable, then there might be a problem with the sampler. A common cause for divergence is setting a discretisation step-size that is too large. We would advise not to proceed with the estimation of $\theta$ until the sampler shows a stable behaviour similar to the ones shown on \Cref{fig:logPiCompare}.

 \subsection{Monitoring convergence in \Cref{algo:MCMC_single_chain}, \Cref{algo:MCMC_single_chain_separable} and \Cref{algo:MCMC_double_chain}}
 \label{append:ssec:checking-convergence}
 \paragraph{Lack of convergence due to bound saturation} If one observes that the iterate $\theta_n$ saturates the limits of the projection interval $\Theta$,  one should first verify that the Markov kernels are working properly (see recommendations in \ref{append:ssec:setup-mcmc-sampler}). If they are, then the problem might be that the solution lies outside $\Theta$. If $\btheta$ is multivariate and only some components are saturating the bounds, then check the scale and projection bounds for those specific components.  
 \paragraph{Verifying proper convergence}
 As the algorithm converges, the iterates $\theta_n$ get closer to a maximiser of $p(y|\theta)$ and the gradient estimates $\Delta_{m_n, \theta_n}$ vanish in expectation. Hence, the residual $\normLigne{\Delta_{m_n, \theta_n}}$ should become small (on average) as $n$ increases, i.e., $\g(X^n_k)$ will become close to $\dim/(\alpha\theta_n)$ in \Cref{algo:MCMC_single_chain}, or close to $ \g(\bX^n_k)$ in \Cref{algo:MCMC_double_chain}\footnote{For \Cref{algo:MCMC_single_chain_separable} use a component-wise comparison between $ \frac{\abs{\msa_i}}{\alpha_i \theta_n^i}$ and $\tg_i\left({X^n_{k}}_{[\msa_i]}\right)$ for every $i \in \{1, \dots, \dimtheta\}$}. It is therefore useful to plot the traces of $(\g(X^n_{k=k_0}))_{n \in \nset}$ together with $(\g(\bX^n_{k=k_0}))_{n \in \nset}$ or $(\dim/(\alpha\theta_n))_{n \in \nset}$ as appropriate, to check that the algorithm is converging. The trace can be plotted for a fixed value of $k=k_0$ as this is enough to monitor the convergence. This is illustrated for \Cref{algo:MCMC_double_chain} in \Cref{fig:coupling-g-2mcmc} below, where we observe how these terms become closer as the number of iterations increases.
 \begin{figure}[h!]
 	\centering
 	\includegraphics[width= 0.9 \textwidth]{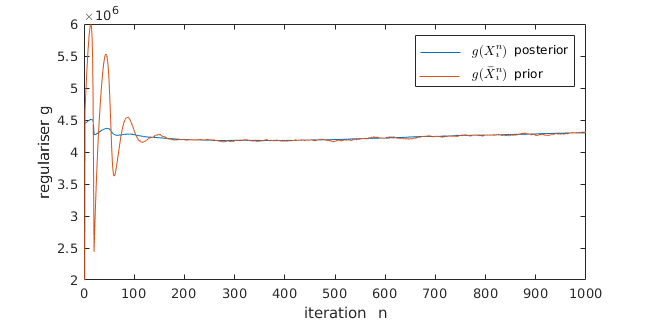}
 	\caption{\small Evolution of the iterates $(\g(X^n_{1}))_{n \in \nset}$ and $ (g(\bX^n_{1}))_{n \in \nset}$ for the proposed method in a deblurring experiment with a \tvname \  prior and $\mathrm{SNR}=40 \mathrm{dB}$. }
 	\label{fig:coupling-g-2mcmc}
 \end{figure}
 
 If $\normLigne{\Delta_{m_n, \theta_n}}$ does not vanish as $n$ increases this could indicate a problem with the choice of $\delta_n$ or that the two MCMC kernels have very different speed (See \Cref{append:ssec:tips-for-2-chains}).

 \subsection{Working with two MCMC chains in \Cref{algo:MCMC_double_chain}}
 \label{append:ssec:tips-for-2-chains}	
 Using two MCMC kernels simultaneously can be problematic if their
 convergence speed, or effective sample size per iteration, is very
 dissimilar as this will deteriorate the convergence properties of the
 SAPG
 algorithm. 
 
 This kind of imbalance can be detected by plotting the sample autocorrelation for each chain using $g$ as a summary statistic. If the autocorrelation plots decay at significantly different rates, it is necessary to reduce the correlation within the slower chain by either introducing some thinning (which essentially amounts to concatenating several iterations of the kernel to  improve its convergence speed) or by increasing the step-size $\gamma$ (see \Cref{append:ssec:convergence-speed}).
 
 \subsection{Working with multivariate $\theta$} 
 \label{append:ssec:vectorial-theta} 
 When $\btheta$ is multivariate each component of the solution might have a different order of magnitude. In this case, we recommend using different step-size scales for each component of $\btheta$. For example, we can compute  $\btheta_n = \Pi_{\Theta}\left[\theta_n + D~ \delta_n \Delta_{m_n, \theta_n}\right]$, where $D \in \mathbb{\rset}^{\dimtheta \times \dimtheta}$ is a diagonal matrix, and each element of the diagonal scales one component of $\theta$.
 It is also helpful to remember that one can run the algorithm with some components of $\btheta$ fixed. This allows isolating components and verifying convergence for subsets of $\btheta$.
 \subsection{Convergence speed} 
 \label{append:ssec:convergence-speed}
 The bottleneck in convergence speed is the correlation between the samples generated by the MCMC kernels. To increase the convergence speed, one has two main alternatives: a) to reduce the correlation between samples, or b) to reduce the computational cost of each iteration in order to afford more iterations.
 
 \paragraph{Reducing sample correlation}
 To reduce the correlation between samples, the step-size $\gamma$ must be as large as possible. If running the algorithm with two chains, and the kernel sampling from the prior distribution is the limiting factor, one can consider increasing the smoothing parameter $\lambda'$ of this particular kernel, in order to be able to increase the value of the discretisation step-size $\gamma'$. In more general cases where the limiting factor for $\gamma$ is $\Ly$ 
 there are a few strategies that might help overcome this difficulty. The first strategy is to use preconditioning 
 (see the hyperspectral unmixing experiment in \Cref{ssec:experiments.hyperunmix}) to reduce gradient anisotropy and improve the condition number of the problem.
 \paragraph{Speeding up each iteration}
 The most computationally heavy step in a MYULA iteration is usually the evaluation of the proximal operator. If the proximal operator is being approximated by an iterative solver, it is worth trying to improve efficiency by either using better solver, by warm starting iterations, or by using a weaker convergence criterion.
 
 \subsection{Estimation Bias} 
 \label{append:ssec:estimation-bias}
 If the algorithm converges but towards a poor value of $\theta \in \Theta$ it might be due to the bias in the MCMC kernels. As mentioned previously, there are many levels of approximation and the bias is mostly affected by the discretisation step $\gamma$ and the smoothing parameter $\lambda$. However, based on what we have observed in practice, the limiting factor tends to be $\lambda$. If there is a bias issue, we recommend trying to reduce $\lambda$ to obtain a better approximation of the target distribution, at the expense of some deterioration in convergence speed. When convergence is slowed down, special attention has to be paid in the case of the double MCMC chain algorithm. If the effective sample size of the two chains becomes too dissimilar, the algorithm might have difficulty converging. In this case, it is possible to do some thinning in the slower chain, as suggested in \Cref{append:ssec:tips-for-2-chains}.
 
 \section{Fair comparison of different methodologies}
 \label{append:method-comparison}
 Comparing different techniques for selecting the value of the regularisation parameter is highly non-trivial. Methods such as SUGAR are solver dependent and try to find the best value of $\theta$ for a given solver, with a given setup (number of iterations, parameters, etc.). Other algorithms such as the hierarchical one proposed in \cite{EUSIPCO}, depend on the solver but do not seek to optimise $\theta$ for that particular solver. The algorithm we propose does not depend directly on the solver. 
 
 When running statistics on our experiments we noticed an interesting phenomenon. For the deblurring experiments, we use the solver SALSA \cite{salsa2010fast}, which is an efficient implementation of the alternating direction method of multipliers (ADMM). When running the hierarchical Bayesian algorithm, we implement it with SALSA and set up the tolerance to $ 10^{-3} $ and 150 iterations which seemed sufficient to render very good results. However, when we build the MSE($\theta$) curves for \Cref{fig:num-res-synthL1-MSE} (by sampling many points and interpolating), we use SALSA with tolerance $ 10^{-5} $ and 1000 iterations as there were some values of $\theta$ for which SALSA did not converge well with tolerance $ 10^{-3} $. See in \Cref{fig:method-compare} that the position of the minimum MSE changes for the two different SALSA configurations. 
 When computing the average results for 10 images, the parameters obtained with the hierarchical method fell closer to the minimum of the red curve, and the ones obtained with the proposed empirical method fell closer to the minimum of the blue curve. Running the hierarchical method with tolerance $ 10^{-5}$ produced similar results but significantly increased the computing times.
 
 The criterion we opt for was to use SALSA with the strictest tolerance and highest number of iterations, because this configuration gives the overall best estimations.

 \subsection{Comparing with solver-dependent methods}
 As mentioned previously, algorithms like SUGAR  try to find the best value of $\theta$ for a given solver, with a given number of iterations, and specific parameters. This means that unless SUGAR is implemented with the exact same solver used to construct the $\mathrm{MSE}(\theta)$ curves as the ones in \Cref{fig:method-compare}, the values of $\theta$ computed with SUGAR might yield poor results according to the $\mathrm{MSE}(\theta)$ curve but good results with the specific solver used in SUGAR. For this reason, to achieve a fairer comparison, we compute an equivalent $\theta_{\EQ}$ in the following way. The SUGAR algorithm returns an estimated $\theta_{\SUG}$ and a corresponding $\mathrm{MSE}_{\SUG}$ obtained with that $\theta_{\SUG}$. Given an $\mathrm{MSE}(\theta)$ curve, we define the equivalent  $\theta_{\EQ}$ as 
 $\theta_{\EQ} = \underset{\theta \in \Theta}{\mathrm{argmin}}~ |\theta-\theta_{\SUG}|\quad s.t. \quad \mathrm{MSE}(\theta)=\mathrm{MS}E_{\SUG}\,$,
 which we plot in \Cref{fig:num-res-TV-MSEplot}. 
 \begin{figure}[h!]
 	\centering
 	\includegraphics[width= 0.65 \textwidth]{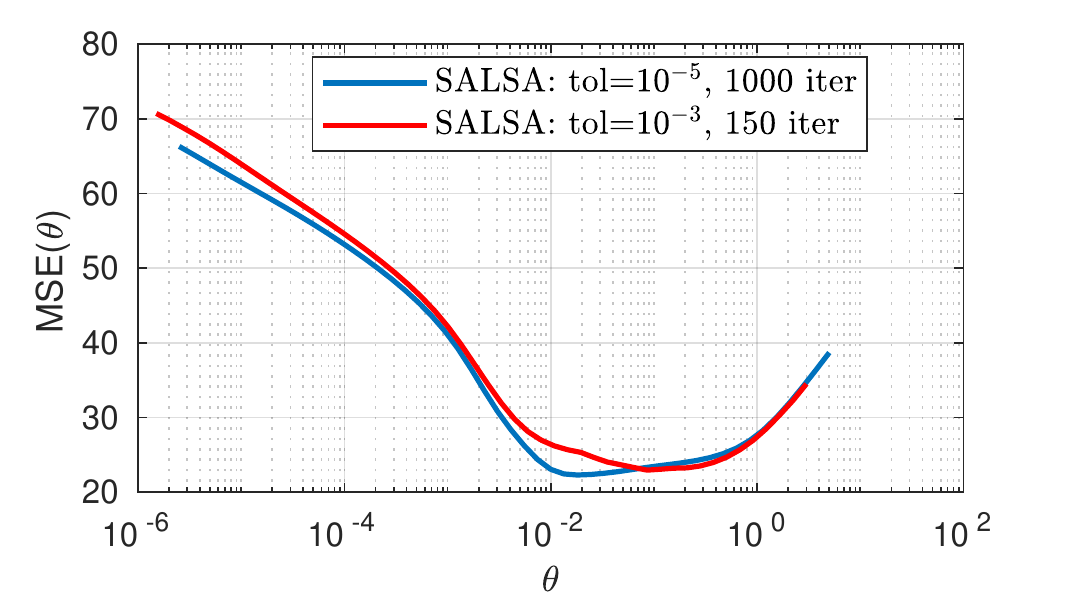}
 	\caption{\small \MSEname($\theta$) for wavelet synthesis-$\ell_1$ deconvolution for $\mathrm{SNR}=20 \mathrm{dB}$ with \texttt{boat} test image. The curves are computed with different tolerance and maximum iterations using $\mathrm{SALSA}$ solver.  }
 	\label{fig:method-compare}
 \end{figure}
  
\end{document}